\documentclass[10pt, letterpaper, twocolumn]{IEEEtran}

\usepackage[mathscr]{eucal}
\usepackage[cmex10]{amsmath}
\usepackage{epsfig,epsf,psfrag}
\usepackage{amssymb,amsmath,amsthm,amsfonts,latexsym}
\usepackage{amsmath,graphicx,bm,xcolor,url}
\usepackage[caption=false]{subfig} 
\usepackage{fixltx2e}
\usepackage{array}
\usepackage{verbatim}
\usepackage{bm}
\usepackage{algorithmic, cite}
\usepackage{algorithm}
\usepackage{verbatim}
\usepackage{textcomp}
\usepackage{mathrsfs}
\usepackage{epstopdf}

\catcode`~=11 \def\UrlSpecials{\do\~{\kern -.15em\lower .7ex\hbox{~}\kern .04em}} \catcode`~=13 

\allowdisplaybreaks[4]
 
\newcommand{\nn}{\nonumber}

\newcommand{\calA}{\mathcal{A}}
\newcommand{\calB}{\mathcal{B}}
\newcommand{\calC}{\mathcal{C}}

\newcommand{\calE}{\mathcal{E}}

\newcommand{\calI}{\mathcal{I}}
\newcommand{\calJ}{\mathcal{J}}

\newcommand{\calL}{\mathcal{L}}
\newcommand{\calM}{\mathcal{M}}
\newcommand{\calN}{\mathcal{N}}

\newcommand{\calP}{\mathcal{P}}

\newcommand{\calS}{\mathcal{S}}
\newcommand{\calT}{\mathcal{T}}
\newcommand{\calU}{\mathcal{U}}

\newcommand{\calX}{\mathcal{X}}
\newcommand{\calY}{\mathcal{Y}}
\newcommand{\calZ}{\mathcal{Z}}

\newcommand{\hatcalX}{\hat{\calX}}

\newcommand{\bA}{\mathbf{A}}

\newcommand{\bI}{\mathbf{I}}
\newcommand{\bj}{\mathbf{j}}
\newcommand{\bJ}{\mathbf{J}}

\newcommand{\bM}{\mathbf{M}}

\newcommand{\bR}{\mathbf{R}}

\newcommand{\bS}{\mathbf{S}}

\newcommand{\bU}{\mathbf{U}}
\newcommand{\bv}{\mathbf{v}}
\newcommand{\bV}{\mathbf{V}}

\newcommand{\bX}{\mathbf{X}}

\newcommand{\bY}{\mathbf{Y}}
\newcommand{\bz}{\mathbf{z}}
\newcommand{\bZ}{\mathbf{Z}}

\newcommand{\rmb}{\mathrm{b}}

\newcommand{\rmc}{\mathrm{c}}

\newcommand{\rmd}{\mathrm{d}}

\newcommand{\rme}{\mathrm{e}}

\newcommand{\rmg}{\mathrm{g}}

\newcommand{\rmp}{\mathrm{p}}
\newcommand{\rmP}{\mathrm{P}}

\newcommand{\rms}{\mathrm{s}}

\newcommand{\bbE}{\mathbb{E}}

\newcommand{\bbN}{\mathbb{N}}

\newcommand{\bbR}{\mathbb{R}}

\newcommand{\scC}{\mathscr{C}}

\newcommand{\scP}{\mathscr{P}}

\newcommand{\scR}{\mathscr{R}}
\newcommand{\scS}{\mathscr{S}}

\newcommand{\scV}{\mathscr{V}}

\DeclareMathAlphabet{\mathbsf}{OT1}{cmss}{bx}{n}
\DeclareMathAlphabet{\mathssf}{OT1}{cmss}{m}{sl}

\newcommand{\rvd}{\mathsf{d}}

\newcommand{\rvg}{\mathsf{g}}

\DeclareSymbolFont{bsfletters}{OT1}{cmss}{bx}{n}  
\DeclareSymbolFont{ssfletters}{OT1}{cmss}{m}{n}
\DeclareMathSymbol{\bsfGamma}{0}{bsfletters}{'000}
\DeclareMathSymbol{\ssfGamma}{0}{ssfletters}{'000}
\DeclareMathSymbol{\bsfDelta}{0}{bsfletters}{'001}
\DeclareMathSymbol{\ssfDelta}{0}{ssfletters}{'001}
\DeclareMathSymbol{\bsfTheta}{0}{bsfletters}{'002}
\DeclareMathSymbol{\ssfTheta}{0}{ssfletters}{'002}
\DeclareMathSymbol{\bsfLambda}{0}{bsfletters}{'003}
\DeclareMathSymbol{\ssfLambda}{0}{ssfletters}{'003}
\DeclareMathSymbol{\bsfXi}{0}{bsfletters}{'004}
\DeclareMathSymbol{\ssfXi}{0}{ssfletters}{'004}
\DeclareMathSymbol{\bsfPi}{0}{bsfletters}{'005}
\DeclareMathSymbol{\ssfPi}{0}{ssfletters}{'005}
\DeclareMathSymbol{\bsfSigma}{0}{bsfletters}{'006}
\DeclareMathSymbol{\ssfSigma}{0}{ssfletters}{'006}
\DeclareMathSymbol{\bsfUpsilon}{0}{bsfletters}{'007}
\DeclareMathSymbol{\ssfUpsilon}{0}{ssfletters}{'007}
\DeclareMathSymbol{\bsfPhi}{0}{bsfletters}{'010}
\DeclareMathSymbol{\ssfPhi}{0}{ssfletters}{'010}
\DeclareMathSymbol{\bsfPsi}{0}{bsfletters}{'011}
\DeclareMathSymbol{\ssfPsi}{0}{ssfletters}{'011}
\DeclareMathSymbol{\bsfOmega}{0}{bsfletters}{'012}
\DeclareMathSymbol{\ssfOmega}{0}{ssfletters}{'012}

\newcommand{\hatC}{\hat{C}}

\newcommand{\hatL}{\hat{L}}
\newcommand{\till}{\tilde{l}}

\newcommand{\hatm}{\hat{m}}
\newcommand{\hatM}{\hat{M}}

\newcommand{\hatR}{\hat{R}}

\newcommand{\hatU}{\hat{U}}

\newcommand{\hatx}{\hat{x}}
\newcommand{\hatX}{\hat{X}}
\newcommand{\tilx}{\tilde{x}}
\newcommand{\tilX}{\tilde{X}}

\newcommand{\tilY}{\tilde{Y}}

\newcommand{\tilZ}{\tilde{Z}}

\newcommand{\barx}{\bar{x}}

\newcommand{\barP}{\bar{P}}

\newcommand{\barW}{\bar{W}}

\newcommand{\veps}{\varepsilon}

\newcommand{\bmu}{\bm{\mu}}

\newcommand{\bSigma	}{\bm{\Sigma}}

\def\fndot{\, \cdot \,}

\DeclareMathOperator{\var}{\mathsf{Var}}

\DeclareMathOperator{\cov}{\mathsf{Cov}}

\newcommand{\bzero}{\mathbf{0}}
\newcommand{\bone}{\mathbf{1}}

\newtheorem{theorem}{Theorem} 
\newtheorem{lemma}[theorem]{Lemma}

\newtheorem{proposition}[theorem]{Proposition}
\newtheorem{corollary}[theorem]{Corollary}
\newtheorem{definition}{Definition}

\newtheorem{remark}{Remark}

\newcommand{\qednew}{\nobreak \ifvmode \relax \else
      \ifdim\lastskip<1.5em \hskip-\lastskip
      \hskip1.5em plus0em minus0.5em \fi \nobreak
      \vrule height0.75em width0.5em depth0.25em\fi}

\newcommand{\Pe}{\rmP_{\rme}}
\DeclareMathOperator*{\plimsup}{\mathfrak{p}-lim\, sup\,}
\DeclareMathOperator*{\pliminf}{\mathfrak{p}-lim\, inf\,}
\newcommand{\underI}{\underline{I}}
\newcommand{\overI}{\overline{I}}

\newcommand{\overH}{\overline{H}} 
\newcommand{\WAK}{\mathrm{WAK}}
\newcommand{\WZ}{\mathrm{WZ}}
\newcommand{\GP}{\mathrm{GP}}
\newcommand{\hcalX}{\hat{\calX}}

\newcommand{\sizeI}{\lvert\calI\rvert}

\newcommand{\sizeL}{\lvert\calL\rvert}
\newcommand{\sizeM}{\lvert\calM\rvert}

\newcommand{\calTsWZ}{\calT_{\rmd\mathrm{,st}}^{\WZ}}
 
\usepackage{color}

\newcommand{\bin}{\kappa}
\usepackage{cite}
\allowdisplaybreaks[1]

\title{Non-Asymptotic and Second-Order Achievability Bounds for   Coding With Side-Information\thanks{This paper was presented in part at the 2013 IEEE International Symposium on Information Theory.}}
  
\author{Shun~Watanabe~\IEEEmembership{Member,~IEEE},       
\thanks{The first author is with the Department
of Information Science and Intelligent Systems, 
University of Tokushima,
2-1, Minami-josanjima, Tokushima,
770-8506, Japan, and with the Institute for Systems Research, University of Maryland, 
College Park, MD 20742, USA, 
e-mail:shun-wata@is.tokushima-u.ac.jp.}
Shigeaki~Kuzuoka~\IEEEmembership{Member,~IEEE},       
\thanks{The second author is with the Department of Computer and Communication Sciences,
Wakayama University, Wakayama, 640-8510, Japan, e-mail:kuzuoka@ieee.org.}
and Vincent Y.~F.\ Tan~\IEEEmembership{Member,~IEEE}
\thanks{The third author is with the Department of Electrical and Computer Engineering and Department of Mathematics, National University of Singapore (NUS), e-mail:vtan@nus.edu.sg}

\thanks{Manuscript received ; revised }}

\begin{document}
\flushbottom
\maketitle

\begin{abstract} 
We present   novel  non-asymptotic or finite blocklength achievability bounds for three       side-information problems in network  information theory. These include (i) the Wyner-Ahlswede-K\"orner (WAK) problem of almost-lossless source coding with rate-limited side-information, (ii) the Wyner-Ziv (WZ) problem of lossy source coding with side-information at the decoder and (iii) the Gel'fand-Pinsker (GP) problem of channel coding with noncausal state information available at the encoder. The bounds are proved using ideas from channel simulation and channel resolvability. Our bounds for all three problems improve  on all previous non-asymptotic bounds on the error probability of the WAK, WZ and GP problems--in particular those derived by Verd\'u. Using our novel non-asymptotic bounds, we  recover the general formulas for the optimal rates of these side-information problems. Finally, we also present achievable second-order coding rates by applying the multidimensional Berry-Ess\'een theorem to our new non-asymptotic bounds. Numerical results show that the second-order coding rates obtained using our non-asymptotic achievability bounds  are superior to those obtained using existing finite blocklength bounds. 
\end{abstract}

\section{Introduction}
The study of {\em network   information theory}~\cite{elgamal} involves characterizing the optimal rate regions or capacity regions for problems involving compression and transmission from multiple sources to multiple destinations. Apart from a few special channels or source models, optimal rate regions and capacity regions for many network information theory problems are still not known. In this paper, we revisit three coding problems whose asymptotic  rate characterizations are well known. These include 
\begin{itemize}
\item The {\em Wyner-Ahlswede-K\"orner} (WAK) problem of almost-lossless source coding with rate-limited (aka coded) side-information  \cite{Wyner75, Ahl75},
\item The {\em Wyner-Ziv} (WZ) problem of lossy source coding with side-information at the decoder  \cite{wynerziv}, and
\item The {\em Gel'fand-Pinsker} (GP) problem of channel coding with noncausal state information   at the encoder \cite{GP80}. 
\end{itemize}%
 These problems fall under the class of coding problems with {\em side-information}. That is, a subset of terminals has access to either a correlated source or the state of the channel. In most cases, this knowledge helps to strictly improve the rates of compression or transmission over the case where there is no side-information.

While the study of asymptotic characterizations of network information theory problems has been of key interest and importance for the past $50$ years, it is important to analyze     non-asymptotic (or finite blocklength) limits of  various network information theory problems. This is because there may be  hard constraints on decoding complexity or delay in modern, heavily-networked systems. The paper derives new non-asymptotic bounds on the error probability for the WAK  and GP problems as well as the probability of excess distortion for the WZ problem.  Our bounds improve on all existing finite blocklength bounds for these problems such as those in~\cite{Ver12}.  In addition, we use these bounds to recover  known  general formulas~\cite{Han10,miyake, iwata02, Tan12b}  and we also derive achievable second-order coding rates~\cite{Hayashi08, Hayashi09} for these side-information problems.

Traditionally,   achievability proofs of the direct pats of these coding problems are common and  involve  a covering step, a packing step and   the use of the Markov lemma~\cite{Wyner75} (also known as conditional typicality lemma in El Gamal and Kim~\cite{elgamal}). As such to prove tighter bounds, it is  necessary to develop new proof techniques in place of these lemmas~\cite{elgamal} and their non-asymptotic versions~\cite{Ver12, Han10}. These new techniques are based on the notion of {\em channel resolvability}~\cite{HV93,Hayashi06, Han10}  and {\em channel simulation}  \cite{BSST02 , Win02,Cuff12}. We use the former  in the helper's code   construction. 

\begin{figure}
\centering
\setlength{\unitlength}{.028cm}
\begin{picture}(180, 100)
\put(0, 20){\vector(1, 0){40}}
\put(0, 80){\vector(1, 0){40}}
\put(80, 80){\vector(1, 0){40}}
\put(80, 20){\line(1, 0){60}}
\put(140, 20){\vector(0, 1){40}}
\put(160, 80){\vector(1, 0){40}}
\put(40, 0){\line(1, 0){40}}
\put(40, 40){\line(1, 0){40}}
\put(40, 0){\line(0,1){40}}
\put(80, 0){\line(0,1){40}}

\put(40, 60){\line(1, 0){40}}
\put(40, 100){\line(1, 0){40}}
\put(40, 60){\line(0,1){40}}
\put(80, 60){\line(0,1){40}}

\put(120, 60){\line(1, 0){40}}
\put(120, 100){\line(1, 0){40}}
\put(120, 60){\line(0,1){40}}
\put(160, 60){\line(0,1){40}}

\put(17, 25){\mbox{$Y$}}
\put(17, 85){\mbox{$X$}}

\put(97, 25){\mbox{$L$}}
\put(97, 85){\mbox{$M$}}

\put(180, 85){\mbox{$\hatX$}}
\put(168, 60){\mbox{$\Pr(\hatX\ne X)$}}

\put(57, 78){\mbox{$f$}}
\put(57, 18){\mbox{$g$}}

\put(137, 78){\mbox{$\psi$}}
\end{picture}
\caption{Illustration of the WAK problem}
\label{fig:wak}
\end{figure}
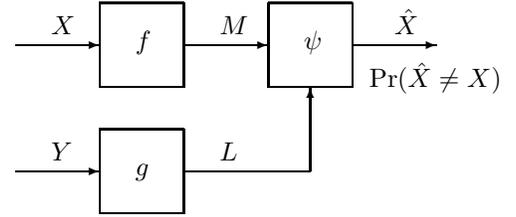

To illustrate our idea at a high level, let us use the WAK problem as a canonical example of all three problems  of interest.  Recall that in the classical WAK problem, there is an independent and identically distributed  (i.i.d.) joint source $P_{XY}^n(x^n,y^n)=\prod_{i=1}^n P_{XY}(x_i, y_i)$. The  main source $X^n\sim  P_X^n$   is to be reconstructed almost losslessly from rate-limited versions of both $X^n$ and  $Y^n$, where $Y^n$ is  a correlated random variable regarded as side-information. See Fig.~\ref{fig:wak}. The compression rates of $X^n$ and $Y^n$ are denoted as $R_1$ and $R_2$  respectively.  The  optimal rate region is the set of  rate pairs $(R_1, R_2)$ for which there exists a {\em reliable} code, that is one whose error probability can be made arbitrarily small with increasing blocklengths. WAK~\cite{Wyner75, Ahl75} showed that the optimal rate region is  
\begin{align}
 R_1 \ge H(X|U),\quad R_2\ge I(U;Y) \label{eqn:wak}
\end{align}
for some $P_{U|Y}$. For the direct part, the helper encoder compresses the side-information and transmits a description represented by $U^n$.  By the covering lemma~\cite{elgamal}, this results in the rate constraint $R_2\ge I(U;Y)$. The main encoder then uses    binning~\cite{cover75}  as in the achievability proof of the Slepian-Wolf theorem~\cite{sw73} to help the decoder recover $X$ given  the description  $U$. This results in the rate constraint $R_1\ge H(X|U)$.


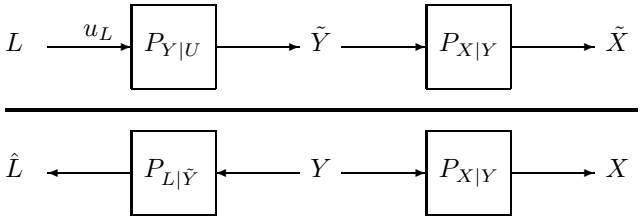
\begin{figure}
\centering
\setlength{\unitlength}{.028cm}
\begin{picture}(300, 100)
\put(60, 20){\vector(-1, 0){40}}
\put(20, 80){\vector(1, 0){40}}
\put(100, 80){\vector(1, 0){40}}
\put(140, 20){\vector(-1, 0){40}}
\put(160,80){\vector(1,0){40}}
\put(160,20){\vector(1,0){40}}
\put(240, 80){\vector(1, 0){40}}
\put(240, 20){\vector(1, 0){40}}
\put(60, 0){\line(1, 0){40}}
\put(60, 40){\line(1, 0){40}}
\put(60, 0){\line(0,1){40}}
\put(100, 0){\line(0,1){40}}

\put(60, 60){\line(1, 0){40}}
\put(60, 100){\line(1, 0){40}}
\put(60, 60){\line(0,1){40}}
\put(100, 60){\line(0,1){40}}

\put(200, 60){\line(1, 0){40}}
\put(200, 100){\line(1, 0){40}}
\put(200, 60){\line(0,1){40}}
\put(240, 60){\line(0,1){40}}

\put(200, 0){\line(1, 0){40}}
\put(200, 40){\line(1, 0){40}}
\put(200, 00){\line(0,1){40}}
\put(240, 00){\line(0,1){40}}

\put(0,78){\mbox{$L$}}
\put(0,18){\mbox{$\hat{L}$}}
\put(37, 85){\mbox{$u_L$}}



\put(66, 78){\mbox{$P_{Y|U}$}}
\put(66, 18){\mbox{$P_{L|\tilY}$}}
\put(145,78){\mbox{$\tilY$}}
\put(145,18){\mbox{$Y$}}

\put(206, 78){\mbox{$P_{X|Y}$}}
\put(206,18){\mbox{$P_{X|Y}$}}

\put(285,78){\mbox{$\tilX$}}
\put(285,18){\mbox{$X$}}

\linethickness{0.5mm}
\put(0,50){\line(1, 0){300}}

\end{picture}
\caption{High level description of helper's coding scheme for WAK. The upper row is a virtual scheme in which
the uniform random number $L$ is sent over channel $P_{Y|U}$. The lower row is the corresponding actual scheme
in which message $\hat{L}$ is stochastically generated via $P_{L|\tilY}$.}
\label{fig:sim}
\end{figure}

The main idea in our proof of the new non-asymptotic upper bound on the error probability of the WAK problem  is as follows:  In the channel resolvability problem, for given channel $P_{Y|U}$ and input distribution $P_U$,
the goal is  to approximate the output distribution $P_Y$ (induced by $(P_{Y|U},P_U)$)
by the output distribution $P_{\tilY}$ of codewords
for a codebook\footnote{Usually, the codebook is randomly generated according to the input distribution $P_U$.} 
${\cal C} = \{u_1,\ldots,u_{\sizeL}\}$ and the uniform random number $L \in \calL$.
Asymptotically, the approximation can be done successfully if the rate $R_2$ of the random number $L$ satisfies 
$R_2 \ge I(U;Y)$. In our helper's coding scheme (see Fig.~\ref{fig:sim}), 
we use   channel resolvability   as a virtual scheme that is applied to the reverse test channel
$P_{Y|U}$ of a given test channel  and the marginal $P_U$ of the auxiliary random variable as the input distribution. 
Then, we flip the roles of the input and the output, i.e., we construct the conditional distribution $P_{L|\tilY}$ from
the joint distribution $P_{L\tilY}$. In the actual coding scheme, the message $\hatL$ on $\calL$ is stochastically 
generated from helper's source $Y$ via $P_{L|\tilY}$, which is known as the {\em likelihood encoder} \cite{Cuff12}.
Since the successful approximation in the channel 
resolvability guarantees $P_{\tilY} \simeq P_Y$, the joint distributions in the virtual scheme and the actual scheme are also close, i.e.,
\begin{eqnarray} \label{eq:closeness-of-virtual-actual}
P_{\hatL XY} = P_Y P_{L|\tilY} P_{X|Y} \simeq P_{\tilY} P_{L|\tilY} P_{X|Y} = P_{L \tilX \tilY}.
\end{eqnarray}
The decoder reproduces $X$ via a Slepian-Wolf decoder by using $u_{\hatL}$ as the side-information. 
Because of \eqref{eq:closeness-of-virtual-actual}, the analysis of error probability can be done as if the decoder's observation
is $u_L$ and the underlying distribution is the virtual one $P_{L \tilX \tilY}$. 
Moreover, by taking the average over the randomly
generated codebook $\calC$, since the codeword $u_L$ is distributed according to $P_U$, $(\tilX,u_L)$ behaves like
$(X,U)$. Thus, the analysis of error probability can be done in the same manner as the Slepian-Wolf coding
with full side-information $U$. 
The above argument enables us to circumvent the need to use the so-called piggyback coding lemma (PBL)
and the Markov lemma \cite{Wyner75} which result in much poorer estimates on the error probability.

\subsection{Main Contributions}
We now describe the three main contributions in this paper. 

Our first main contribution in this paper  is to show improved bounds on the probabilities of error for   WAK, WZ and GP  coding. We briefly describe the form of the bound for WAK coding here. The primary part of the new upper  bound on the error probability $\Pe(\Phi)$ for WAK coding depends on two positive constants $\gamma_{\rmb}$ and $\gamma_{\rmc}$ and is essentially given by  
\begin{equation}
\Pe(\Phi)\lesssim \Pr(\calE_{\rmc}\cup\calE_{\rmb}) \label{eqn:le_approx}
\end{equation}
 where  the {\em covering error} is
\begin{equation}
\calE_{\rmc}:= \left\{\log \frac{P_{Y|U}(U|Y)}{P_Y(Y)}\ge \gamma_{\rmc}\right\}
\end{equation}
and  the {\em binning error} is 
\begin{equation}
\calE_{\rmb}:=\left\{\log \frac{1}{P_{X|U}(X|U)}\ge \gamma_{\rmb} \right\}.
\end{equation}
The notation  $\lesssim$     is not meant to be precise and, in fact, we are dropping several residual terms that do not contribute to the second-order coding rates in the $n$-fold i.i.d.\ setting if $\gamma_{\rmb}$ and $\gamma_{\rmc}$ are chosen appropriately. This result is stated precisely in Theorem~\ref{thm:dt}. From \eqref{eqn:le_approx}, we deduce that in the $n$-fold i.i.d.\ setting, if we choose $\gamma_{\rmc}$ and $\gamma_{\rmb}$ to be fixed numbers that are strictly larger than the mutual information $I(U;Y)$ and the conditional entropy $H(X|U)$ respectively, we are guaranteed that the error probability  $\Pe(\Phi)$ decays to zero.  This follows from Khintchine's law of large numbers~\cite[Ch.\ 1]{Han10}. Thus, we recover the direct part of  WAK's result. In fact, we can take this one step further (Theorem \ref{thm:gen_wak}) to obtain an achievable {\em general formula} (in the sense of Verd\'u-Han~\cite{Han10, VH94}) for the    WAK problem with general source~\cite[Ch.\ 1]{Han10}. This was previously done by Miyake-Kanaya~\cite{miyake} but their derivation is based on a different non-asymptotic formula more akin to Wyner's PBL. Also,  since we have the freedom to design  $\gamma_{\rmc}$ and $\gamma_{\rmb}$ as sequences instead of fixed positive numbers, if we let them be $O(\frac{1}{\sqrt{n}})$-larger than $I(U;Y)$ and $H(X|U)$, then the error probability is smaller than a prescribed constant depending on the implied constants in the $O(\fndot)$-notations.  This follows from the multivariate Berry-Ess\'een theorem~\cite{Got91}. This bound is useful because it is a {\em union} of two events and  $\calE_{\rmc}$ and $\calE_{\rmb}$ are both  information spectrum~\cite{Han10} events which are easy to analyze.

Secondly,  the preceding discussion  shows that the bound in~\eqref{eqn:le_approx} also yields  an achievable second-order coding rate~\cite{Hayashi08,Hayashi09}. However, unlike in the point-to-point setting~\cite{PPV10,Hayashi08,Hayashi09}, the achievable  second-order coding rate is expressed in terms of a so-called {\em dispersion matrix}~\cite{TK12}. We can easily show that if $\scR_{\WAK}(n,\veps)$ is the set of all rate pairs $(R_1,R_2)$ for which there exists a length-$n$ WAK code with error probability not exceeding $\veps>0$ (i.e., {\em the $(n,\veps)$-optimal rate region}), then for any $P_{U|Y}$, the set
\begin{equation} \label{eqn:2nd_order_intro}
  \begin{bmatrix}
I(U;Y)\\ H(X|U)
\end{bmatrix}  + \frac{\scS(\bV ,\veps)}{\sqrt{n}} + O\left(\frac{\log n}{n}\right)\bone_2
\end{equation}
is an inner bound to $\scR_{\WAK}(n,\veps)$. In \eqref{eqn:2nd_order_intro}, $\scS(\bV,\veps)\subset\bbR^2$ denotes the analogue of the $Q^{-1}$ function~\cite{TK12} and it depends on the covariance matrix of the so-called information-entropy density vector 
\begin{equation}
 \begin{bmatrix}
\log \frac{P_{Y|U}(U|Y)}{P_Y(Y)}  &  \log \frac{1}{P_{X|U}(X|U)} 
\end{bmatrix}^T .
\end{equation}
 The precise statement for the second-order coding rate for the WAK problem is given in Theorem~\ref{thm:second}. We see from~\eqref{eqn:2nd_order_intro} that for a fixed test channel $P_{U|Y}$, the redundancy  at blocklength $n$ in order to achieve an error probability $\veps>0$ is governed by the term $\frac{\scS(\bV ,\veps)}{\sqrt{n}}$. The pre-factor of this term $\scS(\bV ,\veps)$, is likened to the {\em dispersion}~\cite{PPV10, wang11,ingber11,kost12}, and depends not only the variances of the information and entropy densities but also their correlations. 
 
Thirdly, we note that the same flavour of non-asymptotic bounds and second-order coding rates hold  verbatim for the WZ and GP problems. In addition, since the canonical rate-distortion problem \cite{Sha59} is a special case of the WZ problem, we show that our non-asymptotic achievability bound for the WZ problem, when suitably specialized, yields the correct dispersion for lossy source coding~\cite{ingber11,kost12}. We do so using two methods: (i) the method of types~\cite{Csi97} and (ii) results involving the  $D$-tilted information~\cite{kost12}. Finally, we not only improve on the existing bounds for the GP problem~\cite{Tan12b, Ver12}, but we also consider an almost sure cost constraint on the channel input.

\subsection{Related Work} \label{subsection:related-work}

Wyner~\cite{Wyner75} and Ahlswede-K\"orner~\cite{Ahl75}  were the first to consider and solve (in the first-order sense) the problem of almost-lossless source coding with coded side information. Weak converses were proved in~\cite{Wyner75, Ahl75} and a strong converse was proved in \cite{Ahls76} using the ``blowing-up lemma''. An information spectrum characterization was provided by Miyake and Kanaya~\cite{miyake} and Kuzuoka~\cite{Kuz12} leveraged on the non-asymptotic bound which can be extracted from~\cite{miyake} to derive the redundancy for the WAK problem. Verd\'u~\cite{Ver12} strengthened the non-asymptotic bound and showed that the error probability for the WAK problem is essentially bounded as
\begin{equation}
\Pe(\Phi)\lesssim\Pr(\calE_{\rmc})+  \Pr(\calE_{\rmb}), \label{eqn:ver_bd}
\end{equation}
which is the result upon using the union bound on our bound in~\eqref{eqn:le_approx}. The notation $\lesssim$ means that the residual terms do not affect the second-order coding rates. 

Wyner and Ziv~\cite{wynerziv}   derived the rate-distortion function for lossy source coding with decoder side-information. However, they do not consider the probability of excess distortion. Rather, the quantity of interest is the expected  distortion. The generalization of the WZ problem for general correlated  sources was considered by Iwata and Muramatsu~\cite{iwata02} who showed that the general WZ function can be written as a difference of a limit superior in probability and a limit inferior in probability, reflecting the covering and packing components in the classical achievability  proof. 

The problem of channel coding with noncausal random state information was solved by Gel'fand and Pinsker~\cite{GP80}.  A general formula  for the  GP problem (with general channel and general state) was  provided by Tan~\cite{Tan12b}. Tyagi and Narayan~\cite{tyagi} proved the strong converse for this problem and used it to derive a sphere-packing bound. For both the WZ and GP problems, Verd\'u~\cite{Ver12} used generalizations of the   packing and covering lemmas in~\cite{elgamal} to derive non-asymptotic bounds on the probability of excess distortion (for WZ) and the average error probability (for GP). However, they yield   worse second-order rates because the main part of the bound is a sum of two or three probabilities as in~\eqref{eqn:ver_bd}, rather than the probability of the union as in \eqref{eqn:le_approx}.


In our work, we derive tight non-asymptotic bounds by using ideas from channel resolvability \cite{HV93} \cite[Ch.\ 6]{Han10} and channel simulation  \cite{BSST02}\footnote{Steinberg and Verd\'u also studied the channel simulation problem \cite{Ste96}. However, their problem formulation is slightly different from the one in \cite{BSST02}.}  to replace the covering part and Markov lemma. 
It was shown by Han and Verd\'u \cite{HV93}  that this problem is closely connected to channel coding and channel identification. Hayashi also studied the channel resolvability problem \cite{Hayashi06} and derived a non-asymptotic formula. We   leverage on a key lemma in Hayashi~\cite{Hayashi06} (and also Cuff~\cite{Cuff12}) to derive our  bounds.  

In \cite{BSST02}, Bennett \emph{et~al.}  proposed a problem to simulate a channel by the aid of common randomness. An application of the channel simulation to simulate the test channel in the rate-distortion problem was first investigated by Winter \cite{Win02}, and then extensively studied mainly in the field of the quantum information. 
Cuff investigated the trade-off between the rates of the message and common randomness for the channel simulation \cite{Cuff12}
(see also \cite{BDHSW09}). For a thorough list of literatures related to the channel simulation, see \cite{Cuff12,BDHSW09}.
In these works, channel resolvability is used as a building block for channel simulation. 
In particular, a code construction and analysis techniques that do not rely on the typicality argument were developed in \cite{Cuff12}.
The idea  to use channel simulation instead of the Markov lemma is motivated by aforementioned papers, and our code construction
and analysis are based on the ones in \cite{Cuff12}.
However, we stress that the derivations of our non-asymptotic bounds are not straightforward applications of channel simulation and channel resolvability. Indeed, our code construction is tailored to derive the bound as in \eqref{eqn:le_approx}, 
and we also introduce bounding techniques that have not appeared previously to the best of our knowledge.

Recently,  Yassaee-Aref-Gohari (YAG) \cite{YRG12}  proposed an alternative approach for   channel simulation, in which they   exploited the (multi-terminal version of) intrinsic randomness \cite[Ch.\ 2]{Han10} instead of   channel resolvability. This approach is coined {\em output statistics of random binning} (OSRB).  Although their approach is also used to replace the Markov lemma~\cite{Wyner75}, it was not {\em a priori} yet clear when~\cite{YRG12} was published whether our bounds can be also derived from the OSRB approach~\cite{YRG12}. One of difficulties to apply the OSRB  approach  for non-asymptotic analysis is that the amount of common randomness that can be used in the channel simulation is limited by the randomness of sources involved in a coding problem, which is not the case with the approach using the channel resolvability.   It was shown more recently by YAG~\cite{YAG13a} that a modification of the OSRB framework can,  in fact, be used to obtain achievable  dispersions of Marton's region for the broadcast channel~\cite{Marton79} and the wiretap channel~\cite{Wyn75}. In fact, in another concurrent work by YAG~\cite{YAG13b}, the authors derived very similar second-order results to the ones presented here. They derive bounds on the probability of error for Gel'fand-Pinsker,   Heegard-Berger and multiple  description coding~\cite{elgamal} among others. The main idea in their proofs is to use the {\em stochastic likelihood coder} (SLC) and exploit the   convexity of  $(x_1,x_2)\mapsto 1/(x_1x_2)$ (for $x_1,x_2>0$) to lower bound the probability of correct detection.  
Although the results in this paper and those in \cite{YAG13b} partly overlap, the approaches to derive the results are different.
To the best of our knowledge, this paper is the first to demonstrate usefulness of the channel simulation in non-asymptotic
analysis of network information theory problems, which we believe to be interesting in its own right.


Our main motivation in this work is to derive tight non-asymptotic bounds on the error probabilities. We are also interested in second-order coding rates. The study of the asymptotic expansion of the logarithm of  the maximum  number codewords that are achievable for $n$ uses a channel with maximum error probability no larger than $\veps$ was first done by Strassen~\cite{strassen}. This was re-popularized in recent times by  Kontoyiannis~\cite{Kot97},  Baron-Khojastepour-Baraniuk~\cite{Baron04b}, Hayashi~\cite{Hayashi08, Hayashi09}, and Polyanskiy-Poor-Verd\'u~\cite{PPV10} among others.  Second-order analysis for network information theory problems were considered in Tan and Kosut~\cite{TK12} as well as other authors~\cite{Huang12, Mol12, nomura, haim12}. However, this is the first work that considers second-order rates for problems with side-information.

\subsection{Paper Organization}
In Section~\ref{sec:prelims}, we state our notation and formally define the three coding problems with side-information. We then review existing first-order asymptotic results in Section~\ref{sec:exist}.  In Section~\ref{sec:nonasy}, we state our new non-asymptotic  bounds for the three problems. We then use these bounds to re-derive (direct parts of) known general formulas~\cite{Tan12b, miyake,iwata02} in Section~\ref{sec:gen}. Following that, we present achievable second-order coding rates for these coding problems. We will see that just as in the Slepian-Wolf setting~\cite{TK12, nomura}, the dispersion is in fact a matrix. In Section~\ref{sec:numerical}, we show via numerical examples that our non-asymptotic bounds lead to larger $(n,\veps)$-rate regions compared with~\cite{Ver12}. Concluding remarks and directions for future work are provided Section~\ref{sec:con}. This paper only contains achievability bounds. In the conclusion, we also discuss the difficulties associated with obtaining non-asymptotic converse bounds. To ensure that the main ideas are seamlessly communicated  in the main text, we relegate all proofs to the appendices.

\section{Preliminaries } \label{sec:prelims}

In this section, we introduce our notation and recall the WAK, WZ and GP problems.

\subsection{Notations}
Random variables (e.g., $X$) and their realizations (e.g., $x$) are in
capital and lower case respectively. All random variables take values in
some alphabets which are denoted in calligraphic font (e.g.,
$\calX$). The cardinality of $\calX$, if finite, is denoted as
$|\calX|$. Let the random vector $X^n:= (X_1,\ldots, X_n)$ and similarly
for a realization $x^n = (x_1, \ldots, x_n)$. The set of all
distributions supported on alphabet $\calX$ is denoted as $\scP(\calX)$.
The set of all channels with the input alphabet $\calX$ and the output
alphabet $\calY$ is denoted by $\scP(\calY|\calX)$.  We
will at times use the method of types~\cite{Csi97}.  The joint
distribution induced by a marginal distribution $P\in\scP(\calX)$ and a
channel $V\in\scP(\calY|\calX)$ is denoted interchangeably as $P\times
V$ or $PV$.  This should be clear from the context.

For a sequence $x^n=(x_1,\ldots, x_n)\in\calX^n$ in which $|\calX|$ is finite, its {\em type}  or {\em empirical distribution} is the probability mass function  $P(x)=\frac{1}{n}\sum_{i=1}^n\bone\{x=x_i\}$ where  the indicator function $\bone\{x\in\calA\}=1$ if $x\in\calA$ and $0$ otherwise. The set of types with denominator $n$ supported on  alphabet $\calX$ is denoted as $\scP_n(\calX)$. The {\em type class} of $P$ is  denoted as $\calT_P :=\{x^n \in\calX^n : x^n \mbox{ has type } P\}$.  For  a sequence   $x^n \in\calT_P$,  the set of sequences $y^n \in\calY^n$ such that $(x^n, y^n)$ has joint type $PV = P(x) V(y|x)$ is the {\em $V$-shell} $\calT_V(x^n)$.  Let $\scV_n(\calY;P)$  be the family of stochastic matrices $V : \calX \to \calY$ for which the $V$-shell of a sequence of type $P\in\scP_n(\calX)$ is  not empty. Information-theoretic quantities are denoted in the usual way. For example, $I(X;Y)$ and  $I(P,V)$ denote the mutual information where the latter expression makes clear that the joint  distribution of $(X,Y)$ is      $PV$.   All logarithms are with respect to  base~$2$ so information quantities are measured in bits.

The multivariate normal distribution with mean $\bmu$ and covariance matrix $\bSigma$ is denoted as $\calN(\bmu,\bSigma)$. The complementary Gaussian cumulative distribution function $Q(t) : =\int_{t}^{\infty} \frac{1}{\sqrt{2\pi}} e^{-u^2/2}\, du$ and its inverse is denoted as  $Q^{-1}(\veps):=\min\{ t\in\bbR: Q(t)\le \veps\}$. Finally, $|z|^+:=\max\{z,0\}$.

\subsection{The Wyner-Ahlswede-K\"orner  (WAK) Problem}
In this section, we recall the WAK problem  of lossless source coding with coded side-information~\cite{Wyner75, Ahl75}.  Let us consider a correlated source $(X,Y)$ taking values in $\calX\times\calY$ and having joint distribution $P_{XY}$.  Throughout, $X$, a discrete random variable, is the main source while $Y$ is the helper or side-information. The WAK problem involves reconstructing $X$ losslessly given rate-limited  (or coded) versions of both $X$ and  $Y$.  See Fig.~\ref{fig:wak}.

\begin{definition}
A (possibly stochastic) {\em source coding with side-information  code or Wyner-Ahlswede-K\"orner (WAK) code} $\Phi=(f,g,\psi)$ is a triple of mappings that includes two encoders $f:\calX\to\calM$ and $g:\calY\to\calL$ and a decoder $\psi:\calM\times\calL\to\calX$. The {\em error probability} of the WAK code $\Phi$ is defined as 
\begin{equation}
\Pe(\Phi):=\Pr\left\{X\ne\psi(f(X), g(Y)) \right\}. \label{eqn:error_pr}
\end{equation}
\end{definition}
In the following, we may call $f$ as the main encoder and $g$ the helper.

In Section~\ref{sec:2nd}, we   consider $n$-fold  i.i.d.\ extensions   of $X$ and $Y$, denoted as $X^n$ and $Y^n$.  In this case, we   use the subscript $n$ to specify the blocklength,  i.e., the code is $\Phi_n=(f_n,g_n,\psi_n)$ and the compression index sets are $\calM_n = f_n(\calX^n)$ and $\calL_n = g_n(\calY^n)$. In this case, we can define the pair of rates of the code $\Phi_n$ as 
\begin{align}
R_1(\Phi_n)&:=\frac{1}{n}\log|\calM_n|, \\
 R_2(\Phi_n)&:=\frac{1}{n}\log|\calL_n|.
\end{align}
\begin{definition}
The {\em $(n,\veps)$-optimal rate region  for the WAK problem} $\scR_{\WAK}(n,\veps)$ is defined as the set of all pairs of rates $(R_1, R_2)$ for which there exists  a blocklength-$n$  WAK code    $\Phi_n$ with rates at most $(R_1, R_2)$ and with error probability not exceeding $\veps$. In other words,
\begin{align}
 \scR_{\WAK}(n,\veps) := \bigg\{ (R_1, R_2)\in\bbR_+^2:& \exists \, \Phi_n \mbox{ s.t. } \nonumber \\
 & \frac{1}{n}\log|\calM_n|\le R_1, \nonumber \\
 & \frac{1}{n}\log|\calL_n|\le R_2, \nonumber \\
 & \Pe(\Phi_n)\le\veps \bigg\} \label{eqn:Rne}
\end{align}
We also define the {\em asymptotic rate regions}
\begin{align}
\scR_{\WAK}(\veps) &:=\mathrm{cl}\Bigg[ \bigcup_{n\ge 1}\scR_{\WAK}(n,\veps)\Bigg],  \label{eqn:Rwake}\\
\scR_{\WAK} &:=\bigcap_{0<\veps<1}\scR_{\WAK}(\veps). \label{eqn:Rwak}
\end{align}
where $\mathrm{cl}$ denotes set closure in $\bbR^2$.
\end{definition}
In the following, we will provide an inner bound to $\scR_{\WAK}(n,\veps)$ that improves on   inner bounds that can be derived from previously obtained non-asymptotic bounds on $\Pe(\Phi_n)$~\cite{Ver12, Kuz12}.

\subsection{The Wyner-Ziv (WZ) Problem}
In this section, we recall the WZ problem of lossy source coding with full side-information at the decoder~\cite{wynerziv}. Here, as in the WAK problem, we have a correlated source $(X,Y)$ taking values in $\calX\times\calY$ and having joint distribution $P_{XY}$. Again, $X$ is   the main source and $Y$ is the helper or side-information. Neither $X$ nor $Y$ has to be a discrete random variable. Unlike the WAK problem, it is not required to reconstruct $X$ exactly, rather a distortion $D$ between $X$ and its reproduction $\hatX$ is allowed. Let $ \hcalX$ be the reproduction alphabet and let $\rvd:\calX\times\hcalX\to [0,\infty)$ be a bounded distortion measure such that for every $x\in\calX$ there exists a $\hatx\in\hcalX$ such that $\rvd(x,\hatx)=0$ and $\max_{x,\hatx}\rvd(x,\hatx)=D_{\max} <\infty$. See Fig.~\ref{fig:wz}.

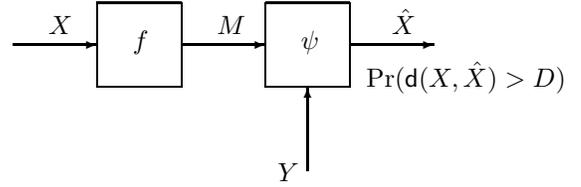
\begin{figure}
\centering
\setlength{\unitlength}{.028cm}
\begin{picture}(250, 82)
\put(0, 65){\vector(1, 0){40}}
\put(80, 65){\vector(1, 0){40}}
\put(140, 5){\vector(0, 1){40}}
\put(160, 65){\vector(1, 0){40}}


\put(40, 45){\line(1, 0){40}}
\put(40, 85){\line(1, 0){40}}
\put(40, 45){\line(0,1){40}}
\put(80, 45){\line(0,1){40}}

\put(120, 45){\line(1, 0){40}}
\put(120, 85){\line(1, 0){40}}
\put(120, 45){\line(0,1){40}}
\put(160, 45){\line(0,1){40}}

\put(126, 0){\mbox{$Y$}}
\put(17, 70){\mbox{$X$}}

\put(97, 70){\mbox{$M$}}

\put(180, 70){\mbox{$\hatX$}}
\put(168, 45){\mbox{$\Pr(\rvd( X,\hatX ) > D)$}}

\put(57, 63){\mbox{$f$}}

\put(137, 63){\mbox{$\psi$}}
\end{picture}
\caption{Illustration of the WZ problem with probability of excess distortion criterion}
\label{fig:wz}
\end{figure}

\begin{definition}
A (possibly stochastic) lossy source coding with side-information or Wyner-Ziv (WZ) code $\Phi=(f,\psi)$ is a pair of mappings that includes an encoder $f:\calX\to\calM$ and a decoder $\psi:\calM\times\calY\to\hcalX$. The {\em probability of excess distortion} for the WZ code $\Phi$ at distortion level $D$ is defined as 
\begin{equation}
\Pe (\Phi ; D) := \Pr\{ \rvd ( X, \psi ( f(X), Y) )  > D \} . \label{eqn:excess}
\end{equation}
\end{definition}
We will again consider $n$-fold extensions of $X$ and $Y$, denoted as $X^n$ and $Y^n$ in Section~\ref{sec:2nd}. The code is indexed by the blocklength as $\Phi_n = (f_n,\psi_n)$. Furthermore, the compression index set is denoted as $\calM_n= f_n(\calX^n)$.  The rate of the code $\Phi_n$ is defined as 
\begin{equation}
R(\Phi_n):=\frac{1}{n}\log |\calM_n|. \label{eqn:rate_wz}
\end{equation}
The distortion between two length-$n$ sequences $x^n \in\calX^n$ and $\hatx^n \in\hcalX^n$ is defined as 
\begin{equation}
\rvd_n(x^n,\hatx^n):=\frac{1}{n}\sum_{i=1}^n \rvd(x_i, \hatx_i).
\end{equation}
\begin{definition}
The {\em $(n,\veps)$-Wyner-Ziv rate-distortion region} $\scR_{\WZ}(n,\veps)\subset\bbR_+^2$ is the  set of all rate-distortion pairs $(R,D)$ for which there exists a blocklength-$n$ WZ code $\Phi_n$  at distortion level $D$ with rate at most $R$ and probability of excess distortion not exceeding $\veps$. In other words,
\begin{align}
 \scR_{\WZ}(n,\veps):= \bigg\{ (R,D)\in\bbR_+^2:& \exists \, \Phi_n  \mbox{ s.t. } \nonumber \\
 & \frac{1}{n}\log|\calM_n|\le R, \nonumber \\
 & \Pe (\Phi_n ; D)\le\veps \bigg\} \label{eqn:ne_rd}
\end{align}
We also define the {\em asymptotic rate-distortion regions}
\begin{align}
\scR_{\WZ}(\veps)  &:= \mathrm{cl} \Bigg[ \bigcup_{n\ge 1}\scR_{\WZ}(n,\veps)\Bigg], \\
\scR_{\WZ} &:=\bigcap_{0<\veps<1}\scR_{\WZ}(\veps).
\end{align}
The {\em $(n,\veps)$-Wyner-Ziv rate-distortion function} $R_{\WZ}(n,\veps,D)$ is defined as 
\begin{equation}
R_{\WZ} (n,\veps,D) :=\inf\{R: (R,D)\in \scR_{\WZ}(n,\veps)\} \label{eqn:ne_wz}
\end{equation}
We also define the  {\em asymptotic rate-distortion functions}
\begin{align}
R_{\WZ}(\veps,D)  &= \inf\{R: (R,D)\in \scR_{\WZ}(\veps)\} \label{eqn:wz_eps} \\
R_{\WZ}(D)  &= \lim_{\veps\to 0 }R_{\WZ}(\veps,D) \label{eqn:wz_function}
\end{align}
\end{definition}
Note that the use of the limit (as opposed to the limit superior or limit inferior) in \eqref{eqn:wz_function} is justified because $R_{\WZ}(\veps,D)$ is, from its definition, monotonically non-increasing in $\veps$. In the sequel, we will provide an inner bound to $\scR_{\WZ}(n,\veps)$ and thus an upper bound on $R_{\WZ} (n,\veps,D)$ by appealing to a new non-asymptotic upper bound on the probability of excess distortion $\Pe (\Phi_n ; D) $.  In addition, note that if $Y= \emptyset$, i.e., side-information is not available, this reduces to the point-to-point rate-distortion (lossy source coding) problem. 

Conventionally~\cite{wynerziv, elgamal}, the WZ problem is stated not with the probability of excess distortion criterion but with the {\em average fidelity criterion}. That is, the requirement that $\Pe(\Phi_n;D)\to 0$ (implicit in~\eqref{eqn:wz_function}) is replaced by 
\begin{equation}
\limsup_{n\to\infty}\bbE [ \rvd_n( X^n, \psi_n(f_n(X^n), Y^n) ) ] \le D . \label{eqn:average_fid}
\end{equation}

\subsection{The Gel'fand-Pinsker (GP) Problem}
In the previous two subsections, we dealt exclusively with source coding problems, either lossless (WAK) or lossy (WZ). In this section, we review the setup of the GP problem~\cite{GP80} which involves channel coding with noncausal state information at the encoder. It is the dual to the WZ problem~\cite{Gupta10}.  In this problem, there is a state-dependent channel $W:\calX\times\calS\to\calY$ and a random variable  representing the state $S$ with distribution $P_S$ taking values in some set $\calS$. A message $M$ chosen uniformly at random from $\calM$   is to be sent and the encoder has information about which message is to be sent as well as the channel state information $S$, which is known {\em noncausally}. (Noncausality only applies when the blocklength is larger than $1$.)  It is assumed that the message and the state are independent.  Let $\rvg:\calX\to [0,\infty)$ be some cost function.  The encoder  $f$ encodes the message and state into a  codeword (channel input) $X=f(M,S)$ that satisfies the cost constraint 
\begin{equation}
\rvg(X)\le \Gamma  ,\label{eqn:gpcost} 
\end{equation}
for some $\Gamma\ge 0$ with high probability. See precise definition/requirement in \eqref{eqn:gp_error} as well as Proposition~\ref{proposition:cost-conversion}.  The decoder receives the channel output $Y |\{ X=x, S=s \}\sim W(\fndot|x,s)$ and decides which message was sent via a decoder $\psi:\calY\to\calM$.  See Fig.~\ref{fig:gp}. More formally, we have the following definition. 

\begin{figure}
\centering
\setlength{\unitlength}{.4mm}
\begin{picture}(230, 90)
\put(0, 15){\vector(1, 0){30}}
\put(60, 15){\vector(1,0){30}}
\put(120, 15){\vector(1,0){30}}
\put(180, 15){\vector(1,0){30}}
\put(30, 0){\line(1, 0){30}}
\put(30, 0){\line(0,1){30}}
\put(60, 0){\line(0,1){30}}
\put(30, 30){\line(1,0){30}}

\put(90, 0){\line(1, 0){30}}
\put(90, 0){\line(0,1){30}}
\put(120, 0){\line(0,1){30}}
\put(90, 30){\line(1,0){30}}

\put(10, 20){  $M$}
\put(68, 20){  $X$}
\put(128, 20){  $Y$} 
\put(41, 12){$f$ } 
\put(99, 12){$W$} 

\put(150, 0){\line(1, 0){30}}
\put(150, 0){\line(0,1){30}}
\put(180, 0){\line(0,1){30}}
\put(150, 30){\line(1,0){30}}
\put(161, 12){$\psi$} 
\put(190, 20){  $\hatM $} 
\put(170, 35){  $\Pr(\hatM \ne M)$} 

\put(90, 60){\line(1, 0){30}}
\put(90, 60){\line(0,1){30}}
\put(120, 60){\line(0,1){30}}
\put(90, 90){\line(1,0){30}}
\put(105, 60){\vector(0,-1){30}}
\put(90, 75){\line(-1,0){45}}
\put(45, 75){\vector(0,-1){45}}


\put(105, 45){  $S$} 
\put(45, 45){  $S$} 
\put(99, 71){$P_{S}$} 
  \end{picture}
  \caption{Illustration of the GP problem. The channel input $X$ must satisfy \eqref{eqn:gpcost}. }
  \label{fig:gp}
\end{figure}
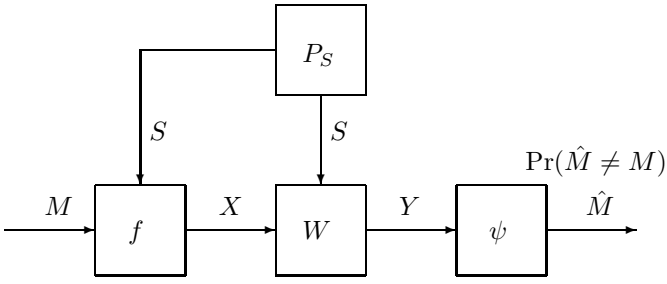
 
\begin{definition} \label{def:GP-code}
A (possibly stochastic) code for the channel coding problem with noncausal state information or Gel'fand-Pinsker (GP) code $\Phi=(f,\psi)$ is a pair of mappings that includes an encoder $f:\calM\times\calS\to\calX$   and a decoder $\psi:\calY\to\calM$. The average probability of error for the GP code is defined as 
\begin{align} \label{eqn:gp_error}
\Pe(\Phi;\Gamma) :=& 
 \frac{1}{|\calM|}\sum_{m \in \calM }\sum_{s\in\calS}P_S(s)\sum_{y\in\calY}W(y|f(m,s),s) \nonumber \\
  & \bone\left\{ \rvg( f(m,s) ) >\Gamma \, \cup \, y\in\calY\setminus\psi^{-1}(m) \right\}.
\end{align}
More simply, $\Pe(\Phi;\Gamma) = \Pr(\{\rvg(f(M,S))>\Gamma\} \, \cup \,\{  \hatM\ne M\})$ where $M$ is uniform on $\calM$ and independent of $S\sim P_S$, $\hatM :=\psi(Y)$  and $Y$ is the random variable whose conditional  distribution given $M=m$ and $S=s$ is $W(\fndot| f(m,s),s)$. 
\end{definition}
The following proposition, which will be proved in Appendix \ref{proof:proposition:cost-conversion}, guarantees that we can always convert a code in 
the sense of Definition \ref{def:GP-code} into a code in the sense of an almost sure cost constraint.
\begin{proposition}[Expurgated Code] \label{proposition:cost-conversion}
Let the set of admissible inputs in $\calX$ be 
\begin{eqnarray}
\calT_{\rmg}^{\GP}(\Gamma) := \left\{ x \in \calX : \rvg(x) \le \Gamma \right\}. \label{eqn:Tg_gp}
\end{eqnarray}
For any (stochastic) encoder $P_{X|MS}$ (this plays the role of $f$ in Definition~\ref{def:GP-code}) and   decoder $P_{\hatM|Y}$ (this plays the role of $\psi$ in Definition~\ref{def:GP-code}), there exists an encoder $\tilde{P}_{X|MS}$
such that 
\begin{equation}
\tilde{P}_X\left(\calT_{\rmg}^{\GP}(\Gamma)\right) = 1 \label{eqn:almost_sure}
\end{equation}
and 
\begin{eqnarray}
\tilde{P}_{MSXY\hatM}\left[m\neq\hatm \right] \le P_{MSXY\hatM}\left[ \rvg(x) > \Gamma \cup m \neq \hatm \right],
\end{eqnarray} 
where 
\begin{align}
P_{MSXY\hatM} &:= P_M P_S P_{X|MS} W P_{\hatM|Y}, \\
\tilde{P}_{MSXY\hatM} &:= P_M P_S \tilde{P}_{X|MS} W P_{\hatM|Y}.
\end{align}
\end{proposition}

From Proposition~\ref{proposition:cost-conversion},  noting that  $\Pe( (P_{X| MS}, P_{\hatM|Y});\Gamma)=P_{MSXY\hatM}\left[ \rvg(x) > \Gamma \cup m \neq \hatm \right]$, we see that the constraint in \eqref{eqn:gpcost} is equivalent to $\rvg(X)\le\Gamma$ {\em almost surely} (implied by \eqref{eqn:almost_sure}).  For the purposes of deriving channel simulation-based bounds in Section~\ref{sec:novel_gp}, it is  easier to work with the error criterion in \eqref{eqn:gp_error} so we adopt Definition~\ref{def:GP-code}.

In order to  obtain  achievable second-order coding rates for the GP problem, we consider $n$-fold i.i.d.\ extensions of the channel and state. Hence, for every $(s^n,x^n, y^n)$, we have $W^n(y^n|x^n,s^n)=\prod_{i=1}^n W(y_i|x_i, s_i)$ and the state $S^n$ evolves in a stationary, memoryless fashion according to $P_S$. For blocklength $n$, the code and message set are denoted as $\Phi_n= (f_n, \psi_n)$ and $\calM_n$ respectively.  The cost function is denoted as $\rvg_n:\calX^n\to [0,\infty)$ and is defined as the average of the per-letter costs, i.e., 
\begin{equation}
\rvg_n(x^n):=\frac{1}{n}\sum_{i=1}^n \rvg(x_i)
\end{equation} 
For example, in the Gaussian GP problem (which is also known as {\em dirty paper coding}~\cite{costa}),    $\rvg(x)=x^2$. This  corresponds to a power constraint and  $\Gamma$ is the upper bound on the permissible power. The rate of the code is the normalized logarithm of the number of messages, i.e., 
\begin{equation}
R (\Phi_n):=\frac{1}{n}\log |\calM_n|.
\end{equation}
\begin{definition}
The {\em $(n,\veps)$-GP capacity-cost region} $\scC_{\GP}(n,\veps) \subset\bbR_+^2$ is the set of all rate-cost pairs $(R,\Gamma)$ for which there exists a blocklength-$n$  GP code $\Phi_n$ with cost not exceeding $\Gamma$, with rate at least $R$ and probability of error not exceeding $\veps$. In other words,
\begin{align}
\scC_{\GP}(n,\veps):=\bigg\{(R,\Gamma) \in\bbR_+^2:& \exists\,\Phi_n \mbox{ s.t. } \nonumber \\
 & \frac{1}{n}\log|\calM_n|\ge R, \nonumber \\
 & \Pe(\Phi_n;\Gamma)\le\veps \bigg\}.
\end{align}
We also define the {\em asymptotic capacity-cost regions}
\begin{align}
\scC_{\GP}(\veps) &:=\mathrm{cl}\left[\bigcup_{n\ge 1}\scC_{\GP}(n,\veps)\right],\\
\scC_{\GP} &:=\bigcap_{0<\veps<1}\scC_{\GP}(\veps) . \label{eqn:gp_asymp}
\end{align}
The  {\em $(n,\veps)$-capacity-cost function}  $C_{\GP}(n,\veps,\Gamma)$ is defined as 
\begin{equation}
C_{\GP}(n,\veps, \Gamma)  :=\sup\left\{ R:(R,\Gamma)\in \scC_{\GP}(n,\veps)\right\} \label{eqn:cne_gp}
\end{equation}
We also define the {\em asymptotic capacity-cost functions}
\begin{align}
C_{\GP}(\veps, \Gamma) &:=\sup \left\{ R : (R,\Gamma)\in \scC_{\GP}( \veps) \right\} \label{eqn:CGPeps}  \\
C_{\GP}(  \Gamma) &:=\lim_{\veps\to 0} C_{\GP}(\veps, \Gamma)  \label{eqn:CGP}
\end{align}
If the cost constraint \eqref{eqn:gpcost} is absent (i.e., every codeword in $\calX^n$ is admissible), we will write $ C_{\GP}(n,\veps)$ instead of  $C_{\GP}(n,\veps,\infty)$, $\Pe(\Phi_n)$ instead of $\Pe(\Phi_n;\infty)$ and so on.
\end{definition}
Once again, the limit in~\eqref{eqn:CGP} exists because the function $C_{\GP}(\veps,\Gamma)$ is monotonically non-decreasing in $\veps$. In the sequel, we will provide a lower bound on $C_{\GP} (n,\veps ,\Gamma)$ by appealing to a new non-asymptotic upper bound on the average probability of error $\Pe(\Phi_n;\Gamma)$.

\section{Review of Existing  First-Order  Results} \label{sec:exist}

\subsection{First-Order Result for the WAK Problem}

Let $\scP(P_{XY})$ be the set of all joint distributions $P_{UXY} \in \scP(\calU\times\calX\times\calY)$  such that the $\calX\times\calY$-marginal  of $P_{UXY}$  is the source distribution $P_{XY}$,  $U-Y-X$ forms a Markov chain   in that order  and\footnote{The cardinality bound on ${\cal U}$ in the definition of $\scP(P_{XY})$ is applied when we consider the single letter characterization $\scR_{\WAK}^*$ and the inner bound to the $(n,\varepsilon)$-optimal rate region $\scR_{\WAK}(n,\varepsilon)$. It is not applied when we consider non-asymptotic analysis. Similar remarks are also applied for the WZ and GP problems.} $|\calU|\le|\calY|+1$.  Define 
\begin{align}
\scR_{\WAK}^*  :=    \bigcup_{P_{UXY}  \in\scP(P_{XY})}    \{ (R_1,R_2) \in\bbR_+^2:& R_1 \ge  H(X|U), \nonumber \\
& R_2 \ge  I(U;Y)\}. \label{eqn:wak_region}
\end{align}
Wyner~\cite{Wyner75} and Ahlswede-K\"{o}rner~\cite{Ahl75} proved the following: 
\begin{theorem}[Wyner~\cite{Wyner75}, Ahlswede-K\"{o}rner~\cite{Ahl75}]
For every $0<\veps<1$, we have 
\begin{equation}
\scR_{\WAK}(\veps)=\scR_{\WAK}=\scR_{\WAK}^*,
\end{equation}
where $\scR_{\WAK}(\veps)$ and $\scR_{\WAK}$ are defined in \eqref{eqn:Rwake} and \eqref{eqn:Rwak} respectively. 
\end{theorem}
To prove the direct part, Wyner used the PBL  and the Markov lemma~\cite{Wyner75} while  Ahlswede-K\"{o}rner~\cite{Ahl75} used a maximal code construction. Only weak converses were provided in \cite{Wyner75} and \cite{Ahl75}.  Ahlswede-G\'{a}cs-K\"{o}rner~\cite{Ahls76} proved the strong converse using entropy and image-size  characterizations~\cite[Ch.\ 15]{Csi97}, which are based on the so-called blowing-up lemma~\cite[Ch.\ 5]{Csi97}. See~\cite[Thm.\ 16.4]{Csi97}.

\subsection{First-Order Result for the WZ Problem}

Let $\scP_D(P_{XY})$ be the set of all pairs $(P_{UXY}, g)$ where $P_{U X Y} \in \scP(\calU\times\calX\times\calY)$  is a joint distribution and $g:\calU\times\calY\to\hcalX$ is a (reproduction) function such that the $\calX\times\calY$-marginal of $P_{UXY}$ is the source distribution $P_{XY}$,  $U-X-Y$ forms a Markov chain   in that order, $|\calU|\le|\calX|+1$ and the distortion constraint is satisfied, i.e.,
\begin{equation} \label{eqn:wz_distortion}
\bbE[ \rvd(X, g( U, Y) ) ]= \sum_{ u, x, y} P_{U X Y}(u,x,y) \rvd( x, g(u, y))\le D.
\end{equation}
In Section~\ref{sec:2nd_wz}, we   allow $g$ to be stochastic  (i.e., represented by a conditional probability mass function $P_{\hatX|UY}$) but we still retain the use of the notation $\scP_D(P_{XY})$.
 Define the function
\begin{equation}
R_{\WZ}^* (D):=\min_{ (P_{UXY},g)\in\scP_D(P_{XY}) } I(U;X)-I(U;Y). \label{eqn:wz_rd_func}
\end{equation}
Note from Markovity that $I(U;X)-I(U;Y)=I(U;X|Y)$. Then, we have the following asymptotic characterization of the WZ rate-distortion function. 

\begin{theorem}[Wyner-Ziv~\cite{wynerziv}] \label{thm:wz_asy}
We have 
\begin{equation}
R_{\WZ}  (D)= R_{\WZ}^* (D),
\end{equation}
where   $R_{\WZ}  (D)$ is defined in   \eqref{eqn:wz_function}. 
\end{theorem}

The direct part of the proof of the theorem in the original Wyner-Ziv paper~\cite{wynerziv} is based on the average fidelity criterion in \eqref{eqn:average_fid}. It  relies on the {\em compress-bin} idea. That is, binning is used to reduce the rate of the description of the main source to the receiver. The encoder transmits the bin index and the decoder searches within that bin for the transmitted codeword. The reproduction function $g$ is then used to reproduce the source to within a distortion $D$.  To prove Theorem~\ref{thm:wz_asy} for the probability of excess distortion criterion, we may use the new non-asymptotic bound in Section~\ref{sec:novel_wz} or the weaker non-asymptotic bounds in \cite{iwata02} or \cite{Ver12}. 

\subsection{First-Order Result for the GP Problem}
We conclude this section by stating the capacity of the GP problem~\cite{GP80}. Recall that in the GP problem, we have a channel $W:\calX\times\calS\to\calY$ and  a state distribution $P_S\in\scP(\calS)$. Assume for simplicity that all alphabets are finite sets. Let $\scP_\Gamma( W, P_S)$ be the collection of all joint distributions $P_{UXSY} \in \scP(\calU\times\calX\times\calS\times\calY)$ such that the $\calS$-marginal is $P_S$, the conditional distribution $P_{ Y|XS}=W$, $U-(X,S)-Y$ forms a Markov chain in that order, 
\begin{equation}
\bbE[\rvg(X)]\le\Gamma \label{eqn:gpcost2}
\end{equation}
and\footnote{Because of cost constraint, the second entry of the cardinality bound is increased by one compared to
the case without cost constraint \cite[Thm.\ 7.3]{elgamal}.}
$|\calU|\le \min\{ |\calX\|\calS|, |\calS| +|\calY|\}$. Define the quantity
\begin{equation}
C_{\GP}^*(\Gamma) :=\max_{P_{UXSY}\in\scP_\Gamma(W,P_S)} I(U;Y)-I(U;S), \label{eqn:gp_formula}
\end{equation}
where $I(U;Y)$ and $I(U;S)$ are computed with respect to the joint distribution  $P_{UXSY}$.  If there is no cost constraint~\eqref{eqn:gpcost2}, we simply write $C_{\GP}^*$ instead of $C_{\GP}^*(\infty)$. Then, we have the following asymptotic characterization. 
\begin{theorem}[Gel'fand-Pinsker~\cite{GP80}]
If the alphabets $\calS,\calX$ and $\calY$ are discrete, for every $0<\veps<1$, we have 
\begin{equation}
C_{\GP}(\veps)=C_{\GP} =C_{\GP}^*
\end{equation}
where  $C_{\GP}(\veps)$  and $C_{\GP}$ are defined in \eqref{eqn:CGPeps} and  \eqref{eqn:CGP} respectively. 
\end{theorem}
The direct part was proved using a covering-packing argument as well as the conditional typicality lemma (using the notion of strong typicality). Essentially, each message $m \in \calM$ is uniquely associated to a subcodebook of size $L$. To send message $m$, the encoder looks in the $m$-th subcodebook for a codeword that is jointly typical with the noncausal state. The decoder then searches for the unique  subcodebook which contains at least one  codeword  that is jointly typical with the channel output. The weak converse in the original Gel'fand-Pinsker paper was proved using the Csisz\'ar-sum-identity. See \cite[Thm.\ 7.3]{elgamal}. In fact the weak converse shows that encoding function $P_{X|US}$ can be restricted to the  set of deterministic functions. Tyagi and Narayan proved a strong converse~\cite{tyagi}  using entropy and image-size characterizations via judicious choices of auxiliary  channels.  Their proof only applies to discrete memoryless channels with discrete state distribution without cost constraints. 
\section{Main Results: Novel Non-Asymptotic Achievability Bounds} \label{sec:nonasy}

In this section, we describe our results concerning novel non-asymptotic achievability bounds for the WAK, WZ and GP problems. We show using ideas from  channel resolvability~\cite[Ch.\ 6]{Han10}  \cite{HV93} \cite{Hayashi06} and channel simulation~\cite{BSST02 , Win02,  Cuff12} that the bounds obtained by Verd\'u in~\cite{Ver12}  can be refined so as to obtain better second-order coding rates. The definition of and techniques involving  channel resolvability and channel simulation are reviewed in Appendices~\ref{app:channel_res} and~\ref{app:sim} respectively. These are  concepts that form  crucial components of the proofs of the \underline{C}hannel-\underline{S}imulation-type (CS-type)   bounds in the sequel. 

The following quantity, introduced in \cite{Cuff12}, will be used extensively in this section so we provide its definition here. For a joint distribution $P_{UY} \in \scP(\calU\times\calY)$ and a positive constant $\gamma_{\rmc}$, define
\begin{align}
& \Delta(\gamma_{\rmc}, P_{UY}) 
 := \sum_{ y\in \calY} P_Y(y) \nonumber \\
&\times \sqrt{ \sum_{u \in \calU} P_{U|Y}(u|y)  \frac{ P_{Y|U}(y|u)}{P_Y(y)} \bone\left\{   \log\frac{ P_{Y|U}(y|u)}{P_Y(y)}  \le\gamma_{\rmc}\right\} } \label{eqn:Delta}
\end{align}
By applying the Jensen inequality, we find that 
$\Delta(\gamma_{\rmc}, P_{UY})$ has the property that
\begin{equation}
\Delta(\gamma_{\rmc}, P_{UY})\le \sqrt{2^{\gamma_{\rmc}}}. \label{eqn:bound_Delta}
\end{equation}

\subsection{Novel Non-Asymptotic Achievability Bound for the WAK Problem} \label{sec:novel_wak}
Fix an auxiliary alphabet $\calU$ and a joint distribution $P_{UXY}\in\scP(P_{XY})$. See definition of $\scP(P_{XY})$ prior to \eqref{eqn:wak_region}.  For   arbitrary non-negative constants $\gamma_{\rmb}$ and $\gamma_{\rmc}$, define two sets 
\begin{align}
\calT_{\rmb}^{\WAK}(\gamma_{\rmb})&:=\left\{ (u,x) \in \calU\times\calX:   \log \frac{1}{P_{X|U}(x|u)}\le\gamma_{\rmb} \right\}, \label{eqn:T1}\\
\calT_{\rmc}^{\WAK}(\gamma_{\rmc})&:= \left\{ (u,y)\in \calU\times\calY:  \log \frac{P_{Y|U}(y|u)}{P_Y(y)}\le\gamma_{\rmc} \right\}. \label{eqn:T2}
\end{align}
These sets are similar to the {\em typical} sets used extensively in network information theory~\cite{elgamal} but note that these sets only involve the entropy and information densities. Consequently, the probabilities of these sets (events) are entropy and information spectrum quantities~\cite{Han10}. The subscripts $\rmb$ and $\rmc$ refer respectively to {\em binning} and  {\em covering}. Similar subscripts and will be used in the sequel for the other side-information problems to demonstrate the similarities between the proof techniques all of which leverage on ideas from channel resolvability~\cite[Ch.\ 6]{Han10} \cite{Hayashi06}  and channel simulation~\cite{BSST02 , Win02,  Cuff12}. 

\begin{theorem}[CS-type bound for WAK coding] \label{thm:dt}
For arbitrary $\gamma_{\rmb},\gamma_{\rmc} \ge 0$, there exists a WAK code $\Phi$  with  error probability satisfying
\begin{align} 
 \Pe   (\Phi)  
\le& P_{UXY} \left[ (u,x)  \in \calT_{\rmb}^{\WAK}(\gamma_{\rmb})^c \cup (u,y)  \in   \calT_{\rmc}^{\WAK}(\gamma_{\rmc})^c \right]   \nn\\
&   +\frac{1}{|\calM|}\sum_{(u,\tilx)\in  \calT_{\rmb}^{\WAK}(\gamma_{\rmb}) } P_U(u) 
     + \frac{\Delta(\gamma_{\rmc}, P_{UY})  }{2 \sqrt{|\calL|}}.
     \label{eqn:dt}
\end{align}
\end{theorem}
See Appendix~\ref{app:dt_wz} for the proof of Theorem~\ref{thm:dt}.   Observe that the  primary novelty of the bound in \eqref{eqn:dt} lies in the fact  that both error events $\{(u,x)  \in \calT_{\rmb}^{\WAK}(\gamma_{\rmb})^c\}$ and $\{  (u,y)  \in    \calT_{\rmc}^{\WAK}(\gamma_{\rmc})^c\}$ lie under the {\em same} probability and so can be bounded together (as a vector) in second-order coding analysis.  The sum of the information spectrum  terms (first two terms) in Verd\'u's bound in \cite[Thm.~1]{Ver12} is the result upon invoking the union bound on the first term     in~\eqref{eqn:dt}.  We illustrate the differences in the resulting second-order coding rates numerically in Section~\ref{sec:numerical}. The bound in \eqref{eqn:dt} is rather unwieldy. We can simplify it without losing too much. Indeed, using the definition  of $\calT_{\rmb}^{\WAK}(\gamma_{\rmb})$, we observe that the second term in \eqref{eqn:dt} can be bounded as 
\begin{align}
& \frac{1}{|\calM|}\sum_{(u,\tilx)\in  \calT_{\rmb}^{\WAK}(\gamma_{\rmb}) } P_U(u) \\
& =\frac{1}{|\calM|}\sum_{(u,\tilx)\in  \calT_{\rmb}^{\WAK}(\gamma_{\rmb}) } P_U(u) \frac{P_{X|U}(\tilx|u)}{P_{X|U}(\tilx|u)}\\
&\le\frac{1}{|\calM|}\sum_{(u,\tilx)\in  \calT_{\rmb}^{\WAK}(\gamma_{\rmb}) } P_U(u)  {P_{X|U}(\tilx|u)}  2^{\gamma_{\rmb}} \\
&\le\frac{  2^{\gamma_{\rmb}}}{|\calM|}.
\end{align}
Together with~\eqref{eqn:bound_Delta}, we have 
the following simplified CS-type bound, which resembles a Feinstein-type~\cite{feinstein} achievability bound (but average instead of maximum error probability).

\begin{corollary}[Simplified CS-type bound for WAK coding]\label{cor:fein}
For arbitrary $\gamma_{\rmb},\gamma_{\rmc} \ge 0$, there exists a WAK code $\Phi$ with  error probability satisfying
\begin{align} 
\Pe  (\Phi) \le&   P_{UXY} \left[ (u,x)  \in \calT_{\rmb}^{\WAK}(\gamma_{\rmb})^c \cup (u,y)  \in     \calT_{\rmc}^{\WAK}(\gamma_{\rmc})^c \right]   \nonumber \\
&  +\frac{2^{\gamma_{\rmb}}}{|\calM|} + \frac{1}{2} \sqrt{  \frac{2^{\gamma_{\rmc}}}{|\calL|}}.
\label{eqn:fein}
\end{align}
\end{corollary}

If $(X^n, Y^n)$ is drawn from the product distribution $P_{XY}^n$, then
by designing $\gamma_{\rmb}$ and $\gamma_{\rmc}$ appropriately, we see
that the dominating term in~\eqref{eqn:fein} is the first one. The other
terms vanish with $n$.  

By modifying the helper in the proof of Theorem~\ref{thm:dt}, we can show the following theorem.

\begin{theorem}[Modified CS-type bound for WAK coding] \label{thm:dt-helper-binning}
For arbitrary $\gamma_{\rmb},\gamma_{\rmc} \ge 0$, and
 positive integer $J$, there exists a WAK code $\Phi$  with  error probability satisfying
\begin{align}
\Pe   (\Phi) \le& P_{UXY} \left[ (u,x)  \in \calT_{\rmb}^{\WAK}(\gamma_{\rmb})^c \cup (u,y)  \in   \calT_{\rmc}^{\WAK}(\gamma_{\rmc})^c \right]   \nn\\
&   +\frac{1}{|\calM|}\sum_{(u,\tilx)\in
 \calT_{\rmb}^{\WAK}(\gamma_{\rmb}) } P_U(u)  \nn \\
&+\frac{J}{|\calM||\calL|}\sum_{(u,\tilx)\in  \calT_{\rmb}^{\WAK}(\gamma_{\rmb}) } P_U(u)
+ \frac{\Delta(\gamma_{\rmc}, P_{UY}) }{2 \sqrt{J}}.
\label{eqn:dt_modified}
\end{align}
\end{theorem}
See Appendix~\ref{app:dt_wz-helper-binning} for the proof of Theorem~\ref{thm:dt-helper-binning}.
By letting $J=\sizeL$ in \eqref{eqn:dt_modified}, we recover \eqref{eqn:dt} up to an additional residual term, which is unimportant in second-order analysis. 
A close inspection of the proof reveals that 
the additional term is due to 
additional random bin coding at the helper, which is not needed if $J=\sizeL$.

\begin{remark} \label{remark:non-asymtptic-bound-for-corner}
For the special case such that test channel $P_{U|Y}$ is noiseless, we can show that there exists a WAK code satisfying 
\begin{align} 
\Pe(\Phi) \le& P_{XY}\left[ (x,y) \in \calT_{\rmb}^{\WAK}(\gamma_{\rmb})^c \cup \calT_{\rms}^{\WAK}(\gamma_{\rms})^c \right] \nn \\
 & + \frac{2^{\gamma_{\rmb}}}{|\calM|} + \frac{2^{\gamma_{\rms}}}{|\calM| |\calL|}
 \label{eq:non-aymptotic-bound-for-corner}
\end{align}
for any $\gamma_{\rmb}, \gamma_{\rms} \ge 0$, where 
\begin{align}
\calT_{\rms}^{\WAK}(\gamma_{\rms}) := \left\{(x,y) \in \calX \times \calY :  \log\frac{1}{P_{XY}(x,y)}\le\gamma_{\rms} \right\}.
\end{align}
We can prove the bound \eqref{eq:non-aymptotic-bound-for-corner} by using the standard Slepian-Wolf type bin coding
for both the main encoder and the helper~\cite{TK12, nomura}. As it will turn out later in Section \ref{sec:2nd_wak}, this simple bound gives tighter second-order
achievability in some cases.  
\end{remark}


\subsection{Novel Non-Asymptotic Achievability Bound for the WZ Problem} \label{sec:novel_wz}
We now turn our attention to the WZ problem where we derive a similar
bound as in Theorem~\ref{thm:dt}. This improves on Verd\'u's bound in
Theorem~\cite[Thm.~2]{Ver12}. It again uses the same CS idea for the
covering part.  

Define the three sets for fixed $(P_{UXY},g)\in\scP_D(P_{XY})$ and non-negative constants $\gamma_{\rmp}$ and $\gamma_{\rmc}$:
\begin{align}
\calT_{\rmp}^{\WZ}(\gamma_{\rmp})&:=\left\{ (u,y) \in \calU\times\calY:   \log \frac{ P_{Y|U}(y|u)}{P_Y(y)}  \ge \gamma_{\rmp} \right\} \label{eqn:T1_wz}\\
\calT_{\rmc}^{\WZ}(\gamma_{\rmc})&:= \left\{ (u,x)\in \calU\times\calX:  \log \frac{P_{X|U}(x|u)}{P_X(x)}\le \gamma_{\rmc} \right\} \label{eqn:T2_wz} \\
\calT_{\rmd}^{\WZ}(D)&:=  \left\{ (u,x,y)\in\calU\times\calX\times\calY: \rvd(x, g(u,y))\le D \right\}\label{eqn:Td_wz}.
\end{align}
These sets have intuitive explanations: $\calT_{\rmc}^{\WZ}(\gamma_{\rmc})^c$ represents the {\em covering} error that $U$ is unable to describe $X$ to the desired level indicated by $\gamma_{\rmc}$; $\calT_{\rmp}^{\WZ}(\gamma_{\rmp})^c$ represents the {\em packing} error in which the decoder is unable to decode the correct codeword $U$ given $Y$ using a threshold test based on the information density statistic and $\gamma_{\rmp}$;  $\calT_{\rmd}^{\WZ}(D)^c$ represents the {\em distortion} error in which the the reproduction $\hatX$ not within a distortion of $D$ of the source $X$. 

In the following, we allow the reproduction function
$g\colon\calU\times\calY\to\hatcalX$ to be stochastic; i.e., we consider
a reproduction channel $P_{\hatX|UY}\colon\calU\times\calY\to\hatcalX$.
When we consider a stochastic function instead of a deterministic one, we will use the set
\begin{equation}
 \calTsWZ(D):=  \left\{ (x,\hatx)\in\calX\times\hatcalX: \rvd(x,\hatx)\le D \right\}\label{eqn:Td_wz_stoch}
\end{equation}
instead of $\calT_{\rmd}^{\WZ}(D)$; see \eqref{eqn:dt_wz} and Remark
\ref{remark:WZ_deterministic_function} below.

In this subsection, a pair $(P_{U|X},P_{\hatX|UY})$ of a test
channel $P_{U|X}\colon\calX\to\calU$ and a reproduction channel
$P_{\hatX|UY}\colon\calU\times\calY\to\hatcalX$ is fixed.  Note that the joint
distribution $P_{UXY\hatX}$ of $U,X,Y,\hatX$ is also fixed as
\begin{equation}
 P_{UXY\hatX}(u,x,y,\hatx)=P_{XY}(x,y)P_{U|X}(u|x)P_{\hatX|UY}(\hatx|u,y).
\end{equation}

\begin{theorem}[CS-type bound for WZ coding] \label{thm:cs_wz}
For arbitrary constants $\gamma_{\rmp},\gamma_{\rmc} \ge 0$ and  positive integer $L$, there exists a WZ code $\Phi$  with    probability  of excess distortion satisfying
\begin{align}
\Pe   (\Phi;D)  
&\le P_{UX Y \hatX} [ (u,x)\in\calT_\rmc^\WZ(\gamma_\rmc)^c \nn\\
 & ~~~~~~~~~\cup (x,\hatx)\in\calTsWZ(D)^c \cup (u,y)\in\calT_\rmp^\WZ(\gamma_\rmp)^c ]  \nn\\
 &~+\frac{L}{|\calM|} \sum_{(u,y) \in   \calT_\rmp^\WZ(\gamma_\rmp)} P_U(u) P_Y(y) + \frac{\Delta(\gamma_{\rmc}, P_{UX})}{2\sqrt{L}} .
 \label{eqn:dt_wz}
\end{align}
where $\Delta(\gamma_{\rmc},P_{UX})$ is defined in \eqref{eqn:Delta}.
\end{theorem}

\begin{remark}\label{remark:WZ_deterministic_function}
 If $P_{\hatX|UY}$ is deterministic and represented by $g\colon\calU\times\calY\to\hatcalX$ then the event $\{ (x,\hatx)\in\calTsWZ(D)^c \}$ can be replaced by $\{(u,x,y) \in \calT_{\rmd}^{\WZ}(D)^c\}$.  In fact, by an application of the functional representation lemma~\cite[Appendix~A]{elgamal}, the assumption that the reproduction channel   $P_{\hatX|UY}$ is deterministic can be made without any loss of generality. 
\end{remark}

The proof of Theorem~\ref{thm:cs_gp} is provided in
Appendix~\ref{app:cs_wz}.  
As with Theorem~\ref{thm:dt}, the main novelty of our bound lies in the fact that the three error events lie under the same probability, making it amendable to treat all three error events {\em jointly}. The residual terms in \eqref{eqn:dt_wz} (namely, the second, third and fourth terms) are relatively small with a proper choice of constants $\gamma_{\rmp},\gamma_{\rmc}$ and $L\in\bbN$ as we shall see in the sequel. We can again relax the somewhat cumbersome second and third terms in \eqref{eqn:dt_wz} by noting the definition of $\calT_{\rmp}^{\WZ}(\gamma_{\rmp})$ and by going through the same steps  to  upper bound $\Delta$; cf.~\eqref{eqn:bound_Delta}. We thus obtain:

\begin{corollary}[Simplified CS-type bound for WZ coding] \label{thm:fein_wz}
For arbitrary constants $\gamma_{\rmp},\gamma_{\rmc}\ge 0$ and   positive integer $L$, there exists a WZ code $\Phi$  with    probability  of excess distortion satisfying
\begin{align} 
\Pe   (\Phi;D)  &\le P_{UXY\hatX} [ (u,y)  \in \calT_{\rmp}^{\WZ}(\gamma_{\rmp})^c \nn\\
&~~~~~~~\cup (u,x)  \in   \calT_{\rmc}^{\WZ}(\gamma_{\rmc})^c\cup (x,\hatx)\in\calTsWZ(D)^c ]   \nn\\
&~+\frac{L}{2^{\gamma_{\rmp}}|\calM| } + \frac{1}{2}\sqrt{\frac{2^{\gamma_{\rmc}}}{L}}    .
\label{eqn:fein_wz}
\end{align}
\end{corollary}
To obtain achievable second-order coding rates for the WZ problem, we evaluate the bound in \eqref{eqn:fein_wz} for appropriate choices of $\gamma_{\rmp},\gamma_{\rmc} \ge 0$   and $L\in\bbN$ in Section~\ref{sec:2nd_wz}.  Since the lossy source coding problem is a special case of WZ coding, we use a specialization of the bound in~\eqref{eqn:fein_wz} to derive an achievable dispersion (or second-order coding rate) of lossy source coding~\cite{ingber11,kost12}, which turns out to be tight.

\subsection{Novel Non-Asymptotic Achievability Bound for the GP Problem} \label{sec:novel_gp}
This section presents with a novel non-asymptotic achievability bound for the GP problem, which is the dual of the WZ problem~\cite{Gupta10}. Our bound improves on Verd\'u's non-asymptotic bound for GP coding~\cite[Thm.~3]{Ver12} and uses  the same Channel-Simulation idea for the covering part. 

To state the bound, we define the sets  
\begin{align}
\calT_{\rmp}^{\GP}(\gamma_{\rmp}) & := \left\{ (u,y)\in\calU\times\calY : \log \frac{P_{Y|U}(y|u ) }{P_Y(y)}\ge\gamma_{\rmp}\right\}  \label{eqn:Tp_gp}\\
\calT_{\rmc}^{\GP}(\gamma_{\rmc}) & := \left\{ (u,s)\in\calU\times\calS : \log \frac{P_{S|U}(s|u) }{P_S(s)} \le \gamma_{\rmc}\right\}  \label{eqn:Tc_gp} 
\end{align}
These are analogous to the typical sets used extensively in network information theory~\cite{elgamal} but they only involve the information densities. The first set  in~\eqref{eqn:Tp_gp}  represents {\em packing} event while the second in  \eqref{eqn:Tc_gp} represents {\em covering} event. Also recall the definition of the set $\calT_{\rmg}^{\GP}(\Gamma)$ in  \eqref{eqn:Tg_gp} which represents satisfaction of the cost constraints.  

In the following, the distribution $P_{USXY} \in\scP(\calU\times\calS\times\calX\times\calY)$ satisfying (i) the $\calS$-marginal of $P_{USXY}$ is $P_{S}$, (ii) $P_{Y|XS}=W$ and (iii)  $U-(X,S)-Y$ forms a Markov chain is fixed. Note the encoding function $P_{X|US}$ is allowed to be stochastic but  just as in Remark~\ref{remark:WZ_deterministic_function}, there is no loss in assuming $P_{X|US}$ is deterministic by the functional representation lemma. We prefer to use  $P_{X|US}$ for convenience.

\begin{theorem}[CS-type bound for GP coding] \label{thm:cs_gp}
For arbitrary constants $\gamma_{\rmp},\gamma_{\rmc} \ge 0$  and   positive integer $L$, there exists a GP code $\Phi$  with    average error probability satisfying
\begin{align}
\Pe   (\Phi;\Gamma )  
&\le P_{USXY} [ (u,y)  \in \calT_{\rmp}^{\GP}(\gamma_{\rmp})^c \nn\\
&~~~~~~~ \cup (u,s)  \in   \calT_{\rmc}^{\GP}(\gamma_{\rmc})^c \cup x\in\calT_{\rmg}^{\GP}(\Gamma)^c  ]         \nn\\*
&~ +L |\calM| \sum_{(u,y) \in \calT_{\rmp}^{\GP}(\gamma_{\rmp})} P_U(u) P_Y(y) +\frac{\Delta(\gamma_{\rmc}, P_{US}) }{2\sqrt{L  }}  \label{eqn:dt_gp} 
\end{align}
where $\Delta(\gamma_{\rmc},P_{US})$ is defined in \eqref{eqn:Delta}.
\end{theorem}
Because the technique to prove Theorem~\ref{thm:cs_gp} is similar to that for Theorems~\ref{thm:dt} and \ref{thm:cs_wz}, we only sketch the code construction in Appendix~\ref{app:cs_gp}. In the second-order asymptotics sense, Theorem~\ref{thm:cs_gp}  improves on~\cite[Thm.~3]{Ver12} because the error events are under the {\em same} error probability.  Notice that unlike the existing asymptotic and non-asymptotic results for GP coding \cite{Tan12b,Ver12, Mou07}, the channel input $x$ satisfies the cost constraint \eqref{eqn:gpcost} or its almost sure equivalent  (cf.~Proposition~\ref{proposition:cost-conversion}). Direct application of \eqref{eqn:bound_Delta}  to bound  $\Delta(\gamma_{\rmc}, P_{US})$  and the definition of $\calT_{\rmp}^{\GP}(\gamma_{\rmp})$ in~\eqref{eqn:Tp_gp} yields the following:

\begin{corollary}[Simplified CS-type bound for GP coding] \label{thm:fein_gp}
For arbitrary constants $\gamma_{\rmp},\gamma_{\rmc} \ge 0$  and   positive integer $L$, there exists a GP code $\Phi$  with    average error probability satisfying
\begin{align} 
\Pe   (\Phi;\Gamma )  
 &\le P_{USXY} [ (u,y)  \in \calT_{\rmp}^{\GP}(\gamma_{\rmp})^c \nn\\
 &~~~~~~~~~\cup (u,s)  \in   \calT_{\rmc}^{\GP}(\gamma_{\rmc})^c   \cup x\in\calT_{\rmg}^{\GP}(\Gamma)^c ]  \nn\\
 &~+\frac{L |\calM|}{2^{\gamma_{\rmp}} }    +  \frac{1}{2}\sqrt{\frac{2^{\gamma_{\rmc}}}{L}}     .
\label{eqn:fein_gp}
\end{align}
\end{corollary}
To obtain achievable second-order coding rates for the GP problem, we evaluate the bound in \eqref{eqn:fein_gp} for appropriate choices of $\gamma_{\rmp},\gamma_{\rmc}$  and $L\in\bbN$ in Section~\ref{sec:2nd_gp}.  
\section{General Formulas} \label{sec:gen}
In this section, we use  the simplified CS-type bounds in Corollaries~\ref{cor:fein}, \ref{thm:fein_wz} and \ref{thm:fein_gp} to derive achievable general formulas for the optimal rate region of the WAK problem, the rate-distortion function of the  WZ problem  and the capacity of the GP problem. This allows us to recover known results in~\cite{miyake,iwata02,Tan12b}. By {\em general formula}, we mean that we consider sequences of these problems and do not place any underlying structure such as stationarity, memorylessness and ergodicity on the source and channel~\cite{Han10, VH94}.  To state our results, let us  first recall the following probabilistic limit operations.   Their properties are similar to the limit superior and limit inferior for numerical sequences in mathematical analysis and are summarized in~\cite{Han10}.
\begin{definition}
Let $\bU:=\{U_n \}_{n=1}^{\infty}$ be a sequence of real-valued random variables. The {\em limit superior in probability} of $\bU$ is defined as 
\begin{equation}
\plimsup_{n\to\infty} U_n := \inf \left\{ \alpha\in\bbR: \lim_{n\to\infty}\Pr( U_n>\alpha) = 0 \right\}. \label{eqn:plimsup}
\end{equation}
The {\em limit inferior in probability} of $\bU$ is defined as 
\begin{equation}
\pliminf_{n\to\infty} U_n := - \plimsup_{n\to\infty} (-U_n ) \label{eqn:pliminf}
\end{equation}
\end{definition}
We also recall the following definitions from Han~\cite{Han10}. These definitions  play a prominent role in the rest of this section.
\begin{definition}
Given a pair of   stochastic processes $( \bX,\bY)=\{ X^n,Y^n\}_{n=1}^{\infty}$ with joint distributions $\{P_{X^n, Y^n}\}_{n=1}^{\infty}$, the {\em spectral sup-mutual information rate} is defined as 
\begin{equation}
\overI(\bX;\bY):=\plimsup_{n\to\infty}\frac{1}{n}\log \frac{P_{Y^n|X^n}(Y^n|X^n)}{P_{Y^n}(Y^n)}. \label{eqn:spectral_mi}
\end{equation}
The {\em spectral  inf-mutual information rate} $\underI(\bX;\bY)$  is defined as in~\eqref{eqn:plimsup} with $\pliminf$ in place of $\plimsup$.  The {\em spectral sup- and inf-conditional mutual information rates}  are defined similarly. 

The {\em spectral sup-conditional  entropy rates} is defined as 
\begin{equation}
\overH(\bY|\bX):=\plimsup_{n\to\infty}\frac{1}{n}\log \frac{1}{P_{Y^n|X^n}(Y^n|X^n)} . \label{eqn:spectral_ce}
\end{equation}
The {\em spectral inf-conditional  entropy rates} is defined as in \eqref{eqn:spectral_ce} with $\pliminf$ in place of $\plimsup$. 
\end{definition}
\subsection{General Formula for the WAK problem} \label{sec:gen_wak}
In this section, we consider sequences of the WAK problem indexed by the blocklength $n$ where the  sequence of source distributions $\{P_{X^n  Y^n}\}_{n=1}^{\infty}$ is {\em general}, i.e., we do not place any assumptions on the structure of the source such as stationarity, memorylessness and ergodicity. We aim to characterize an inner bound to the optimal rate region defined in \eqref{eqn:Rwak}. We show that our inner bound coincides with that derived by Miyake and Kanaya~\cite{miyake} but is derived based on  the upper bound on the error probability provided in our CS-type  bound in Corollary~\ref{cor:fein}. The choice of the parameters $\gamma_{\rmb}, \gamma_{\rmc}$ and $\delta$ plays a crucial role and guides our choice of these parameters for second-order coding analysis in the following section. 

Let $\scP( \{ P_{X^n  Y^n}\}_{n=1}^{\infty}  )$  be  the set of all sequences of distributions $\{ P_{U^n X^n Y^n} \}_{n=1}^{\infty}$ such that for every $n\ge 1$, $U^n-Y^n-X^n$ forms a Markov chain and the $(\calX^n\times\calY^n)$-marginal of $P_{U^n X^n Y^n} $ is $P_{ X^n Y^n}$. Define the set
\begin{align}
\hat{\scR}_{\WAK}^*&:=\bigcup_{ \{ P_{U^n X^n Y^n} \}_{n=1}^{\infty}\in \scP( \{ P_{X^n  Y^n}\}_{n=1}^{\infty}  ) } \nn\\
 &\left\{ (R_1, R_2)\in\bbR_+^2: R_1\ge\overH(\bX|\bU),R_2\ge\overI(\bU;\bY) \right\}
\end{align}
\begin{theorem}[Inner Bound to  the Optimal Rate Region for WAK~\cite{miyake}] \label{thm:gen_wak}
We have 
\begin{equation}
\hat{\scR}_{\WAK}^*\subset\scR_{\WAK}. \label{eqn:inner_general_wak}
\end{equation}
\end{theorem}
We remark that by using techniques from~\cite{Ste96}, Miyake and Kanaya~\cite{miyake} showed that~\eqref{eqn:inner_general_wak} is in fact an equality, i.e., $\hat{\scR}_{\WAK}^*$ is also an outer bound to $\scR_{\WAK}$. In addition, when the   source distributions $\{P_{X^n  Y^n}\}_{n=1}^{\infty}$  are stationary and memoryless (and the alphabets $\calX$ and $\calY$ are discrete and finite), $\hat{\scR}_{\WAK}^*$ reduces to the single-letter region ${\scR}_{\WAK}^*$ defined in \eqref{eqn:wak_region}. This follows easily from the law of large numbers. The proof  of Theorem~\ref{thm:gen_wak} follows directly from the finite blocklength bound in Corollary~\ref{cor:fein}. In fact, the weaker bounds in \cite{Kuz12} and \cite{Ver12}  suffice for this purpose.
\begin{proof}
Consider \eqref{eqn:fein} and let us  fix a process $\{ P_{U^n X^n Y^n} \}_{n=1}^{\infty}\in \scP( \{ P_{X^n  Y^n}\}_{n=1}^{\infty} )$ and a constant $\eta>0$. Set  
\begin{align}
\frac{1}{n}\log |\calM| &:= \overH(\bX|\bU)+2 \eta \label{eqn:sizeM}\\
\frac{1}{n}\log |\calL| &:= \overI(\bU;\bY)+2 \eta \label{eqn:sizeL}\\
\gamma_{\rmb} &:= n( \overH(\bX|\bU)+ \eta)\\
\gamma_{\rmc} &:= n(\overI( \bU;\bY)+ \eta)
\end{align}
Then for blocklength $n$, the probability on the RHS of~\eqref{eqn:fein} can be written as 
\begin{align}
  &P_{U^n X^n Y^n} \bigg[  \bigg\{\frac{1}{n}\log\frac{1}{P_{X^n|U^n}(X^n|U^n)} \ge \overH(\bX|\bU)+ \eta \bigg\} \nn\\
   &~~~~\bigcup\bigg\{ \frac{1}{n} \log\frac{ P_{Y^n|U^n}(Y^n|U^n)}{P_{Y^n}(Y^n)}\ge \overI( \bU;\bY)+ \eta \bigg\}\bigg]    \label{eqn:prob_in_fein}
\end{align}
By the definition of the spectral sup-entropy rate and the spectral sup-mutual information rate, the probabilities of both events in \eqref{eqn:prob_in_fein} tend  to zero.  Further,
\begin{align}
\frac{2^{\gamma_{\rmb}}}{|\calM|} &=2^{-n\eta}\to 0  ,\quad \mbox{and}\quad  \frac{1}{2}\sqrt{\frac{2^{\gamma_{\rmc}}}{|\calL|}} =\frac{1}{2}\cdot 2^{-n\eta/2}\to 0. \label{eqn:res_terms}
\end{align}
Hence, $\Pe(\Phi_n)\to 0$. Since $\eta > 0$ is arbitrary, from \eqref{eqn:sizeM} and \eqref{eqn:sizeL} we deduce that any    pair of rates $(R_1,R_2)$ satisfying  $R_1 > \overH(\bX|\bU)$ and $R_2 > \overI(\bU;\bY)$  is achievable. 
\end{proof}

\subsection{General Formula for the WZ problem} \label{sec:gen_wz}
In a similar way, we can recover the general formula for WZ coding derived by Iwata and Muramatsu~\cite{iwata02}. Note however, that we directly work with the probability of excess distortion, which is related to but different from the maximum-distortion criterion employed in~\cite{iwata02}. Once again, we assume that the source is $\{ P_{X^n  Y^n}\}_{n=1}^{\infty}$ is {\em general} in the sense explained in Section~\ref{sec:gen_wak}.

Let $\scP_D( \{ P_{X^n  Y^n}\}_{n=1}^{\infty}  )$  be  the set of all sequences of distributions $\{ P_{U^n X^n Y^n} \}_{n=1}^{\infty}$ and reproduction functions $\{ g_n:\calU^n\times\calY^n\to\hcalX^n \}$ such that for every $n\ge 1$,  $U^n-X^n-Y^n$ forms a Markov chain, the $(\calX^n\times\calY^n)$-marginal of $P_{U^n X^n Y^n} $ is $P_{ X^n Y^n}$ and 
\begin{equation}
\plimsup_{n\to\infty}  \rvd_n (X^n , g_n(U^n, Y^n))\le D \label{eqn:dist_cond}
\end{equation}
Define the rate-distortion function
\begin{equation}
\hatR^*_{\WZ}(D) := \inf \left\{ \overI(\bU;\bX)-\underI(\bU;\bY) \right\} \label{eqn:rd_func}
\end{equation}
where the infimum is over all $  \{ P_{U^n X^n Y^n}, g_n \}_{n=1}^{\infty} \in \scP_D( \{ P_{X^n  Y^n}\}_{n=1}^{\infty}  )$. 
\begin{theorem}[Upper Bound to the  Rate-Distortion Function for WZ~\cite{iwata02}] \label{thm:gen_wz}
We have 
\begin{equation}
R_{\WZ}(D) \le \hatR^*_{\WZ}(D) . \label{eqn:inner_general_wz}
\end{equation}
\end{theorem}
Iwata and Muramatsu~\cite{iwata02} showed in fact that \eqref{eqn:inner_general_wz} is an equality by proving a converse along the lines of \cite{Ste96}. It can be shown that the general rate-distortion function defined in \eqref{eqn:rd_func} reduces to the one derived by Wyner and Ziv~\cite{wynerziv} in the case where the alphabets are finite and the source is stationary and memoryless.  Also Iwata and Muramatsu~\cite{iwata02} showed that deterministic reproduction functions $g_n:\calU^n\times\calY^n\to\hcalX^n$ suffice and we do not need the more general stochastic reproduction functions $P_{\hatX^n|U^nY^n}$.
\begin{proof}
Let $\eta>0$. 
We start from the bound on the probability of excess distortion in~\eqref{eqn:fein_wz},
where we first consider $D+\eta$ instead of $D$. Let us fix the sequence of distribution and the sequence of  functions $  \{ (P_{U^n X^n Y^n}, g_n )\}_{n=1}^{\infty} \in \scP_{D }( \{ P_{X^n  Y^n}\}_{n=1}^{\infty}  )$.  Set
\begin{align}
\frac{1}{n}\log |\calM| &:=  \overI(\bU;\bX)-\underI(\bU;\bY) + 4\eta \label{eqn:Mwz}\\
\frac{1}{n}\log L &:=  \overI(\bU;\bX) + 2 \eta\\
\gamma_{\rmp} &:= n( \underI(\bU;\bY) -   \eta )\\
\gamma_{\rmc} &:= n( \overI(\bU;\bX)  +   \eta ) . 
\end{align}
Then, the probability in \eqref{eqn:fein_wz}  for  blocklength $n$ can be written as 
\begin{align}
P_{U^n X^n Y^n} \bigg[ & \left\{\frac{1}{n}\log \frac{P_{Y^n|U^n}(Y^n|U^n)}{P_{Y^n}(Y^n)}\le  \underI(\bU;\bY) -   \eta   \right\}\nn\\*
&\quad\bigcup\left\{\frac{1}{n}\log \frac{P_{X^n|U^n}(X^n|U^n)}{P_{X^n}(X^n)}\ge \overI(\bU;\bX)  +   \eta \right\}\nn\\*
&\qquad \bigcup\bigg\{  \rvd_n(X^n, g_n(U^n, Y^n))\ge D+\eta \bigg\} \bigg] \label{eqn:three_terms_prob}
\end{align}
By the definition of the spectral sup- and inf-mutual information rates and the distortion condition in \eqref{eqn:dist_cond}, we observe that the probability in \eqref{eqn:three_terms_prob} tends to zero as $n$ grows.  By a similar calculation as in~\eqref{eqn:res_terms}, the other terms in \eqref{eqn:fein_wz} also tend to zero. Hence, the probability of excess distortion $\Pe(\Phi_n;D+\eta)\to 0$ as $n$ grows. This holds for every $\eta>0$. By \eqref{eqn:Mwz}, the any rate below $\overI(\bU;\bX)-\underI(\bU;\bY)+4\eta$ is achievable. In order to complete the proof, we choose a positive sequence satisfying $\eta_1 > \eta_2 > \cdots > 0$ and
$\eta_k \to 0$ as $k \to \infty$. Then, by using the {\em diagonal line argument} \cite[Thm.~1.8.2]{Han10}, we complete the proof of \eqref{eqn:inner_general_wz}. 
\end{proof}
\subsection{General Formula for the GP problem} \label{sec:gen_gp}
We conclude this section by showing that the non-asymptotic bound on the average probability of error derived in Corollary~\ref{thm:fein_gp} can be adapted to recover the general formula for the GP problem derived in Tan~\cite{Tan12b}. Here, both the state distribution $\{P_{S^n} \in \scP(\calS^n)\}_{n=1}^{\infty}$ and the channel $\{W^n : \calX^n\times\calS^n\to\calY^n\}_{n=1}^{\infty}$ are general. In particular, the only requirement on  the stochastic mapping $W^n$ is that for every $(x^n,s^n)\in\calX^n\times\calS^n$,
\begin{equation}
\sum_{y^n\in\calY^n}W^n(y^n|x^n,s^n)=1.
\end{equation}
Let $\scP_{\Gamma}( \{ W^n, P_{S^n}\}_{n=1}^{\infty} )$  be the family of joint distributions $P_{U^n S^n X^n Y^n}$ such that for every $n\ge 1$, $U^n-(X^n,S^n)-Y^n$ forms a Markov chain, the $\calS^n$-marginal of $P_{U^n S^n X^n Y^n}$ is $P_{S^n}$, the channel law $P_{Y^n|X^n, S^n}=W^n$ and
\begin{equation}
\plimsup_{n\to\infty}  \rvg_n (X^n  )\le \Gamma \label{eqn:cost_con}
\end{equation}
  Define the quantity
\begin{equation}
\hatC_{\GP}^*(\Gamma) := \sup\left\{ \underI(\bU;\bY)-\overI(\bU;\bS)\right\} \label{eqn:gen_gp}
\end{equation}
where the supremum is over all  joint distributions $ \{P_{U^n S^n X^n Y^n}\}_{n=1}^{\infty} \in \scP_\Gamma( \{ W^n, P_{S^n}\}_{n=1}^{\infty} )$.

\begin{theorem}[Lower Bound to the  GP capacity~\cite{Tan12b}] \label{thm:gen_gp}
We have 
\begin{equation}
C_{\GP}(\Gamma)\ge \hatC_{\GP}^* (\Gamma) . \label{eqn:lower_gen_gp}
\end{equation}
\end{theorem}
Tan~\cite{Tan12b} also showed that the inequality in \eqref{eqn:lower_gen_gp} is, in fact, tight. However, unlike  in the general WZ scenario, the encoding function $P_{X^n|U^n S^n}$ cannot be assumed to be deterministic in general.  When  the channel and state are discrete, stationary and memoryless, Tan~\cite{Tan12b} showed that the general formula in \eqref{eqn:gen_gp} reduces to the conventional one derived by Gel'fand-Pinsker~\cite{GP80} in~\eqref{eqn:gp_formula}.  The proof of Theorem~\ref{thm:gen_gp} parallels that for Theorem~\ref{thm:gen_wz} and thus, we  omit it. 

\section{Achievable Second-Order Coding Rates} \label{sec:2nd}
In this section, we demonstrate achievable second-order coding rates~\cite{strassen, PPV10,Hayashi09, Hayashi08, Kot97} for the three side-information problems of interest. Essentially, we are interested in characterizing the $(n,\veps)$-optimal rate region for the WAK problem, the $(n,\veps)$-Wyner-Ziv rate-distortion function and the $(n,\veps)$-capacity of GP problem up to the second-order term.  We do this by applying the multidimensional Berry-Ess\'een theorem~\cite{Got91, Ben03} to the finite blocklength CS-type bounds in Corollaries~\ref{cor:fein}, \ref{thm:fein_wz} and \ref{thm:fein_gp}.   Throughout, we will not concern ourselves with optimizing    the third-order terms.  

The following important definition will be used throughout this section.
\begin{definition}
Let $k$ be a positive integer. Let $\bV \in \bbR^{k\times k}$  be a positive-semidefinite matrix that is not the all-zeros matrix but is allowed to be rank-deficient. Let the Gaussian random vector $\bZ\sim\calN(\bzero,\bV)$. Define the set 
\begin{equation}
\scS(\bV,\veps):= \{ \bz\in \bbR^k: \Pr ( \bZ\le\bz )\ge 1-\veps\}. \label{eqn:Sv}
\end{equation}
\end{definition}
This set was introduced in~\cite{TK12} and is, roughly speaking, the multidimensional analogue of the $Q^{-1}$ function.  Indeed, for $k=1$ and  any standard deviation $\sigma>0$, 
\begin{equation}
 \scS(\sigma^2,\veps) = [\sigma Q^{-1}(\veps), \infty). 
\end{equation}
Also, $\bone_k$ and $\bzero_{k\times k}$ denote the length-$k$ all-ones column vector  and the $k\times k$ all-zeros matrix respectively. 

\subsection{Achievable Second-Order Coding Rates for the WAK problem} \label{sec:2nd_wak}
In this section, we derive an inner bound to $\scR_{\WAK}(n,\veps)$  in \eqref{eqn:Rne}  by the use of Gaussian approximations.   Instead of simply applying the Berry-Ess\'een theorem to the information spectrum term within the  simplified CS-type bound in~\eqref{eqn:fein}, we enlarge our inner bound   by using a ``time-sharing'' variable $T$, which is independent of $(X,Y)$. This technique was also used for the multiple access channel  (MAC) by Huang and Moulin~\cite{Huang12}. Note that in the finite blocklength setting, the region $\scR_{\WAK}(n,\veps)$ does not have to be convex unlike in the asymptotic case; cf.~\eqref{eqn:wak_region}.  For fixed finite sets $\calU$ and $\calT$, let $\tilde{\scP}(P_{XY})$ be the set of   all  $P_{UTXY}\in\scP(\calU\times\calT\times\calX\times\calY)$ such that the $\calX\times\calY$-marginal of $P_{UTXY}$ is $P_{XY}$,  $U-(Y,T)-X$  forms a Markov chain and $T$ is independent of $(X,Y)$.
\begin{definition}
The {\em entropy-information density vector} for the WAK problem for $P_{UTXY}\in\tilde{\scP}(P_{XY})$  is defined as 
\begin{equation} \label{eqn:ei_vec}
\bj(U,X,Y|T):= \begin{bmatrix}
 \log \frac{1}{P_{X|UT}(X|U,T)} \\ \log\frac{P_{Y|UT}(Y|U,T)}{P_{Y }(Y )} 
\end{bmatrix}. 
\end{equation}
\end{definition}
Note that the mean of the entropy-information density vector  in~\eqref{eqn:ei_vec} is the vector of the entropy and mutual information, i.e.,
\begin{equation}
\bJ(P_{UTXY}):=\bbE[\bj(U,X,Y|T)]=\begin{bmatrix}
H(X|U,T) \\ I(U;Y|T)
\end{bmatrix} .
\end{equation}
The mutual information $I(U;Y|T)=I(U, T;Y)$ because $T$ and $Y$ are independent.
\begin{definition}
The {\em entropy-information dispersion matrix}  for the WAK problem for a fixed $P_{UTXY}\in\tilde{\scP}(P_{XY})$  is defined as 
\begin{align}
\bV(P_{UTXY} ) &:=\bbE_T \left[ \cov( \bj(U,X,Y|T) ) \right] \label{eqn:defV} \\
&=\sum_{t \in\calT} P_T(t) \cov  ( \bj(U,X,Y|t) ) .
\end{align}
\end{definition}
We abbreviate the deterministic quantities $\bJ(P_{UTXY}) \in \bbR_+^2$ and $\bV(P_{UTXY})\succeq 0$ as $\bJ$ and $\bV$ respectively when the distribution $P_{UTXY}\in\tilde{\scP}(P_{XY})$ is obvious from the context.  
\begin{definition}
If $\bV(P_{UTXY})\ne \bzero_{2\times 2}$,  define $\scR_{\mathrm{in}}(n,\veps; P_{UTXY})$ to be   the set of rate pairs $(R_1,R_2)$ such that $\bR := [R_1, R_2]^T$ satisfies
\begin{equation} \label{eqn:Rin}
  \bR\in\bJ + \frac{\scS(\bV,\veps)}{\sqrt{n}}+ \frac{2\log n}{n}\bone_2 .
\end{equation}
If $\bV(P_{UTXY})= \bzero_{2\times 2}$,  define $\scR_{\mathrm{in}}(n,\veps; P_{UTXY})$ to be  the set of rate pairs $(R_1,R_2)$ such that
\begin{equation}
  \bR\in\bJ+  \frac{2\log n}{n}\bone_2 .
\end{equation}
\end{definition}
From  the simplified CS-type bound for the WAK problem in Corollary~\ref{cor:fein}, we can derive the following:
\begin{theorem}[Inner Bound to $(n,\veps)$-Optimal Rate Region] \label{thm:second}
For every $0<\veps<1$ and all $n$ sufficiently large, the $(n,\veps)$-optimal rate region $\scR_{\WAK}(n,\veps)$ satisfies
\begin{equation} \label{eqn:ib_wak1}
\bigcup_{P_{UTXY}\in\tilde{\scP}(P_{XY})} \scR_{\mathrm{in}} (n,\veps; P_{UTXY}) \subset\scR_{\WAK}(n,\veps).
\end{equation}
Furthermore, the union over $P_{UTXY}$ can be restricted to those
 distributions for which the supports $\calU$ and $\calT$ of auxiliary random
 variables $U$ and $T$ satisfy that 
$\lvert\calU\rvert\leq \lvert\calY\rvert+4$ and
 $\lvert\calT\rvert\leq 5$ respectively.
\end{theorem}
From the modified CS-type bound for the WAK problem in Theorem~\ref{thm:dt-helper-binning}, we can derive the following:
\begin{theorem}[Modified Inner Bound to $(n,\veps)$-Optimal Rate Region] \label{thm:second-modified}
For every $0<\veps<1$ and all $n$ sufficiently large, the $(n,\veps)$-optimal rate region $\scR_{\WAK}(n,\veps)$ satisfies
\begin{equation} \label{eqn:ib_wak2}
\bigcup_{P_{UTXY}\in\tilde{\scP}(P_{XY})} \scR_{\mathrm{in}}^\prime (n,\veps; P_{UTXY}) \subset\scR_{\WAK}(n,\veps),
\end{equation}
where $\scR_{\mathrm{in}}^\prime (n,\veps; P_{UTXY})$ is the set defined by replacing \eqref{eqn:Rin} with
\begin{equation}
 \bR\in\bigcup_{\rho\ge 0} \,\left\{\bJ + \frac{\scS(\bV,\veps) + [\rho,-\rho]^T}{\sqrt{n}}+ \frac{2\log n}{n}\bone_2 \right\}.  
\end{equation}
\end{theorem}
\begin{remark}
We can also restrict the cardinalities $\lvert\calU\rvert$ and
$\lvert\calT\rvert$ of auxiliary random variables in Theorem
\ref{thm:second-modified} in the same way as in Theorem
\ref{thm:second}.  The bound in Theorem \ref{thm:second-modified} is at
least as tight as that in Theorem \ref{thm:second}, and the former is
strictly tighter than the latter for a fixed test channel. However, it
is not clear whether the improvement is strict or not when we take the
union over the test channels.
\end{remark}


By setting $T=Y=U=\emptyset$ and $R_2=0$ in Theorem \ref{thm:second-modified},\footnote{In fact, to be precise, we cannot derive Corollary \ref{cor:lossless} from Theorem \ref{thm:second} because there is the residual term $\frac{2\log n}{n}$ and we cannot set $R_2 = 0$. However, we can use Corollary~\ref{cor:fein} with $U=\emptyset$ to obtain Corollary~\ref{cor:lossless} easily. } we obtain a result first discovered by Strassen~\cite{strassen}.
\begin{corollary}[Achievable Second-Order Coding Rate for Lossless Source Coding] \label{cor:lossless}
Define the  {\em second-order coding rate for lossless source coding} to be 
\begin{equation}
\sigma (P_X,\veps) := \limsup_{n\to\infty} \sqrt{n} (R_X(n,\veps)-H(X)) 
\end{equation}
where $R_X(n,\veps)$ is the minimal rate of almost-lossless compression of source $P_X$ at blocklength $n$ with error probability not exceeding $\veps$. Then, 
\begin{equation}
\sigma(P_X,\veps)\le  \sqrt{\var( \log P_X(X))}Q^{-1}(\veps).
\end{equation}
\end{corollary}
It is well-known that the  result in Corollary~\ref{cor:lossless} is tight, i.e.,  $\sqrt{\var( \log P_X(X))}Q^{-1}(\veps)$ is indeed the second-order coding rate for lossless source coding~\cite{strassen, Hayashi08, Kot97}. 

We refer to the reader to Appendix~\ref{app:second} for the proof of Theorem~\ref{thm:second} (Appendix~\ref{app:second-modified} for the proof of Theorem~\ref{thm:second-modified}). The proof is based on the CS-type bound in~\eqref{eqn:fein} and   the non-i.i.d.\ version of the multidimesional Berry-Ess\'een theorem by G\"oetze~\cite{Got91}. The proof of the cardinality bounds is provided in Appendix~\ref{app:cardinality_bound}.  The interpretation of this result is clear: From \eqref{eqn:Rin} which is the non-degenerate case, we see that the second-order coding rate region for a fixed $P_{UTXY}$ is represented by the set ${\scS(\bV(P_{UTXY}),\veps)}/{\sqrt{n}}$. Thus, the $(n,\veps)$-optimal rate region converges to the asymptotic WAK region at a rate of $O(1/\sqrt{n})$ which can be predicted by the central limit theorem. More importantly, because   our finite blocklength bound in \eqref{eqn:fein}  treats both the covering and binning error events {\em jointly}, this results in the coupling of the second-order rates through the set $\scS(\bV(P_{UTXY}),\veps)$ and hence, the dispersion matrix $\bV(P_{UTXY})$. This shows that the correlation between the entropy and information densities matters in the determination of the second-order coding rate. 

More specifically, Theorems~\ref{thm:second} and \ref{thm:second-modified} are proved by taking $P_{U^n|Y^n}(u^n|y^n)$ to be equal to $P_{U|TY}^n(u^n|t^n, y^n)$ for some fixed (time-sharing) sequence $t^n \in\calT^n$  and some  joint distribution $P_{UTXY}\in\tilde{\scP}(P_{XY})$. If $\calT=\emptyset$, this is essentially using i.i.d.\ codes. Theorems~\ref{thm:second} and \ref{thm:second-modified}  also show that $|\calT|$ can be upper bounded by $5$. An alternative to this proof strategy is to use conditionally constant composition codes as was done in Kelly-Wagner~\cite{Kel12} to prove their error exponent result. The advantage of this strategy is that it may yield better dispersion matrices because the unconditional dispersion matrix always dominates the  conditional dispersion matrix~\cite[Lemma~62]{PPV10} (in the partial order induced by  semi-definiteness). For using conditionally constant composition codes, we fix a conditional type $V_{Q_Y} \in \scV_n(\calU; Q_Y)$ for every marginal type $Q_Y\in\scP_n(\calY)$. Then, codewords are generated uniformly at random from $\calT_{V_{Q_Y}}(y^n)$ if $y^n\in\calT_{Q_Y}$. However, it does not appear that this strategy yields improved second-order coding rates compared to using  i.i.d.\ codes as given in Theorems~\ref{thm:second} and \ref{thm:second-modified}.

We emphasize here that the restriction of the sizes of the alphabets $\calU$ and $\calT$ only allows us to only preserve the {\em second-order region} defined by the vector $\bJ(P_{UTXY} )$ and the matrix $\bV(P_{UTXY} )$ over all $P_{UTXY}\in\tilde{\scP}(P_{XY})$. An optimized third-order term in \eqref{eqn:Rin} might be dependent on higher-order statistics of the entropy-information density vector $\bj(U,X,Y|T)$ and the quantities that define this third-order term are {\em not preserved}  by the bounds $\lvert\calU\rvert\leq \lvert\calY\rvert+4$ and
 $\lvert\calT\rvert\leq 5$. This remark is also applicable to the second-order rate regions for WZ and GP in Subsections~\ref{sec:2nd_wz} and~\ref{sec:2nd_gp}. However, we note that for lossless source coding~\cite{strassen} or channel coding~\cite{PPV10,TomTan}, under some regularity conditions, the third-order term is neither  dependent on higher-order statistics nor on the alphabet sizes.

To compare our Theorems~\ref{thm:second} and~\ref{thm:second-modified} to that of Verd\'u~\cite{Ver12}, for a fixed $P_{UXY}\in\scP(P_{XY})$,  define $\scR_{\mathrm{in}}^{\mathrm{V}}(n,\veps; P_{UXY})$ to be the set of rate pairs that satisfy 
\begin{align}
R_1 &  \ge  H(X|U) +  \sqrt{ \frac{ V_H(X|U) }{n} }Q^{-1}(\lambda\veps) + \frac{2\log n}{n}  \label{eqn:R1_split}\\
R_2 &  \ge  I(U;Y) +  \sqrt{ \frac{ V_I(U;Y) }{n} }Q^{-1}((1-\lambda)\veps) + \frac{2\log n}{n} \label{eqn:R2_split}
\end{align}
for some $\lambda \in [0, 1]$ where  the marginal entropy and information dispersions are defined as 
\begin{align}
V_H(X|U):=\var \left( \log \frac{1}{P_{X|U}(X|U)}\right) \\
V_I(U;Y) :=\var\left( \log \frac{P_{Y|U}(Y|U)}{ P_Y(Y)} \right)
\end{align}
respectively. Note that if $T=\emptyset$, then $V_H(X|U)$ and $V_I(U;Y)$ are the diagonal elements of the matrix $\bV(P_{UTXY})$ in \eqref{eqn:defV}.   It can easily be seen that Verd\'u's bound on the error probability of the WAK problem~\eqref{eqn:ver_bd} yields the following inner bound on  $\scR_{\WAK}(n,\veps)$. 
\begin{equation}
\bigcup_{P_{UXY}\in\scP(P_{XY})}  \scR_{\mathrm{in}}^{\mathrm{V}}(n,\veps; P_{UXY})\subset\scR_{\WAK}(n,\veps). \label{eqn:verdu_ib_wak}
\end{equation}
This ``splitting'' technique  of $\veps$ into $\lambda\veps$ and $(1-\lambda)\veps$ in \eqref{eqn:R1_split} and \eqref{eqn:R2_split}   was used by MolavianJazi  and Laneman~\cite{Mol12}  in their work on finite blocklength analysis for the MAC.  In    Section~\ref{sec:numerical}, we  numerically  compare the inner bounds for the WAK problem provided in \eqref{eqn:ib_wak1}, \eqref{eqn:ib_wak2} and \eqref{eqn:verdu_ib_wak}.

\begin{remark} \label{remark:second-order-achievability-corner-point}
From the non-asymptotic bound in Remark~\ref{remark:non-asymtptic-bound-for-corner}, we can also show that 
\begin{eqnarray}
\hat{\scR}_{\mathrm{in}}(n,\veps) \subset \scR_{\WAK}(n,\veps),
\end{eqnarray}
where $\hat{\scR}_{\mathrm{in}}(n,\veps)$ is the set of rate pairs $(R_1,R_2)$ such that
\begin{eqnarray}
\begin{bmatrix}
 R_1 \\ R_1 + R_2
\end{bmatrix} \in 
\begin{bmatrix}
 H(X|Y) \\ H(X,Y)
\end{bmatrix} 
 + \frac{\scS(\hat{\bV},\veps)}{\sqrt{n}}+ \frac{2\log n}{n}\bone_2
\end{eqnarray}
for the covariance matrix
\begin{eqnarray}
\hat{\bV} = \cov  \left( \begin{bmatrix}
 -\log P_{X|Y}(X|Y) \\ - \log P_{XY}(X,Y)
\end{bmatrix}  \right).
\end{eqnarray}
\end{remark}

\subsection{Achievable Second-Order Coding Rates for the WZ problem} \label{sec:2nd_wz}
In this section, we leverage on the simplified CS-type bound in Corollary~\ref{thm:fein_wz} to derive an achievable second-order coding rate for the WZ problem. We do so by first finding an inner bound to the $(n,\veps)$-Wyner-Ziv rate-distortion region $\scR_{\WZ}(n,\veps)$  defined in \eqref{eqn:ne_rd}. Subsequently we find an upper bound to the  $(n,\veps)$-Wyner-Ziv  rate-distortion  function $R_{\WZ}(n,\veps)$  defined in \eqref{eqn:ne_wz}. We also show that the (direct part of the) dispersion of lossy source coding found by Ingber-Kochman~\cite{ingber11} and Kostina-Verd\'u~\cite{kost12} can be recovered from the CS-type bound in Corollary~\ref{thm:fein_wz}. This is not unexpected because the lossy source coding (rate-distortion) problem  is a special case of the Wyner-Ziv problem where the side-information is absent.

We will again employ the ``time-sharing'' strategy used in
Section~\ref{sec:2nd_wak} and show that the cardinality of the time-sharing alphabet $\calT$ can be bounded. Note again that in the finite-blocklength
setting $\scR_{\WZ}(n,\veps)$  does not have to be convex, unlike in the
asymptotic setting. For fixed finite sets $\calU$ and $\calT$, let
$\tilde{\scP}(P_{XY})$ be the collection of all joint distributions
$P_{UTXY} \in \scP(\calU\times\calT\times\calX\times\calY)$ such that
the $\calX\times\calY$-marginal of $P_{UTXY}$ is $P_{XY}$, $U-(X,T)-Y$
forms a Markov chain and $T$ is independent of $(X,Y)$. 
A pair $(P_{UTXY}, P_{\hatX|UYT})$ of a joint distribution
$P_{UTXY}\in\tilde{\scP}(P_{XY})$ and a reproduction channel
$P_{\hatX|UYT}\colon\calU\times\calY\times\calT\to\hatcalX$ defines 
a joint distribution $P_{UTXY\hatX}$ such that
\begin{align}
& P_{UTXY\hatX}(u,t,x,y,\hatx) \nn\\
&=P_{XY}(x,y)P_T(t)P_{U|YT}(u|y,t)P_{\hatX|UYT}(\hatx|u,y,t).
\end{align}
Further, a pair of $P_{UTXY}\in\tilde{\scP}(P_{XY})$ and $P_{\hatX|UYT}$ induces
a random variable 
\begin{equation}
 \rvd(X,\hatX|T):=\rvd(X_T,\hatX_T)
\end{equation}
where $(X_t,\hatX_t)$ for any $t\in\calT$ has distribution
$P_{X\hatX|T=t}$. In other words, for fixed $t\in\calT$, 
\begin{align}
& \Pr\{\rvd(X,\hatX|T=t)=d\} \nn\\
& =\sum_{\substack{x, \hatx:\\ \rvd(x,\hatx)=d}}\sum_{u,y}P_{XY}(x,y)P_{U|YT}(u|y,t)P_{\hatX|UYT}(\hatx|u,y,t).
\end{align}

\begin{definition}
For a pair $(P_{UTXY}, P_{\hatX|UYT})$ of $P_{UTXY}\in\tilde{\scP}(P_{XY})$ and $P_{\hatX|UYT}$, the {\em information-density-distortion  vector} for the WZ problem is defined as 
\begin{equation}
\bj(U,X,Y,\hatX|T):=\begin{bmatrix}
-\log \frac{P_{Y|UT}(Y|U,T)}{P_{Y}(Y)}\\
\log \frac{P_{X|UT}(X|U,T)}{P_{X}(X)}\\
\rvd(X, \hatX|T) \\
\end{bmatrix}. \label{eqn:id_wz}
\end{equation}
\end{definition}

Since $\bbE[\rvd(X,\hatX)]=\sum_{t}P_{T}(t)\bbE_{P_{X\hatX|T}}[\rvd(X_T,\hatX_T)|T=t]$,
the expectation of information-density-distortion  vector  is given by 
\begin{align}
\bJ( P_{UTXY}, P_{\hatX|UYT})
&:=\bbE [\bj(U,X,Y,\hatX|T)] \\
&= \begin{bmatrix}
-I(U;Y|T)\\ I(U;X|T) \\ \bbE [\rvd(X,\hatX)]
\end{bmatrix}. \label{eqn:wz_mi}
\end{align}
Observe that the {\em sum} of the first two components of \eqref{eqn:wz_mi} resembles the Wyner-Ziv rate-distortion function defined in \eqref{eqn:wz_rd_func}. As such when stating an achievable $(n,\veps)$-Wyner-Ziv rate-distortion region, we    project the first two terms onto an affine subspace representing their sum.  See \eqref{eqn:Rin_wz} and \eqref{eqn:Rin_wz2} below. 
\begin{definition}
The {\em information-distortion dipersion matrix} for the WZ problem for
 a pair of  $P_{UTXY} \in \tilde{\scP}(P_{XY})$  and $P_{\hatX|UYT}$ is defined as 
\begin{equation}
\bV(P_{UTXY}, P_{\hatX|UYT}) : =\bbE_T \left[  \cov( \bj(U,X,Y,\hatX|T) )\right].
\end{equation}
\end{definition}
\begin{definition}
Let $\bM \in \bbR^{2\times 3}$ be the matrix
\begin{equation}  \label{eqn:defM}
\bM:=\begin{bmatrix}
1 & 1 & 0\\
0 & 0 & 1
\end{bmatrix}.
\end{equation}
If $\bV(P_{UTXY}, P_{\hatX|UYT})\ne \bzero_{3\times 3}$, define $\scR_{\mathrm{in}}(n,\veps; P_{UTXY},P_{\hatX|UYT})$ to be the set of all  rate-distortion pairs $(R,D)$ satisfying
\begin{equation} \label{eqn:Rin_wz}
\begin{bmatrix}
R \\ D
\end{bmatrix} \in \bM\left(\bJ  + \frac{ \scS(\bV ,\veps) }{\sqrt{n}} + \frac{2\log n}{n} \bone_3\right). 
\end{equation}
where $\bJ := \bJ( P_{UTXY},P_{\hatX|UYT})$ and $\bV := \bV( P_{UTXY},P_{\hatX|UYT})$. Else if $\bV(P_{UTXY}, P_{\hatX|UYT})\ne \bzero_{3\times 3}$, define $\scR_{\mathrm{in}}(n,\veps; P_{UTXY},P_{\hatX|UYT})$ to be the set of all  rate-distortion pairs $(R,D)$ satisfying 
\begin{equation}\label{eqn:Rin_wz2}
\begin{bmatrix}
R \\ D
\end{bmatrix} \in \bM\left(\bJ  +   \frac{2\log n}{n} \bone_3 \right). 
\end{equation}
\end{definition}
In \eqref{eqn:Rin_wz}, the matrix $\bM$ serves  project the three-dimensional set $\bJ + { \scS(\bV ,\veps) }/{\sqrt{n}} \subset\bbR^3$ onto two dimensions by linearly combining the first two mutual information terms to give $I(U;X|T)- I(U;Y|T)  = I(U;X|Y,T)$ (by the Markov chain $U-(X,T)-Y$).  From  the simplified CS-type bound for the WZ problem in Corollary~\ref{thm:fein_wz} and the multidimensional Berry-Ess\'een theorem~\cite{Got91}, we can derive the following:
\begin{theorem}[Inner Bound to the $(n,\veps)$-Wyner-Ziv Rate-Distortion Region] \label{thm:wz_2nd}
For every $0<\veps<1$ and all $n$ sufficiently large, the $(n,\veps)$-Wyner-Ziv rate-distortion region $\scR_{\WZ}(n,\veps)$ satisfies
\begin{align}
&\bigcup_{P_{UTXY}\in\tilde{\scP}(P_{XY}), P_{\hatX|UYT}} \scR_{\mathrm{in}}(n,\veps; P_{UTXY},P_{\hatX|UYT}) \nn\\
&~~~~~~~~~~~~~~~~~~~~~~~\subset \scR_{\WZ}(n,\veps). \label{eqn:inner_wz_rd}
\end{align}
Furthermore, the union over a pair of $P_{UTXY}$ and $P_{\hatX|UYT}$ can be restricted to those
 distributions for which the supports $\calU$ and $\calT$ of auxiliary random
 variables $U$ and $T$ satisfy that 
$\lvert\calU\rvert\leq \lvert\calX\rvert+8$ and
 $\lvert\calT\rvert\leq 9$ respectively.
\end{theorem}
\begin{remark} \label{rem:rep_ch_det}
The assumption that the reproduction channel $P_{\hatX|UTX}$ is
 stochastic is used to establish bounds on the cardinalities of the
 auxiliary random variables $U$ and $T$ (see Remark \ref{remark:support_lemma}).  This is because even though the functional representation lemma~\cite[Appendix~A]{elgamal} ensures that the first two entries of $\bj(u,x,y,\hatx|t)$  in \eqref{eqn:id_wz} are preserved using a deterministic reproduction channel and appropriate bounds on $|\calU |$ and $|\calT|$, the last entry concerning the distortion $\rvd(x,\hatx|t)$ may not be preserved using the same techniques. 
\end{remark}
The proof of this result is provided in Appendix~\ref{app:wz_2nd}. Further projecting onto the first dimension (the rate) for a fixed distortion level $D$ yields the following:
\begin{theorem}[Upper  Bound to the $(n,\veps)$-Wyner-Ziv Rate-Distortion Function] \label{thm:wz_rd_2nd}
For every $0<\veps<1$ and all $n$ sufficiently large, the $(n,\veps)$-Wyner-Ziv rate-distortion function $R_{\WZ}(n,\veps,D)$ satisfies 
\begin{align}
R_{\WZ}(n,\veps,D) \le& \inf\bigg\{R : (R,D) \in \bigcup_{P_{UTXY}\in\tilde{\scP}(P_{XY} ),P_{\hatX|UYT} } \nn\\
&~~~~~~ \scR_{\mathrm{in}}(n,\veps; P_{UTXY},P_{\hatX|UYT}) \bigg\}.
\end{align}
\end{theorem}
Theorems~\ref{thm:wz_2nd}  and~\ref{thm:wz_rd_2nd} are very similar in spirit to the result on the achievable second-order coding rate for the WAK problem. The marginal contributions from the distortion error event, the packing error event, the covering error event as well as their correlations are  all involved in the   dispersion matrix $\bV( P_{UTXY},P_{\hatX|UYT})$. 

It is worth mentioning why  for the inner bound to the second-order region in Theorem~\ref{thm:wz_2nd}, we should, in general, employ  stochastic reproduction functions $ P_{\hatX|UYT}$ instead of  a deterministic ones $g:\calU\times\calY\to\hcalX$. The reasons are twofold: First, this is to facilitate the bounding of
the cardinalities of the auxiliary alphabets $\calU$ and $\calT$  in Theorem~\ref{thm:wz_2nd}. This is done using variants of the support lemma~\cite[Appendix~A]{elgamal}. See Lemma~\ref{support_lemma_U} and~\ref{support_lemma_T} in Appendix~\ref{app:cardinality_bound}. The preservation of the expected distortion $\bbE\rvd(X,\hatX)$ requires that $P_{\hatX|UYT}$ is stochastic. See Theorem~\ref{thm:cardinality_bound_WZ} in Appendix~\ref{app:cardinality_bound}. Second, and more importantly,  it is not {\em a priori} clear without a converse (outer) bound on $ \scR_{\WZ}(n,\veps)$ that the second-order inner bound we have in~\eqref{eqn:inner_wz_rd} cannot be enlarged via the
use of a stochastic reproduction function $P_{\hatX | U Y T}$.   The same observation holds verbatim  for the GP problem where we use $P_{X|US}$ instead of a deterministic encoding function from $\calU\times\calS$ to $\calX$.


At this juncture, it is natural to wonder whether we are able to recover the dispersion for lossy source coding~\cite{ingber11,kost12} as a special case of Theorem~\ref{thm:wz_rd_2nd} (like Corollary~\ref{cor:lossless} is a special case of Theorem~\ref{thm:second-modified}). This does not seem straightforward because of the distortion error event in~\eqref{eqn:fein_wz}.  However, we can start from the CS-type bound in~\eqref{eqn:fein_wz}, set $Y= \emptyset$,  $U=\hatX$ and use the method of types~\cite{Csi97} or the notion of the $D$-tilted information~\cite{kost12} to obtain the specialization for the direct part. Before stating the result, we define a few quantities. 
Let  the rate-distortion function of the source $X\sim Q\in\scP(\calX)$ be denoted as 
\begin{equation}
R(Q,D) :=\min_{P_{\hatX,X}: P_X=Q,\bbE \rvd(X,\hatX)\le D}I(X;\hatX) ,\label{eqn:rdfun}
\end{equation}
where $\bbE\rvd(X,\hatX):=\sum_{ x,\hatx}P_{\hatX,X}(\hatx,x)\rvd(x,\hatx)$.  Also, define the {\em $D$-tilted information}  to be 
\begin{equation} \label{eqn:Dtilt}
j(x,D) :=-\log \bbE \left[ \exp\left( \lambda^*D - \lambda^*\rvd(x,\hatX^*\right)\right]
\end{equation}
where the expectation is with respect to the unconditional distribution of $\hatX^*$, the output distribution that optimizes the rate-distortion function in \eqref{eqn:rdfun} and 
\begin{equation}
\lambda^* := -\frac{\partial}{\partial D} R(P_X, D).
\end{equation}

\begin{theorem}[Achievable Second-Order Coding Rate for Lossy Source Coding] \label{thm:rd_disp}
Define the  {\em second-order coding rate for lossy source coding} to be 
\begin{equation}
\sigma (P_X,D,\veps) := \limsup_{n\to\infty} \sqrt{n} (R_X(n,\veps;D)-R(P_X,D)) 
\end{equation}
where $R_X(n,\veps;D)$ is the minimal rate of compression of source $X\sim P_X$ up to distortion $D$ at blocklength $n$ and probability of excess distortion not exceeding $\veps$.
We have
\begin{equation}
\sigma (P_X,D,\veps)  \le \sqrt{\var(j(X,D))}Q^{-1}(\veps)
\end{equation}
\end{theorem}
Two proofs of Theorem~\ref{thm:rd_disp} are provided in Appendix~\ref{app:2nd_wz_rd}, one based on  the method of types and the other based on the  $D$-tilted information in \eqref{eqn:Dtilt}.  For the former proof based on the method of types, we need to assume that $Q\mapsto R(Q,D)$ is differentiable in a small neighborhood of $P_X$ and $P_X$ is supported on a finite set.  For the second proof, $\calX$ can be an abstract alphabet. Note that $R(P_X,D)=\bbE_{X\sim P_X}[ j(X,D)]$. We remark that for discrete memoryless sources,  the  $D$-tilted information $j(x,D)$ coincides with the  derivative of the  rate-distortion function with respect to the source~\cite{ingber11}
\begin{equation}
R'(x,D)=\frac{\partial}{\partial Q(x)} R(Q,D) \bigg|_{Q=P_X} \label{eqn:der_rd}. 
\end{equation} 


\subsection{Achievable Second-Order Coding Rates for the GP problem} \label{sec:2nd_gp}
We conclude this section by stating and achievable second-order coding rate for the GP problem by presenting  a lower bound to the $(n,\veps,\Gamma)$-capacity $C_{\GP}(n,\veps,\Gamma)$ defined in \eqref{eqn:cne_gp}. As in the previous two subsections, we start with definitions. For two finite sets $\calU$ and $\calT$, define $\tilde{\scP}(W, P_S) $ to be the collection of all $P_{UTSXY}\in\scP  (\calU\times \calT\times\calS\times\calX\times\calY)$ such that  the $\calS$-marginal of $P_{UTSXY}$ is $P_S$, $P_{Y|XS}=W$, $U-(X,S,T)-Y$ forms a Markov chain and $T$  is independent of $S$.  Note that $P_{UTSXY}$ does not necessarily have to satisfy the cost constraint in \eqref{eqn:gpcost2}.

In addition, to facilitate the time-sharing for the cost function, we define
\begin{equation}
\rvg(X|T):=\rvg(X_T)
\end{equation}
where $X_t$ for any $t\in \calT$ has distribution $P_{X|T=t}$. 

\begin{definition}
The {\em information-density-cost vector} for the GP problem for $P_{UTSXY} \in \tilde{\scP}(W, P_S) $ is defined as 
\begin{equation}
\bj(U,S,X,Y|T):=\begin{bmatrix}
\log\frac{P_{Y|UT}(Y|U,T)}{P_{Y|T}(Y|T)} \\
-\log\frac{P_{S|UT}(S|U,T)}{P_S(S)}  \\
-\rvg(X|T)
\end{bmatrix}.
\end{equation}
\end{definition}
Since $\sum_t P_T(t)\bbE_{P_{X|T}}[\rvg(X_T)|T=t]=\bbE[\rvg(X)]$, the expectation of this vector with respect to $P_{UTSXY}$ is the vector of mutual informations and the  negative cost, i.e.,
\begin{equation}
\bJ(P_{UTSXY}):=\bbE[\bj(U,S,X,Y|T)]= \begin{bmatrix}
I(U;Y|T) \\- I(U;S|T)\\ -\bbE[\rvg(X)]
\end{bmatrix}.
\end{equation}
\begin{definition}
The {\em information-dispersion matrix} for the GP problem for $P_{UTSXY} \in \tilde{\scP} (W, P_S) $ is defined as 
\begin{equation}
\bV(P_{UTSXY}) := \bbE_T [ \cov(  \bj(U,S,X,Y|T)  ) ] .
\end{equation}
\end{definition}
\begin{definition}
Let $\bM$ be the matrix defined in~\eqref{eqn:defM}.  If $\bV(P_{UTSXY})\ne \bzero_{3\times 3}$, define the set $\scR_{\mathrm{in}}(n,\veps; P_{UTSXY})$ to be the set of all rate-cost pairs $(R,\Gamma)$ satisfying 
\begin{equation} \label{eqn:Rin_gp}
\begin{bmatrix} R\\ -\Gamma\end{bmatrix}\in \bM\left(\bJ - \frac{\scS(\bV,\veps)}{\sqrt{n}}-\frac{2\log n}{n}\bone_3\right)
\end{equation}
where $\bJ := \bJ( P_{UTSXY})$ and $\bV := \bV( P_{UTSXY})$. Else if $\bV(P_{UTXY}, g)\ne \bzero_{3\times 3}$, define $\scR_{\mathrm{in}}(n,\veps; P_{UTSXY})$ to be the set of all rate-cost pairs $(R,\Gamma)$ satisfying 
\begin{equation}\label{eqn:Rin_gp2}
\begin{bmatrix} R\\ -\Gamma\end{bmatrix}\in \bM\left(\bJ  - \frac{2\log n}{n}\bone_3\right). \end{equation} 
\end{definition}
By leveraging on our finite blocklength CS-type bound for the GP problem in~\eqref{eqn:fein_gp}, we obtain the following:

\begin{theorem}[Inner Bound to the $(n,\veps)$-GP Capacity-Cost Region] \label{thm:gp_disp} For every $0<\veps<1$ and all $n$ sufficiently large, the $(n,\veps)$-GP capacity-cost region $\scC_{\GP}(n,\veps)$ satisfies
\begin{equation}
\bigcup_{P_{UTSXY} \in  \tilde{\scP} (W, P_S)}\scR_{\mathrm{in}}(n,\veps; P_{UTSXY})\subset\scC_{\GP}(n,\veps).
\end{equation}
Furthermore, the union over $P_{UTSXY}$ can be restricted to those
 distributions for which the supports $\calU$ and $\calT$ of auxiliary random
 variables $U$ and $T$ satisfy that 
$\lvert\calU\rvert\leq \lvert\calS\rvert\lvert\calX\rvert+6$ and
 $\lvert\calT\rvert\leq 9$ respectively.
\end{theorem}
The assumption that the encoding function $P_{X|US}$ is stochastic appears to be necessary for establishing bounds on $|\calU|$ and $|\calT|$. See Remark \ref{rem:rep_ch_det}.  By projecting onto the first dimension (the rate) for a fixed cost $\Gamma\ge 0$, we obtain:
\begin{theorem}[Lower Bound to the $(n,\veps)$-GP Capacity] \label{thm:gp_disp2}
For every $0<\veps<1$ and all $n$ sufficiently large,  the $(n,\veps)$-GP capacity-cost function  $C_{\GP}(n,\veps,\Gamma)$ satisfies
\begin{align} \label{eqn:Cne_bound}
C_{\GP}(n,\veps,\Gamma)\ge &  \sup\bigg\{ R:(R,\Gamma)\in\bigcup_{P_{UTSXY} \in  \tilde{\scP} (W, P_S)} \nn\\
&~~~~~~~~~~ \scR_{\mathrm{in}}(n,\veps; P_{UTSXY})\bigg\}.
\end{align}
\end{theorem}
The proof of Theorem~\ref{thm:gp_disp}  parallels that for the WZ case in Theorem~\ref{thm:wz_2nd} so it is omitted for brevity.    The matrix $\bM$ serves to project the first two components of each element in the  set $\bJ + {\scS(\bV,\veps)}/{\sqrt{n}}$ onto one dimension. Indeed, for a fixed $P_{UTSXY} \in \tilde{\scP} (W, P_S)$, the first two components read $I(U;Y|T)-I(U;S|T)$ which, if $T=\emptyset$ and the random variables $(U,S,X,Y)$ are capacity-achieving, reduces  to the GP formula in~\eqref{eqn:gp_formula}. Hence,  the   set $ \bM{\scS(\bV,\veps)}/{\sqrt{n}} \subset\bbR$  quantifies all possible backoffs from the asymptotic GP capacity-cost region $\scC_{\GP} $ (defined in \eqref{eqn:gp_asymp}) at blocklength $n$ and average error probability $\veps$ based on our CS-type finite blocklength bound for the GP problem in~\eqref{eqn:fein_gp}.   The bound in \eqref{eqn:Cne_bound} is clearly much tighter than the one provided in~\cite{Tan12b} which is based on the use of Wyner's PBL and Markov lemma.


Now by setting $S=T=\emptyset$,  $U=X$ and $\Gamma=\infty$ in Theorem~\ref{thm:gp_disp2}, we recover the direct part of the  second-order coding rate  for channel coding without cost constraints~\cite{strassen, PPV10, Hayashi09}.

\begin{corollary}[Achievable Second-Order Coding Rate for Channel Coding] \label{cor:disp_cc}
Fix a  non-exotic~\cite{PPV10} discrete memoryless channel $W:\calX\to\calY$ with channel capacity $C(W) = \max_{P_X} I(X;Y)$. Define the {\em second-order coding rate for channel coding} to be 
\begin{equation}
\sigma(W,\veps):=\limsup_{n\to\infty} \sqrt{n}(C(W)-C_W(n,\epsilon))
\end{equation}
where $C_W(n,\epsilon)$ is the maximal rate of transmission over the channel $W$ at blocklength $n$ and average error probability $\veps$. Then, 
\begin{equation}
\sigma(W,\veps)\le\min_{P_{X^*}}\sqrt{\var\left( \log \frac{ W(Y^*|X^*) } {P_{Y^*} (Y^*)}\right)}Q^{-1}(\veps) \label{eqn:ach_cc}
\end{equation}
where $(X^*,Y^*)\sim P_{X^*}\times W$ and the minimization is over all  capacity-achieving input distributions.
\end{corollary}
The bound in \eqref{eqn:ach_cc} is has long been known to be an equality~\cite{strassen}. Note that the unconditional dispersion in~\eqref{eqn:ach_cc}  $\var\left( \log \frac{ W(Y^*|X^*) } {P_{Y^*} (Y^*)}\right)$ coincides with the conditional dispersion~\cite{PPV10} since it is being evaluated at a capacity-achieving input distribution.  As such, the converse can be proved using the meta-converse in~\cite{PPV10} or an   modification of the Verd\'u-Han converse~\cite[Lem.\ 3.2.2]{Han10} with an judiciously chosen output distribution as was done in~\cite{Hayashi09}. In fact, we can also derive a generalization of Corollary~\ref{cor:disp_cc} with cost constraints incorporated~\cite[Thm.~3]{Hayashi09} using  similar techniques as in the proof of Theorem~\ref{thm:rd_disp}. Namely, we use a uniform distribution over a particular type class (constant composition codes) as the input distribution. The type is chosen to be close to the optimal input distribution (assuming it is unique).
\section{Numerical Examples} \label{sec:numerical}

\subsection{Numerical Example for WAK Problem}

In this section, we use an example to illustrate the inner bound on $(n,\varepsilon)$-optimal rate region 
for the WAK problem obtained in Theorem \ref{thm:second}. We neglect the small $O\left(\frac{\log n}{n}\right)$ term.
The source is taken to be a discrete symmetric binary source DSBS($\alpha$), i.e., 
\begin{equation} \label{eq:dsbs}
P_{XY} = \frac{1}{2}\left[
\begin{array}{cc}
1 - \alpha & \alpha \\
\alpha & 1 - \alpha 
\end{array}
\right].
\end{equation}
In this case, the optimal rate region reduces to
\begin{align}
\scR_{\WAK}^* = \bigg\{ (R_1,R_2):& R_1 \ge h(\beta * \alpha), \nn\\
& R_2 \ge 1 - h(\beta),~0 \le \beta \le \frac{1}{2}  \bigg\},
\end{align}
where $h(\cdot)$ is the binary entropy function and $\beta * \alpha := \beta (1-\alpha) + (1-\beta) \alpha$ is the binary convolution.
The above region is attained by setting the backward test channel from $U$ to $Y$ to be a BSC with some crossover probability $\beta$.
All the elements in the entropy-information dispersion matrix $\bV(\beta)$ can be evaluated in closed form in terms of $\beta$. Define
$\bJ(\beta) := [h(\beta*\alpha), 1-h(\beta)]^T$. In Fig.~\ref{fig:comp}, we plot the second-order region
\begin{equation}\label{eqn:tilde_region}
\tilde{\scR}_{\mathrm{in}}(n,\varepsilon) := \bigcup_{0 \le \beta \le \frac{1}{2}} \left\{ (R_1,R_2) : \bR \in \bJ(\beta) + \frac{\scS(\bV(\beta),\veps)}{\sqrt{n}} \right\}.
\end{equation}
The first-order region $\scR_{\WAK}^*$ and the second-order region with simple time-sharing ($|{\cal T}| = 2$) are also shown for comparison.
More precisely, the simple time-sharing is between $\beta = 0$ and $\beta = 1/2$.
As expected, as the block length increases, the $(n,\varepsilon)$-optimal rate region tends to the first-order one. Interestingly, 
at small block length, time-sharing makes the second-order $(n,\varepsilon)$-optimal rate region in (\ref{eqn:tilde_region}) larger compared to
that without time-sharing. Especially, the simple time-sharing is better than $\tilde{\scR}_{\mathrm{in}}(n,\varepsilon)$ for $n=500$ because
the rank of the entropy-information dispersion matrix $\lambda \bV(0) + (1-\lambda) \bV(1/2)$ for $0 < \lambda \le 1$ is one.\footnote{It should be
noted that the rank of $\bV(1/2)$ is zero.}
\begin{figure}
\centering
\includegraphics[width=.93\columnwidth]{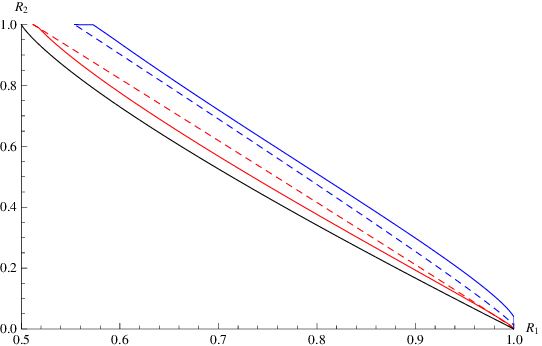} 
\caption{A comparison between $\tilde{\scR}_{\mathrm{in}}(n,\veps)$  without time-sharing (solid line) and the time-sharing region (dashed line) for $\veps=0.1$. The regions are to the top right of the curves. The blue and red curves are for $n = 500$ and $n = 10,000$ respectively. The black curve is the first-order region~\eqref{eqn:wak}.}
\label{fig:comp}
\end{figure}

We also consider the region $\tilde{\scR}_{\mathrm{in}}^{\mathrm{V}}(n,\veps)$ which is  the analogue of $\tilde{\scR}_{\mathrm{in}}(n,\veps)$ but derived from Verd\'u's bound in \eqref{eqn:ver_bd}. In Fig.~\ref{fig:second}, we compare the second-order coefficients, namely that derived from our bound $\scS(\bV(\beta),\veps)$ and 
\begin{align}
\scS^{\mathrm{V}}(\bV(\beta),\veps):=\bigcup_{0\le\lambda\le 1} & \Big\{ (z_1,z_2): z_1 \ge\sqrt{V_H(\beta)}Q^{-1}(\lambda\veps), \nn\\
& z_2 \ge \sqrt{V_I(\beta)}Q^{-1}((1-\lambda)\veps)\Big\}. \label{eqn:SVV}
\end{align}
Note that the difference between the two regions is quite small even for $\veps=0.5$. This is because, for this example, the covariance of the entropy- and information-density  (off-diagonal  in the dispersion matrix) is negative so the difference between $\Pr(Z_1\ge z_1 \mbox{ or } Z_2\ge z_2)$ and $\Pr(Z_1\ge z_1)+ \Pr(Z_2\ge z_2)$ is small. In this case, the $2$-dimensional Gaussian $\bZ\sim\calN(\bzero,\bV(\beta))$ has a negative covariance and hence the probability mass in the first and third quadrants are small. Hence, the union bound is not very loose in this case. 

\begin{figure}
\centering
\includegraphics[width=.93\columnwidth]{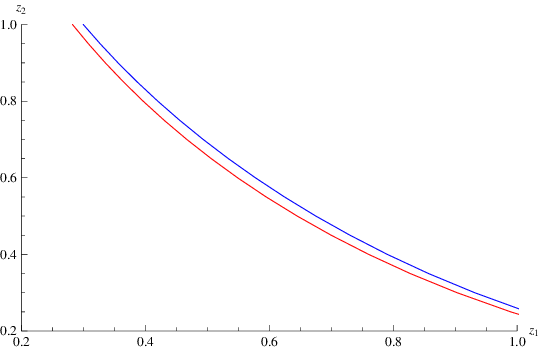} 
\caption{A comparison between $\scS(\bV(\beta), \veps)$ (defined in \eqref{eqn:Sv}) and $\scS^{\mathrm{V}}(\bV (\beta),\veps)$ (defined in~\eqref{eqn:SVV}) for $\beta = h^{-1}(0.5)$ and $\veps = 0.5$.  The red and blue curves are the boundaries  of  $\scS(\bV(\beta), \veps)$ and  $\scS^{\mathrm{V}}(\bV (\beta),\veps)$ respectively.  The regions lie to the top right of the curves.}
\label{fig:second}
\end{figure}

Next, we consider the binary joint source given by $P_{X|Y}(1|0) = P_{X|Y}(0|1) = \alpha$ and $P_Y(0) = p \le \frac{1}{2}$, which is a generalization of \eqref{eq:dsbs}. This example was investigated in \cite{GKEH:09ISIT}, and the optimal rate region reduces to
\begin{align}
\scR_{\WAK}^* = \Big\{ (R_1,R_2):& R_1 \ge h(\beta * \alpha), \nn\\
&R_2 \ge h( p) - h(\beta),~0 \le \beta \le p  \Big\}.
\end{align} 
The above region is attained by setting the backward test channel from $U$ to $Y$ to be BSC with some crossover probability $0 \le \beta \le p$.
All the elements in the entropy-information dispersion matrix $\bV(\beta)$ can be evaluated in closed form in terms of $\beta$. Define
$\bJ(\beta) := [h(\beta*\alpha), h(  p)-h(\beta)]^T$. In Fig.~\ref{fig:comp-asymmetric}, we plot the second-order region
\begin{equation}\label{eqn:tilde_region-asymmetric}
\tilde{\scR}_{\mathrm{in}}(n,\varepsilon) := \bigcup_{0 \le \beta \le p} \left\{ (R_1,R_2) : \bR \in \bJ(\beta) + \frac{\scS(\bV(\beta),\veps)}{\sqrt{n}} \right\}.
\end{equation}
For comparison, we also plot the second-order region derived from Remark~\ref{remark:second-order-achievability-corner-point}. Around the corner point defined by the entropies $[H(X|Y), H(Y)]^T = [h(\beta), h( p)]^T$, we find that the bound from Remark~\ref{remark:second-order-achievability-corner-point} is tighter than that given by~\eqref{eqn:tilde_region-asymmetric}.
\begin{figure}
\centering
\includegraphics[width=.93\columnwidth]{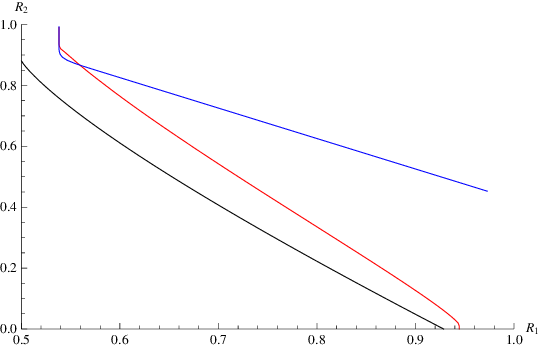} 
\caption{A comparison between $\tilde{\scR}_{\mathrm{in}}(n,\veps)$ (red solid curve) and the bound from Remark \ref{remark:second-order-achievability-corner-point} (blue solid curve) for $\veps=0.1$ and $n=1000$. The regions are to the top right of the curves. }
\label{fig:comp-asymmetric}
\end{figure} 
 

\subsection{Numerical Example for GP Problem}

In this section, we use an example to illustrate the inner bound on $(n,\varepsilon)$-optimal rate 
for the GP problem obtained in Theorem \ref{thm:gp_disp}. We do not consider cost constraints here, i.e., $\Gamma=\infty$.  We also neglect the small $O\left(\frac{\log n}{n}\right)$ term.
We consider the {\em memory with stuck-at faults} example \cite{heegard} (see also \cite[Example 7.3]{elgamal}). 
The state $S=0$ correspond to a faculty memory cell that output $0$ independent of the input value,
the state $S=1$ corresponds to a faculty memory cell that outputs $1$ independent of the input value,
and the state $S=2$ corresponds to a binary symmetric channel with crossover probability $\alpha$.
The probabilities of these states are $\frac{p}{2}$, $\frac{p}{2}$, and $1-p$ respectively. 

It is known \cite{heegard} that
the capacity is
\begin{equation} \label{gp-binary-capacity}
C_{\GP}^* = (1-p)(1 - h(\alpha)).
\end{equation}
The above capacity is attained by setting ${\cal U} = \{0,1\}$ and $P_{U|X}(0|0) = P_{U|S}(1|1) = 1-\alpha$,
$P_{U|S}(u|2) = \frac{1}{2}$, and $X = U$.
All the elements in the information dispersion matrix $\bV$ can be evaluated in closed form.
In Fig.~\ref{fig:gp-comp}, we plot the second-order capacity
\begin{align}
 \tilde{R}_{\GP}(n,\veps;p,\alpha) 
:=& (1-p)(1-h(\alpha)) \nn\\
&- \frac{1}{\sqrt{n}} \min\{ z_1 + z_2 : (z_1,z_2) \in \scS(\bV,\veps) \}.
\end{align}

For comparison, let us consider the case in which the decoder, instead of the encoder, can access the state $S$. In this case, we can regard $X$ as the 
channel input and $(S,Y)$ as the channel output. It is known \cite{heegard} that the capacity $C(W)$ of this channel is the same as \eqref{gp-binary-capacity}.
The dispersion $V$ can be evaluated in closed form by appealing to the law of total variance~\cite{ingber10}. In Fig.~\ref{fig:gp-comp}, we also plot the second order capacity
\begin{equation}
\tilde{C}(n,\varepsilon;p,\alpha) := (1-p)(1-h(\alpha)) - \sqrt{\frac{V}{n}} Q^{-1}(\varepsilon).
\end{equation}
From the figure, we can find that the lower bound $\tilde{R}_{\GP}(n,\veps;p,\alpha)$ on the GP $(n,\varepsilon)$-optimal rate
is smaller than the $(n,\varepsilon)$-optimal rate with decoder side-information
though the first order rates coincide.

\begin{figure}
\centering
\includegraphics[width=.93\columnwidth]{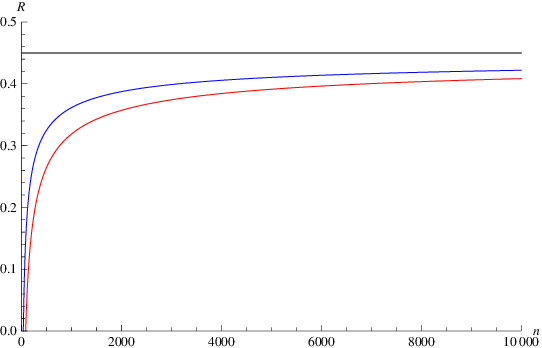} 
\caption{A comparison between $\tilde{R}_{\GP}(n,\veps;p,\alpha)$  (red solid line) and $\tilde{C}(n,\varepsilon;p,\alpha)$ (blue solid line) 
for $\veps=0.001$, $p=0.1$, and $\alpha = 0.11$. 
The black solid line is the first-order capacity~\eqref{gp-binary-capacity}.}
\label{fig:gp-comp}
\end{figure}

\section{Conclusion and Further Work} \label{sec:con}

\subsection{Summary}
In this paper, we proved   new non-asymptotic bounds on the error probability for side-information coding problems, including the WAK, WZ and GP problems.     These bounds then yield known general formulas as simple corollaries. In addition, we used these bounds to provide  achievable second-order coding rates for these three  side-information problems. We argued that when evaluated using i.i.d.\ test channels, the second-order rates evaluated using our non-asymptotic bounds are the best known in the literature  including~\cite{Ver12}. 

\subsection{Further Work on Non-Asymptotic and Second-Order Achievability Bounds }
Other challenging problems involving the derivation of non-asymptotic achievability bounds for multi-terminal problems include the Heegard-Berger \cite[Sec.~11.4]{elgamal}  problem,  multiple description coding \cite[Ch.~13]{elgamal}, Marton's inner bound for  the broadcast channel~\cite[Thm.~8.3]{elgamal},  and hypothesis testing with  multi-terminal data compression~\cite{Han87}.  Achievable second-order coding rate regions for some of these problems have been derived independently and concurrently by Yassaee-Aref-Gohari~\cite{YAG13a,YAG13b} using a completely different technique as discussed in the Introduction but it may be interesting to verify if the technique contained in this paper can be adapted to the above-mentioned coding problems.  

\subsection{Further Work on Non-Asymptotic and Second-Order Converse Bounds}
A natural question that arises from this work is whether one can derive non-asymptotic converse bounds that, when suitably specialized, coincide with the second-order achievability bounds in Section~\ref{sec:2nd}. Apart from the Slepian-Wolf problem~\cite{TK12, nomura} and the Gaussian MAC with degraded message sets \cite{ST13}, this has not been done for other problems in network information theory. Because second-order converse bounds imply the strong converse, it appears that first establishing a strong converse provides intuition for establishing   non-asymptotic converse bounds  that are   tight in the second-order sense after asymptotic evaluation. 

To the best of the authors' knowledge, there are only three approaches that may be used to obtain second-order converses for network problems whose first-order (capacity region) characterization involve  auxiliary random variables. The first is the information spectrum method. For example \cite[Lem.~2]{bouch00} provides a non-asymptotic converse bound for the asymmetric broadcast channel. However, the evaluation is not efficiently computable for large (or even moderate) $n$ as one has to perform an exhaustive search over the space of all $n$-letter auxiliary random variables (or equivalently $n$-letter joint distributions). The second is the entropy and image size characterization technique~\cite{Ahls76}  based on the blowing-up lemma~\cite[Ch.~5]{Csi97}.   This has been used to prove the strong converse for the WAK problem~\cite{Ahls76}  and the GP problem~\cite{tyagi}. However, the use of the blowing-up approach to obtain second-order converse bounds is not straightforward. The third method involves a non-standard change-of-measure argument and was used in the work of Kelly and Wagner~\cite[Thm.~2]{Kel12} to prove an upper bound on the error exponent for WAK coding. Again, it does not appear, at first glance, that this argument is amenable to second-order analysis.

\appendices

\section{Proof of Proposition \ref{proposition:cost-conversion} (Expurgated Code)} \label{proof:proposition:cost-conversion}
\begin{proof}
Let $x_0 \in \calX$ be a prescribed constant satisfying $\rvg(x_0) \le \Gamma$, and let $P_X^*$ be the distribution such that $P_X^*(x_0) = 1$, i.e., $P_X^*(x)=\bone [x=x_0 ]$.
Then, we define
\begin{align}
\tilde{P}_{X|MS}(x|m,s) :=& P_{X|MS}(x|m,s) \bone\left[ \rvg(x) \le \Gamma \right] \nn\\
& + P_{X|MS}\left(\calT_{\rmg}^{\GP}(\Gamma)^c|m,s\right) P_X^*(x).
\end{align}
Then, it is obvious that $\tilde{P}_X\left(\calT_{\rmg}^{\GP}(\Gamma)\right) = 1$. We also have
\begin{align}
& \tilde{P}_{MSXY\hatM}[m \neq \hatm] \nn \\
&= \sum_{\substack{m,\hatm \atop m \neq \hatm}} \sum_{s,x,y} P_M(m) P_S(s) \tilde{P}_{X|MS}(x|m,s) \nn\\
& ~~~~~~~~~~\times W(y|x,s) P_{\hatM|Y}(\hatm|y) \\
&= \sum_{\substack{   m,\hatm \\ m \ne \hatm }} \sum_{s,x,y} P_M(m) P_S(s) P_{X|MS}(x|m,s) \nn\\
& ~~~~~~~~~~\times W(y|x,s) P_{\hatM|Y}(\hatm|y) \bone\left[ \rvg(x) \le \Gamma \right]\nn \\
& \qquad + \sum_{\substack{  m,\hatm \\ m \ne \hatm }} \sum_{s,x,y} P_M(m) P_S(s) P_{X|MS}\left(\calT_{\rmg}^{\GP}(\Gamma)^c|m,s\right) \nn\\
& ~~~~~~~~~~\times P_X^*(x)
   W(y|x,s) P_{\hatM|Y}(\hatm|y) \\
&\le \sum_{\substack{  m,\hatm \\ m \ne \hatm }} \sum_{s,x,y} P_M(m) P_S(s) P_{X|MS}(x|m,s) \nn\\
& ~~~~~~~~~~\times W(y|x,s) P_{\hatM|Y}(\hatm|y) \bone\left[ \rvg(x) \le \Gamma \right]\nn \\
& \qquad+ \sum_{m,\hatm } \sum_{s,x,y} P_M(m) P_S(s) P_{X|MS}\left(\calT_{\rmg}^{\GP}(\Gamma)^c|m,s\right) \nn\\
& ~~~~~~~~~~\times P_X^*(x)
   W(y|x,s) P_{\hatM|Y}(\hatm|y) \\
&= \sum_{\substack{  m,\hatm \\ m \ne \hatm }} \sum_{s,x,y} P_M(m) P_S(s) P_{X|MS}(x|m,s) \nn\\
& ~~~~~~~~~~\times  W(y|x,s) P_{\hatM|Y}(\hatm|y) \bone\left[ \rvg(x) \le \Gamma \right]\nn \\
& \qquad + \sum_{m,s} P_M(m) P_S(s) P_{X|MS}\left(\calT_{\rmg}^{\GP}(\Gamma)^c|m,s\right) \\
&= P_{MSXY\hatM}\left[ \rvg(x) \le \Gamma \cap m \neq \hatm \right] + P_{MSXY\hatM}\left[ \rvg(x) > \Gamma \right] \\
&= P_{MSXY\hatM}\left[ \rvg(x) > \Gamma \cup m \neq \hatm \right]
\end{align}
as desired.
\end{proof}

\section{Channel Resolvability }\label{app:channel_res}
In this appendix, we review notations and known results for   channel
resolvability~\cite[Ch.\ 6]{Han10} \cite{HV93} \cite{Hayashi06} \cite{Cuff12}.

As a start, we first review the properties of the variational distance.  Let
$\scP'(\calU)$ be the set of all sub-normalized non-negative functions
(not necessarily probability distribution unless otherwise stated) on a finite
set $\calU$.  
Note that if $P\in\scP'(\calU)$ is normalized then $P\in\scP(\calU)$,
i.e., $P$ is a distribution on $\calU$.
For $P,Q\in\scP'(\calU)$, we define the variational
distance (divided by 2) as
\begin{align}
 d(P,Q)=\frac{1}{2}\sum_{u\in\calU}\lvert P(u)-Q(u)\rvert.
\end{align}
For two sets $\calU$ and $\calZ$, let $\scP'(\calZ|\calU)$
be the set of all sub-normalized non-negative functions indexed
by $u\in\calU$. When $W\in\scP'(\calZ|\calU)$
is normalized, it is a channel. In this section, we denote the joint distribution induced by $P\in\scP(\calU)$ and $W\in\scP'(\calZ|\calU)$ as $PW\in \scP'(\calU\times\calZ)$. 
 The following properties are useful in
the proof of theorems. Since the proofs are almost the same as well known properties of
the variational distance for normalized distributions, we omit the proofs.
\begin{lemma}
The variational distance satisfies the following properties.
\begin{enumerate}
 \item The monotonicity with respect to marginalization: For $P,Q\in\scP'(\calU)$ and $W,V\in\scP'(\calZ|\calU)$, 
let $P',Q'\in\scP'(\calZ)$ be
\begin{equation}
 P'(z):=\sum_{u\in\calU}P(u)W(z|u),~  Q'(z):=\sum_{u\in\calU}Q(u)V(z|u).
\end{equation}
Then,
\begin{equation}
 d(P',Q')\leq d(PW,QV).
\label{eqn:distance-property1}
\end{equation}
 \item The data-processing inequality: For $P,Q\in\scP'(\calU)$ and $W\in\scP'(\calZ|\calU)$, 
\begin{equation}
 d(PW,QW)\leq d(P,Q).
\label{eqn:distance-property2}
\end{equation}
In particular, when $W \in \scP(\calZ|\calU)$, the equality holds in \eqref{eqn:distance-property2}.
 \item For a distribution $P\in\scP(\calU)$, a sub-normalized measure
       $Q\in\scP'(\calU)$, and any subset $\Gamma\subset\calU$, 
\begin{equation}
 P(\Gamma)\leq Q(\Gamma)+d(P,Q)+\frac{1-Q(\calU)}{2}.
\label{eqn:distance-property3}
\end{equation}
\end{enumerate}
\end{lemma}

\begin{remark}
Combining \eqref{eqn:distance-property1} for $V=W$ and
 \eqref{eqn:distance-property2}, we have
\begin{equation}
 d(P',Q')\leq d(P,Q).
\label{eqn:distance-property4}
\end{equation}
Although the above inequality is usually referred as the data-processing
 inequality, we will use  \eqref{eqn:distance-property2} in the proofs
 of non-asymptotic bounds.
\end{remark}



Next, we introduce the concept of \emph{smoothing} of a distribution \cite{Renner05}.
For a distribution $P\in\scP(\calU)$ and a subset $\calT\subset\calU$,
a smoothed sub-normalized function $\barP$ of $P$ is derived by
\begin{equation}
 \barP(u) := P(u)\bone[u\in\calT].
\end{equation}
Note that 
the distance between the original distribution and a smoothed one is 
\begin{equation}
 d(P,\barP)=\frac{P(\calT^c)}{2}.
\label{eqn:distance-of-smoothing}
\end{equation}
Similarly, for a channel $W\colon\calU\to\calZ$ and a subset $\calT\subset\calU\times\calZ$, a
smoothed one $\barW\in\scP'(\calZ|\calU)$ is derived by 
\begin{equation}
 \barW(z|u) := W(z|u)\bone[(u,z)\in\calT]
\end{equation}
and it satisfies
\begin{equation}
 d(PW,P\barW)=\frac{PW(\calT^c)}{2},
\label{eqn:distance-of-smoothing2}
\end{equation}
where $PW \in \scP(\calU\times\calZ)$ is the joint distribution induced by $P$ and $W$. 

Now, we consider the problem of channel resolvability.
Let a channel $P_{Z|U}:\calU\to\calZ$ and an input distribution $P_U$ be given.
We would like to approximate the output distribution 
\begin{equation}
 P_Z(z)=\sum_{u\in\calU} P_U(u)P_{Z|U}(z|u)
\end{equation}
by using $P_{Z|U}$ and as small an amount of randomness as possible. 
This is done by 
means of a designing a deterministic map from a finite set $\calI$ to a codebook $\calC=\{u_i\}_{i\in\calI} \subset\calU$.  
For a given resolvability code $\calC$, let
\begin{equation}
 P_{\tilZ}(z)=\sum_{i\in\calI}\frac{1}{\sizeI}P_{Z|U}(z|u_i) \label{eqn:deftilZ}
\end{equation}
be the simulated output distribution. The approximation error is
evaluated by the distance $d(P_{\tilZ},P_Z)$.

We consider using the random coding technique as follows. We randomly and independently generate  codewords
$u_1,u_2,\dots,u_{\sizeI}$ according to $P_U$.
To derive an
upper bound on the averaged approximation error 
$\bbE_{\calC}\left[d(P_{\tilZ},P_Z)\right]$, it is convenient to consider a smoothing operation defined as follows.
For the set
\begin{equation}
 \calT_{\rmc}(\gamma_\rmc):=\left\{
(u,z):\log\frac{P_{Z|U}(z|u)}{P_Z(z)}\leq \gamma_\rmc
\right\},
\label{eqn:calT-for-smoothing}
\end{equation}
let
\begin{equation}
 \barP_{Z|U}(z|u):=P_{Z|U}(z|u)\bone[(u,z)\in\calT_\rmc(\gamma_\rmc)].
\label{eqn:smooth-P_Z|U}
\end{equation}
Moreover, for fixed resolvability code
$\calC=\{u_1,\dots,u_{\sizeI}\}$, let
\begin{equation}
 \barP_{\tilZ}(z):=\sum_{i\in\calI}\frac{1}{\sizeI}\barP_{Z|U}(z|u_i).
\end{equation}
Then, we have the following lemma known as {\em soft covering}, which is an improvement of \cite[Lemma 2]{Hayashi06}.

\begin{lemma}
[Corollary 7.2 of \cite{Cuff12}]\label{lem:resolvability}
 For any $\gamma_\rmc\ge0$, we have
\begin{align}
\bbE_{\calC}\left[d(\barP_{\tilZ},\barP_Z)\right]
  \leq \frac{\Delta(\gamma_\rmc,P_{UZ})}{2 \sqrt{\sizeI}} \label{eqn:resolv1} 
\end{align}
where $\barP_Z(z) = \sum_u P_U(u) \barP_{Z|U}(z|u)$.
\end{lemma}

\begin{remark}
Although the statement of \cite[Corollary 7.2]{Cuff12} consists of two terms, 
the second term corresponds to the right hand side of \eqref{eqn:resolv1}.
Since our target distribution $\barP_Z$ is smoothed, the first term of \cite[Corollary 7.2]{Cuff12} does not appear in \eqref{eqn:resolv1}.
\end{remark}


\section{Simulation of Test Channel}\label{app:sim}
In this appendix, we develop two lemmas which form crucial components of
the proof of all CS-type bounds.  To do this, we consider the problem related to
channel simulation~\cite{Luo09 ,BSST, BSST02 , Win02, Cuff12}.  Roughly
speaking, the problem is described as follows. 
For a given message set $\calL$ and a code $\calC = \{u_1,\ldots,u_{\sizeL}\}$, our goal is
to construct a stochastic map $\varphi:\calZ \to \calL$ such that the joint distribution $P_{\hatL Z}$
of $(\varphi(Z),Z)$ is indistinguishable from $P_{L\tilZ}$, where $P_{L\tilZ}$ is the joint distribution such that
$u_L$ is sent over the channel $P_{Z|U}$ for the uniform random number $L$ on $\calL$.
This is done by the argument of the {\em likelihood encoder} \cite{Cuff12} (see also \cite{CS13}). However, we need to modify 
the argument in \cite{Cuff12} since our goal is, in fact, to approximate a smoothed version of $P_{L\tilZ}$.
We will use notations introduced in Appendix \ref{app:channel_res}.

\begin{remark} \label{remark:channel-simulation}
In the earlier version of this paper \cite{WKT13b}, we were considering exactly the problem of channel simulation, where we simulate the joint distribution $P_{UZ}$ by the aid of the common randomness. However, simulating the marginal $P_U$ is unnecessary to derive bounds on WAK, WZ, and GP problems. Thus, we consider approximation of $P_{L\tilZ}$ in this paper, which enables us
to remove a residual term in \cite{WKT13b} that stems from the use of the common randomness.
\end{remark}

To construct a stochastic map from $\calZ$ to $\calL$, 
we first consider the channel resolvability code as follows. 
Let us generate a codebook $\calC=\{u_{1},\dots,u_{\sizeL}\}$, where
each codeword $u_{l}$ is randomly and independently generated from
$P_U$, which is the marginal of $P_{UZ}$.
Let $L$ be the uniform random numbers on $\calL$.
Moreover, let $\barP_{Z|U}$ be a smoothed version of $P_{Z|U}$ defined in 
\eqref{eqn:smooth-P_Z|U}.
Then, $\calC$, $L$, and $\barP_{Z|U}$ induce the sub-normalized
measure
\begin{equation}
 \barP_{L\tilZ}(l,z) := \frac{1}{\sizeL} \barP_{Z|U}(z|u_l).
\end{equation}
Marginal 
$\barP_{\tilZ}$ is also induced as
\begin{align}
 \barP_{\tilZ}(z) = \sum_{l}\frac{1}{\sizeL}\barP_{Z|U}(z|u_l).
\end{align}

Now, we define a stochastic map
$\varphi_\calC\colon\calZ\to\calL$ as\footnote{When $\barP_{\tilZ}(z) = 0$, we define $\varphi_\calC(l|z)$ arbitrarily. }
\begin{equation}
 \varphi_\calC(l|z)=\frac{\barP_{L\tilZ}(l,z)}{\barP_{\tilZ}(z)}.
\label{eqn:varphi_calC}
\end{equation}
Let $\hatL$ be the output of the stochastic map $\varphi_\calC$ for the input $Z$. Then, 
the joint distribution of $\hatL$ and $Z$ is given by
\begin{equation}
 P_{\hatL Z}(l,z) = P_Z(z) \varphi_\calC(l|z).
\label{eqn:sim-joint-distrib}
\end{equation}

We also introduce a smoothed version of 
$P_{\hatL Z}$ as follows:
\begin{equation}
 \barP_{\hatL Z}(l,z) = \barP_Z(z) \varphi_\calC(l|z),
\label{eqn:sim-joint-distrib-smoothed}
\end{equation}
where $\barP_Z$ is the marginal of $\barP_{UZ}:=P_U\barP_{Z|U}$; i.e.~$\barP_{Z}(z):=\sum_uP_U(u)\barP_{Z|U}(z|u)$.

Now, we prove two lemmas which can be used to evaluate the performance of
the approximation of $\barP_{L\tilZ}$.

\begin{lemma}
\label{lem:sim1}
We have 
\begin{equation}
d(P_{\hatL Z}, \barP_{L \tilZ}) \leq \frac{ P_{UZ}( (u,z)\notin \calT_{\rmc}(\gamma_{\rmc}) ) }{2}+ d(\barP_{\hatL Z}, \barP_{L \tilZ}).
\end{equation}
\end{lemma}

\begin{IEEEproof}
By the triangular inequality, we have
\begin{equation}
 d(P_{\hatL Z}, \barP_{L \tilZ}) \leq d(P_{\hatL Z}, \barP_{\hatL Z})+d(\barP_{\hatL Z}, \barP_{L \tilZ}).
\end{equation}
Further, we can bound the first term of the right hand side of the above
 inequality as
\begin{align}
 d(P_{\hatL Z}, \barP_{\hatL Z})
&=d(P_Z \varphi_\calC, \barP_Z \varphi_\calC)\\
&= d(P_Z,\barP_Z)\label{eqn:tmp1-proof-lem-sim1}\\
&\leq d(P_{UZ},\barP_{UZ})\label{eqn:tmp2-proof-lem-sim1}\\
&=\frac{ P_{UZ}( (u,z)\in \calT_{\rmc}(\gamma_{\rmc})^c ) }{2}\label{eqn:tmp3-proof-lem-sim1}
\end{align}
where \eqref{eqn:tmp1-proof-lem-sim1} follows 
the data-processing inequality \eqref{eqn:distance-property2}, \eqref{eqn:tmp2-proof-lem-sim1} follows from the monotonicity property in
 \eqref{eqn:distance-property1}, and 
\eqref{eqn:tmp3-proof-lem-sim1} 
follows from \eqref{eqn:distance-of-smoothing2}.
\end{IEEEproof}

\begin{lemma}
\label{lem:sim2}
We have 
\begin{align}
\bbE_{\calC}[d(\barP_{\hatL Z}, \barP_{L \tilZ})] 
 \leq \frac{\Delta(\gamma_{\rmc}, P_{UZ} )}{2 \sqrt{\sizeL}}. 
\label{eqn:sim2}
\end{align}
\end{lemma}
\begin{IEEEproof}
By noting that the definition of $\varphi_{\calC}$ in \eqref{eqn:varphi_calC} implies 
$\barP_{L\tilZ} = \barP_{\tilZ} \varphi_{\calC}$, we have
\begin{align}
d(\barP_{\hat{L}Z},\barP_{L\tilZ})
&= d(\barP_Z \varphi_{\calC}, \barP_{\tilZ} \varphi_{\calC}) \\
&= d(\barP_Z, \barP_{\tilZ}).
\end{align}
Then, by taking the expectation with respect to the codebook $\calC$ and by invoking 
Lemma \ref{lem:resolvability}, we have the desired bound.
\end{IEEEproof}


\section{Proof of the First Non-Asymptotic Bound for    WAK   in Theorem~\ref{thm:dt}}\label{app:dt_wz}
\subsection{Code Construction}
We construct a WAK code by using the stochastic map introduced in
Appendix \ref{app:sim}.  Let
$\calZ=\calY$ and $Z=Y$, that is, let $P_{UZ}=P_{UY}$, where $P_{UY}$ is
the marginal of the given distribution $P_{UXY}\in\scP(P_{XY})$. Also let   $\tilZ=\tilY$  per~\eqref{eqn:deftilZ}.  It
should be noted here that, in this case, $\calT_\rmc(\gamma_\rmc)$
defined in \eqref{eqn:calT-for-smoothing} is equivalent to
$\calT_\rmc^\WAK(\gamma_\rmc)$ defined in \eqref{eqn:T2}.  Now, let us
consider the stochastic map $\varphi_\calC$ constructed from the smoothed measure $\barP_{L\tilY}$
(cf.~\eqref{eqn:varphi_calC}).

By using $\varphi_\calC$, we construct a WAK code $\Phi$ as follows.  
The main encoder uses a random bin
coding $f\colon\calX\to\calM$.  The helper uses the stochastic map
$\varphi_\calC\colon\calY\to\calL$.  That is, 
when the side
information is $y\in\calY$,
the helper generates
$l\in\calL$ according to $\varphi_\calC(\fndot|y)$ and sends $l$ to the decoder.
For given $m\in\calM$ and 
$l\in\calL$, the decoder
outputs the unique $\hatx\in\calX$ such that $f(\hatx)=m$ and
\begin{equation}
 (u_{l},\hatx)\in\calT_\rmb^\WAK(\gamma_\rmb).
\end{equation}
If no such unique $\hatx$ exists, or if there is more than one such
$\hatx$, then a decoding error is declared.

\subsection{Analysis of Error Probability}
Let $\hatL$ be the random index chosen by the helper via the stochastic
map $\varphi_\calC(\fndot|Y)$.
Note that the joint distribution of $\hatL$ and $Y$ is given as
follows; cf.~\eqref{eqn:sim-joint-distrib}
\begin{equation}
  P_{\hatL Y}(l,y) = P_Y(y)\varphi_\calC(l|y)
\label{eqn:proof-WAK-joint-distrib}
\end{equation}
and then, the joint distribution of $\hatL,Y$ and $X$ is given as
\begin{equation}
  P_{\hatL XY}(l,x,y)=P_{\hatL Y}(l,y)P_{X|Y}(x|y).
\end{equation}
The smoothed versions $\barP_{\hatL Y}$ and $\barP_{\hatL XY}$ are given by substituting $P_Y$ in
\eqref{eqn:proof-WAK-joint-distrib} with $\barP_Y$;
cf.~\eqref{eqn:sim-joint-distrib-smoothed}.

If the decoding error occurs, at least one of the following events occurs:
\begin{align*}
 \calE_1&:=\left\{ (u_{l},x)\notin\calT_\rmb^\WAK(\gamma_\rmb) \right\}\\
 \calE_2&:=\left\{  \exists\,\tilx\neq x\text{ s.t. }f(\tilx)=f(x),(u_{l},\tilx)\in\calT_\rmb^\WAK(\gamma_\rmb)\right\}
\end{align*}
Hence, the error probability averaged over random coding $f$ and
the random codebook $\calC$ can be bounded as
\begin{align}
 \bbE_{f}\bbE_\calC[\Pe(\Phi)]
&= \bbE_{f}\bbE_\calC\left[
P_{\hatL XY}(\calE_1\cup\calE_2)
\right]\label{eqn:tmp1-proof-WAK}.
\end{align}

Let
\begin{align}
 \calE_{12}:=& \Big\{
(u,x):(u,x)\notin\calT_\rmb^\WAK(\gamma_\rmb) \text{ or }
\exists\,\tilx\neq x \nn\\
&~~~\text{ s.t. }f(\tilx)=f(x),(u,\tilx)\in\calT_\rmb^\WAK(\gamma_\rmb)
\Big\}.
\end{align}
Then, for fixed $f$ and $\calC$, we have
\begin{align} 
&P_{\hatL XY}(\calE_1\cup\calE_2) \nn\\
&= P_{\hatL XY}((u_l,x) \in \calE_{12}) \\
&\leq \barP_{LX\tilY}((u_l,x) \in \calE_{12}) + \frac{1 - \barP_{LX\tilY}(\calL \times \calX \times \calY)}{2} \nn\\
&~~~~~~	+ d(P_{\hatL XY}, \barP_{LX\tilY}) \label{eqn:tmp2-proof-WAK} \\
&= \barP_{LX\tilY}((u_l,x) \in \calE_{12}) + \frac{1 - \barP_{LX\tilY}(\calL \times \calX \times \calY)}{2} \nn\\
&~~~~~~	+ d(P_{\hatL Y} P_{X|Y}, \barP_{L\tilY} P_{X|Y}) \\
&\leq \barP_{LX\tilY}((u_l,x) \in \calE_{12}) + \frac{1 - \barP_{LX\tilY}(\calL \times \calX \times \calY)}{2} \nn\\
&~~~~~~	+ d(P_{\hatL Y}, \barP_{L\tilY}) \label{eqn:tmp3-proof-WAK} \\
&\leq \barP_{LX\tilY}((u_l,x) \notin \calT_\rmb^\WAK(\gamma_\rmb)) \nn\\
&~~~ + \barP_{LX\tilY}[ \exists\,\tilx\neq x\text{ s.t. }f(\tilx)=f(x),(u_l,\tilx)\in\calT_\rmb^\WAK(\gamma_\rmb) ] \nn\\
&~~~+\frac{1 - \barP_{LX\tilY}(\calL \times \calX \times \calY)}{2} 
	+ d(P_{\hatL Y}, \barP_{L\tilY}) \\
&= P_{LX\tilY}((u_l,x) \notin \calT_\rmb^\WAK(\gamma_\rmb) \cap (u_l,y) \in \calT_\rmc^\WAK(\gamma_\rmc)) \nn\\
&~~~ + \barP_{LX\tilY}[ \exists\,\tilx\neq x\text{ s.t. }f(\tilx)=f(x),(u_l,\tilx)\in\calT_\rmb(\gamma_\rmb) ] \nn\\
&~~~+\frac{1 - \barP_{LX\tilY}(\calL \times \calX \times \calY)}{2} 
	+ d(P_{\hatL Y}, \barP_{L\tilY})
	\label{eqn:tmp5-proof-WAK}
\end{align}
where 
\eqref{eqn:tmp2-proof-WAK} follows from \eqref{eqn:distance-property3} for 
$\barP_{LX\tilY} = \barP_{L\tilY} P_{X|Y}$ in the role of $Q$, and
\eqref{eqn:tmp3-proof-WAK} follows from the data-processing inequality \eqref{eqn:distance-property2}.
By taking average over $\calC$, the first term in \eqref{eqn:tmp5-proof-WAK} is given by
\begin{align}
& \bbE_\calC\left[ P_{LX\tilY}((u_l,x) \notin \calT_\rmb^\WAK(\gamma_\rmb) \cap (u_l,y) \in \calT_\rmc^\WAK(\gamma_\rmc)) \right] \nn \\
&= \bbE_\calC\bigg[ \sum_{u,x,y} \sum_{l} \frac{1}{\sizeL} \bone[u_l = u] P_{Y|U}(y|u) P_{X|Y}(x|y) \nn\\
&~~~\times	\bone[(u,x) \notin \calT_\rmb^\WAK(\gamma_\rmb) \cap (u,y) \in \calT_\rmc^\WAK(\gamma_\rmc)] \bigg] \\
&= P_{UXY}((u,x) \notin \calT_\rmb^\WAK(\gamma_\rmb) \cap (u,y) \in \calT_\rmc^\WAK(\gamma_\rmc)),
\label{eqn:tmp51-proof-WAK}
\end{align}
the third term in 
\eqref{eqn:tmp5-proof-WAK} is given by
\begin{align}
& \bbE_\calC\left[ 1 - \barP_{LX\tilY}(\calL \times \calX \times \calY) \right]  \nn \\
&= 1 - \bbE_\calC\bigg[ \sum_{u,x,y} \sum_l \frac{1}{\sizeL} \bone[u_l = u] P_{Y|U}(y|u)P_{X|Y}(x|y) \nn\\
&~~~\times 	\bone[ (u,y) \in \calT_\rmc^\WAK(\gamma_\rmc) ] \bigg] \\
&= P_{UY}((u,y) \notin \calT_\rmc^\WAK(\gamma_\rmc)),
\label{eqn:tmp52-proof-WAK}
\end{align}
and the fourth term in 
\eqref{eqn:tmp5-proof-WAK} is upper bounded as 
\begin{align}
\bbE_\calC\left[ d(P_{\hatL Y}, \barP_{L\tilY}) \right]
 \leq& \frac{P_{UY}((u,y) \notin \calT_\rmc^\WAK(\gamma_\rmc))}{2} \nn\\
 &+ \frac{\Delta(\gamma_{\rmc}, P_{UY} )}{2 \sqrt{\sizeL}},
 \label{eqn:tmp53-proof-WAK}
\end{align}
where we used Lemma \ref{lem:sim1} and Lemma \ref{lem:sim2}.
Furthermore, by taking average over $f$ and $\calC$, the second term in \eqref{eqn:tmp5-proof-WAK} is upper bounded as 
\begin{align}
& \bbE_f\bbE_\calC
\Big[
\barP_{LXY}[\exists\,\tilx\neq x\text{ s.t. } \nn \\
& ~~~~~~~~~~~f(\tilx)=f(x),(u_l,\tilx)\in\calT_\rmb^\WAK(\gamma_\rmb)]
\Big] \nn\\
&= \bbE_f\bbE_\calC
\bigg[ \sum_{u,x,y} \sum_l \frac{1}{\sizeL} \bone[u_l = u] \barP_{Y|U}(y|u) P_{X|Y}(x|y) \nn\\
&~~~\times	\bone[\exists\,\tilx\neq x\text{ s.t. }f(\tilx)=f(x),(u,\tilx)\in\calT_\rmb^\WAK(\gamma_\rmb)] \bigg] \\
&= \bbE_f 
\bigg[ \sum_{u,x,y} \barP_{UXY}(u,x,y) \nn\\
&~~~\times	\bone[\exists\,\tilx\neq x\text{ s.t. }f(\tilx)=f(x),(u,\tilx)\in\calT_\rmb^\WAK(\gamma_\rmb)] \bigg] \\
&\leq \sum_{u,x,y}\barP_{UXY}(u,x,y) \nn\\
&~~~\times \sum_{\tilx\neq x}\bbE_f[\bone[f(\tilx)=f(x)]]\bone[(u,\tilx)\in\calT_\rmb^\WAK(\gamma_\rmb)]\\
&\leq
 \frac{1}{\sizeM}\sum_uP_U(u)\sum_{\tilx}  \bone[(u,\tilx)\in\calT_\rmb^\WAK(\gamma_\rmb)] \label{eqn:tmp65-proof-WAK}\\
&=\frac{1}{\sizeM}\sum_{(u,\tilx)\in\calT_\rmb^\WAK(\gamma_\rmb)}P_U(u)
\label{eqn:tmp7-proof-WAK}
\end{align}
where we used the fact $\sum_{x,y}\barP_{UXY}(u,x,y)\leq P_U(u)$ in \eqref{eqn:tmp65-proof-WAK}. 
Hence, by \eqref{eqn:tmp5-proof-WAK},
\eqref{eqn:tmp51-proof-WAK}, \eqref{eqn:tmp52-proof-WAK}, \eqref{eqn:tmp53-proof-WAK},
and \eqref{eqn:tmp7-proof-WAK}, we
have
\begin{align}
& \bbE_{f}\bbE_\calC[\Pe(\Phi)] \nn\\
&=\bbE_f\bbE_\calC\left[P_{K\hatL\hatU XY}(\calE_1\cup\calE_2)\right]\\
&\leq P_{UXY}((u,x)\notin\calT_\rmb^\WAK(\gamma_\rmb)\cup (u,y)\notin\calT_\rmc^\WAK(\gamma_\rmc))\nn\\
&\qquad + 
\frac{\Delta(\gamma_{\rmc}, P_{UY} )}{2 \sqrt{\sizeL}}  +\frac{1}{\sizeM}\sum_{(u,\tilx)\in\calT_\rmb^\WAK(\gamma_\rmb)}P_U(u).
\label{eqn:tmp10-proof-WAK}
\end{align}
Consequently, there exists at least one code $(f,\calC)$ such that 
$\Pe(\Phi)$ is smaller than the right-hand-side of the inequality above.
This completes the proof of Theorem~\ref{thm:dt}.

\section{Proof of the Second Non-Asymptotic Bound for    WAK   in Theorem~\ref{thm:dt-helper-binning}}\label{app:dt_wz-helper-binning}
To prove Theorem~\ref{thm:dt-helper-binning}, we modify the proof of 
Theorem~\ref{thm:dt} as follows. Since the analysis of error can be done in a similar manner as 
Appendix \ref{app:dt_wz}, we only show the code construction.

First, we use $\calJ=\{1,\dots,J\}$ instead of $\calL$ in the construction of 
    $\varphi_\calC$,  where $J$ is the given integer.  
Then, the helper and the decoder are modified as follows.
The helper first uses the stochastic map
$\varphi_\calC\colon\calY\to\calJ$.  That is, it generates
$j\in\calJ$ according to $\varphi_\calC(\fndot|y)$ when the side
information is $y\in\calY$.
Then, the helper sends $j$ by using random bin coding
$\bin\colon\calJ\to\calL$. This means that to every $j\in\calJ$, it  independently and uniformly assigns a random  index $l\in\calL$.  For given $m\in\calM$ and 
$l\in\calL$, the decoder
outputs the unique $\hatx\in\calX$ such that $f(\hatx)=m$ and
\begin{equation}
 (u_{j},\hatx)\in\calT_\rmb^\WAK(\gamma_\rmb)
\end{equation}
for some $j\in\calJ$ satisfying $\bin(j)=l$.
If no such unique $\hatx$ exists, or if there is more than one such
$\hatx$, then  a decoding error is declared.

\section{Proof of the Non-Asymptotic Bound for    WZ   in Theorem~\ref{thm:cs_wz}}\label{app:cs_wz}
\subsection{Code Construction}
Similar to WAK coding in the previous two sections, we use the stochastic map introduced in Appendix \ref{app:sim}.  Also, the proof is rather similar to the WAK one  so we just highlight the key steps, pointing the reader to various points of Appendix~\ref{app:dt_wz} for the details of the calculations. 

In WZ coding, let $\calZ=\calX$ and $P_{UZ}=P_{UX}$.  Also let  $\tilZ=\tilX$  per~\eqref{eqn:deftilZ}.  
Note   that  $\calT_\rmc(\gamma_\rmc)$
defined in \eqref{eqn:calT-for-smoothing} is equivalent  to
$\calT_\rmc^\WZ(\gamma_\rmc)$ defined in \eqref{eqn:T2_wz}.
Now, let us consider the stochastic map $\varphi_\calC$ defined in~\eqref{eqn:varphi_calC}.

By using $\varphi_\calC$, we construct a WZ code $\Phi$ as follows.  
The encoder first uses the stochastic map
$\varphi_\calC\colon \calX\to\calL$.  That is, it generates
$l\in\calL$ according to $\varphi_\calC(\fndot|x)$ when the source
output is $x\in\calX$. 
Then, the encoder sends $l$ by using random bin coding
$\bin\colon\calL\to\calM$. This means that to every $l\in\calL$, it  independently and uniformly assigns a random  index $m\in\calM$.    For given $m\in\calM$, $y\in\calY$, the decoder
finds the unique index $l\in\calL$ such that $\bin(l)=m$ and
\begin{equation}
 (u_{l},y)\in\calT_\rmp^\WZ(\gamma_\rmp). \label{eqn:wz_pack}
\end{equation}
Then, decoder outputs $\hatx\in\hcalX$ according to $P_{\hatX|UY}(\fndot|u_{l},y)$. We assume that we use the {\em stochastic} reproduction function  $P_{\hatX|UY}$ throughout. If the {\em deterministic} reproduction function $g:\calU\times\calY\to\hcalX$ is used, the decoder outputs $\hatx=g(u_{l},y)$.
If no unique $l$ satisfying \eqref{eqn:wz_pack} exists, or if there is more than one such
$l$ satisfying~\eqref{eqn:wz_pack}, then a decoding error is declared.

\subsection{Analysis of Probability of Excess Distortion}
Let $\hatL$ be the random index chosen by the encoder via the stochastic
map $\varphi_\calC(\fndot| X)$.
Note that the joint distribution of $\hatL ,X$ is given as
follows; cf.~\eqref{eqn:sim-joint-distrib}
\begin{equation}
P_{\hatL X} (l,x)= P_{X}(x) \varphi_{\calC}(l|x).
\label{eqn:proof-WZ-joint-distrib}
\end{equation}
Next,  the joint distribution of $\hatL,X,Y,\hatX$ is given as
\begin{equation}
P_{\hatL X Y \hatX}(l,x,y,\hatx) = P_{\hatL X}(l,x) P_{Y|X}(y|x) P_{\hatX | U Y} (\hatx | u_{l} ,y).
\end{equation}
The smoothed versions $\barP_{\hatL X}$ and $\barP_{\hatL XY\hatX}$ are given by substituting $P_X$ in
\eqref{eqn:proof-WZ-joint-distrib} with $\barP_X$;
cf.~\eqref{eqn:sim-joint-distrib-smoothed}.

If the distortion exceeds $D$, 
at least one of the following events occurs:
\begin{align}
 \calE_0&:=\left\{  (x,\hatx)\notin\calTsWZ(D) \right\}\\
 \calE_1&:=\left\{  (u_{l},y)\notin\calT_\rmp^\WZ(\gamma_\rmp)\right\}\\
 \calE_2&:=\left\{  \exists\,\till\neq l\text{ s.t. }\bin(\till)=\bin(l),(u_{\till},y)\in\calT_\rmp^\WZ(\gamma_\rmp)\right\}.
\end{align}
Hence, the   probability  of excess distortion averaged over the random coding $\bin$ and the
random codebook $\calC$ can be bounded as
\begin{align}
& \bbE_{\bin}\bbE_\calC [\Pe(\Phi;D)] \nn\\
&\leq
\bbE_{\bin} \bbE_\calC  \left[ P_{\hatL X Y \hatX} (\calE_0\cup\calE_1\cup\calE_2) \right] \\
& \le  \bbE_\calC \left[ P_{\hatL X Y \hatX} (\calE_0\cup\calE_1) \right] + \bbE_{\bin} \bbE_\calC  \left[P_{\hatL X Y  } ( \calE_2)\right] 
\label{eqn:tmp1-proof-WZ}.
\end{align}

At first, we evaluate the first term in \eqref{eqn:tmp1-proof-WZ}.
For fixed $\calC$,
\begin{align}
& P_{\hatL X Y \hatX} (\calE_0\cup\calE_1) \nn\\
&\le   \barP_{L \tilX Y \hatX} (\calE_0\cup\calE_1)  + \frac{1-\barP_{L \tilX Y \hatX}(\calL\times\calX\times\calY\times\hat{\calX})}{2} \nn\\
&~~~+ d(P_{\hatL X Y \hatX},\barP_{L \tilX Y \hatX}) \label{eqn:tmp2-proof-WZ} \\
 &\le   \barP_{L \tilX Y \hatX} (\calE_0\cup\calE_1)  + \frac{1-\barP_{L \tilX Y \hatX}(\calL\times\calX\times\calY\times\hat{\calX})}{2} \nn\\
 &~~~+ d(P_{\hatL X  },\barP_{L \tilX  }) \label{eqn:tmp3-proof-WZ}  \\
 &=P_{LXY\hatX}((u_l,x)\in\calT_\rmc^\WZ(\gamma_\rmc)\cup (x,\hatx)\notin\calTsWZ(D) \nn\\
 &~~~~~~~~~~~~~~~~ \cup (u_l,y)\notin\calT_\rmp^\WZ(\gamma_\rmp) )\nn\\
 &~~~ + \frac{1-\barP_{L \tilX Y \hatX}(\calL\times\calX\times\calY\times\hat{\calX})}{2} + d(P_{\hatL X  },\barP_{L \tilX  }) \label{eqn:tmp4-proof-WZ}  
\end{align}
where  \eqref{eqn:tmp2-proof-WZ} follows from \eqref{eqn:distance-property3}, \eqref{eqn:tmp3-proof-WZ} follows from the same reasoning that led to \eqref{eqn:tmp3-proof-WAK} and \eqref{eqn:tmp4-proof-WZ} from the same reasoning that led to \eqref{eqn:tmp5-proof-WAK}.

By the same reasoning that led to \eqref{eqn:tmp51-proof-WAK} for the WAK problem, the expectation of the first term in \eqref{eqn:tmp4-proof-WZ} can be expressed as 
\begin{align}
&\bbE_\calC\Big[P_{LXY\hatX}((u_l,x)\in\calT_\rmc^\WZ(\gamma_\rmc)\cup (x,\hatx)\notin\calTsWZ(D) \nn\\
&~~~~~~~~~~~~~~~\cup (u,y)\notin\calT_\rmp^\WZ(\gamma_\rmp) ) \Big]  \\
&= P_{UX Y \hatX} ((u,x)\in\calT_\rmc^\WZ(\gamma_\rmc)\cup (x,\hatx)\notin\calTsWZ(D) \nn\\
&~~~~~~~~~~~~~~~ \cup (u,y)\notin\calT_\rmp^\WZ(\gamma_\rmp) )  \label{eqn:first_WZ_proof}
\end{align}
By the same reasoning that led to \eqref{eqn:tmp52-proof-WAK} for the WAK problem, the expectation of the  second term in \eqref{eqn:tmp4-proof-WZ}   can be evaluated  as 
\begin{equation}
\bbE_\calC [ 1- \barP_{L \tilX Y \hatX}(\calL\times\calX\times\calY\times\hat{\calX}) ] = P_{UX}( (u,x) \notin \calT_\rmc^\WZ(\gamma_\rmc)). \label{eqn:all_WZ_proof}
\end{equation}
Similarly to  \eqref{eqn:tmp53-proof-WAK} for the WAK problem, the expectation of the third term in  \eqref{eqn:tmp4-proof-WZ} can be bounded as 
\begin{align}
\bbE_\calC [ d(P_{\hatL X  },\barP_{L \tilX  })]\le\frac{ P_{UX} ( (u,x) \notin \calT_\rmc^\WZ(\gamma_\rmc)) }{2} + \frac{\Delta(\gamma_{\rmc}, P_{UX})}{2\sqrt{|\calL|}} . \label{eqn:sim_WZ_proof}
\end{align}
Now we bound the final term in \eqref{eqn:tmp1-proof-WZ} using steps similar to the ones leading to \eqref{eqn:tmp7-proof-WAK} for the WAK problem. We have 
\begin{align}
&\bbE_{\bin} \bbE_\calC  \left[P_{\hatL X Y  } ( \calE_2)\right] \nn\\
& =\bbE_{\bin} \bbE_\calC \bigg[\sum_{u,x,y,l}\frac{1}{|\calL|}\bone[u_l=u] P_{\hatL X Y U }(l,x,y,u)\bone[\exists \, \till\ne l \nn\\
&~~~~~~~~~~~~ \mbox{ s.t. }\kappa(\till) =\kappa(l), (u_{\till} ,y) \in  \calT_\rmp^\WZ(\gamma_\rmp)] \bigg] \\
& \le\bbE_{\bin} \bbE_\calC \bigg[\sum_{u,x,y,l}\frac{1}{|\calL|} \bone[u_l=u] P_{\hatL X Y U }(l,x,y,u) \nn\\
&~~~~~~~~~~ \sum_{\till \ne l} \bone[\kappa(\till) =\kappa(l) ]\cdot \bone[ (u_{\till} ,y) \in  \calT_\rmp^\WZ(\gamma_\rmp)] \bigg] \\
&\le \frac{1}{|\calM|} \bbE_\calC\bigg[ \sum_{u,x,y,l}\frac{1}{|\calL|} \bone[u_l=u] P_{\hatL X Y U }(l,x,y,u) \nn\\
&~~~~~~~~~~  \sum_{\till \ne l }  \bone[ (u_{\till} ,y) \in  \calT_\rmp^\WZ(\gamma_\rmp)] \bigg]\\
&\le \frac{|\calL|}{|\calM|} \sum_{u,y} P_U(u) P_Y(y)\bone\left[ (u ,y) \in  \calT_\rmp^\WZ(\gamma_\rmp)] \right] \\
&= \frac{|\calL|}{|\calM|} \sum_{(u,y) \in   \calT_\rmp^\WZ(\gamma_\rmp)} P_U(u) P_Y(y). \label{eqn:packing_WZ_proof}
\end{align}
By uniting \eqref{eqn:tmp1-proof-WZ}, \eqref{eqn:tmp4-proof-WZ}, \eqref{eqn:first_WZ_proof}, \eqref{eqn:all_WZ_proof}, \eqref{eqn:sim_WZ_proof} and \eqref{eqn:packing_WZ_proof}, we obtain the final bound
\begin{align}
& \bbE_{\bin}\bbE_\calC [\Pe(\Phi;D)] \nn\\
&\le  P_{UX Y \hatX} ((u,x)\notin\calT_\rmc^\WZ(\gamma_\rmc)\cup (x,\hatx)\notin\calTsWZ(D) \nn\\
&~~~~~~~~~~~~\cup (u,y)\notin\calT_\rmp^\WZ(\gamma_\rmp) )  \nn\\
 &~~~+ \frac{\Delta(\gamma_{\rmc}, P_{UX})}{2\sqrt{|\calL|}}+\frac{|\calL|}{|\calM|} \sum_{(u,y) \in   \calT_\rmp^\WZ(\gamma_\rmp)} P_U(u) P_Y(y). \label{eqn:final_proof_WZ}
\end{align}
This implies there is a deterministic code whose probability of excess distortion is no greater than the right-hand-side of \eqref{eqn:final_proof_WZ}.  This completes the proof of Theorem~\ref{thm:cs_wz}.
\section{Proof of the Non-Asymptotic Bound for    GP   in Theorem~\ref{thm:cs_gp}}\label{app:cs_gp}

Since the analysis of error probability can be done in an almost similar manner as those
of WAK and WZ, we only show the code construction for GP.

\subsection{Code Construction}
As in WAK, we use the stochastic map introduced in Appendix~\ref{app:sim}.  
In GP coding, let $\calZ=\calS$ and $P_{UZ}=P_{US}$.
Note that that, $\calT_\rmc(\gamma_\rmc)$
defined in \eqref{eqn:calT-for-smoothing} is equivalent to
$\calT_\rmc^\GP(\gamma_\rmc)$ defined in \eqref{eqn:Tc_gp}
in this case.

For GP coding, we construct $\sizeM$ stochastic maps. Each stochastic map corresponds to  a message in $\calM$.  For each message
$m\in\calM$, generate a codebook
$\calC^{(m)}=\{u_{1}^{(m)},\dots,u_{\sizeL}^{(m)}\}$ where each $u_{l}^{(m)}$ is independently drawn according to
$P_U$.  Then, for each $\calC^{(m)}$ ($m\in\calM$), construct a
stochastic map $\varphi_{\calC^{(m)}}$ as defined in
\eqref{eqn:varphi_calC}.

By using $\{\varphi_{\calC^{(m)}}\}_{m\in\calM}$, 
we construct a GP code $\Phi$ as follows.  
Given the message $m\in\calM$ and the channel state $s\in\calS$, the encoder first 
generates
$l\in\calL$ according to $\varphi_{\calC^{(m)}}(\fndot|s)$.
Then, the encoder generates $x\in\calX$ according to $P_{X|US}(\fndot|u_{l}^{(m)},s)$
and inputs $x$ into the channel. 
If the randomly generated $x$ results in $\rvg(x)>\Gamma$ (i.e., the channel input does not satisfy the cost constraint), declare an cost-constraint violation error.\footnote{Even if $\rvg(x)>\Gamma$ occurs, we still send $x$ through the channel. 
The error event for this occurrence must be taken into accounted in the error analysis.}
Given the channel output $y\in\calY$, the decoder
finds the unique index $\hatm\in\calM$ such that 
\begin{equation}
 (u_{l}^{(\hatm)},y)\in\calT_\rmp^\GP(\gamma_\rmp)
\end{equation}
for some $l\in\calL$.  If there is no unique index $\hatm \in\calM$ or more than one, declare a decoding error. This is a Feinstein-like decoder~\cite{Han10} for average probability of error.
If no such unique $\hatm$ exists, or if there exists more than one such
$\hatm$, then a decoding error is declared.

\section{Preliminaries for Proofs of the Second-Order Coding Rate} \label{app:preliminaries-second}

In this appendix, we provide some technical results that will be used in Appendices \ref{app:second} and \ref{app:wz_2nd}.
More specifically, we will use the following multidimensional Berry-Ess\'een theorem and its corollary.
\begin{theorem}[G\"{o}etze~\cite{Got91}] \label{theorem:multidimensional-berry-esseen}
Let $\bU_1,\ldots,\bU_n$ be independent random vectors in $\bbR^k$ with zero mean. Let $\bS_n = \frac{1}{\sqrt{n}}(\bU_1 + \cdots + \bU_n)$,
$\cov(\bS_n) = \bI$, and $\xi = \frac{1}{n} \sum_{i=1}^n \bbE[ \| \bU_i \|^3_2]$. Let the standard Gaussian random vector $\bZ\sim\calN(\bzero,\bI)$.
Then, for all $n \in \bbN$, we have
\begin{equation}
\sup_{\mathscr{C} \in \mathfrak{C}_k} \left| \Pr\{ \bS_n \in \mathscr{C} \} - \Pr\{ \bZ \in \mathscr{C} \} \right| 
\le \frac{C_k \xi}{ \sqrt{n}},
\end{equation}
where $\mathfrak{C}_k$ is the family of all convex, Borel measurable  subsets of $\bbR^k$, and where $C_k$ is a constant
that depends on the dimension $k$.
\end{theorem}
It should be noted that Theorem \ref{theorem:multidimensional-berry-esseen} can be applied for random vectors that are independent but not necessarily
identical. 

We will frequently encounter random vectors with non-identity covariance matrices. Thus, we slightly modify 
Theorem \ref{theorem:multidimensional-berry-esseen} in a similar manner as \cite[Corollary 7]{TK12} as follows.
\begin{corollary} \label{corollary:multidimensional-berry-esseen}
Let $\bU_1, \ldots, \bU_n$ be independent random vectors in $\bbR^k$ with 
zero mean.
Let $\bS_n = \frac{1}{\sqrt{n}}(\bU_1 + \cdots + \bU_n)$,
$\cov(\bS_n) = \bV \succ 0$,
 and $\xi = \frac{1}{n} \sum_{i=1}^n \bbE[ \| \bU_i \|_2^3]$. 
 Let the  Gaussian random vector $\bZ\sim\calN(\bzero,\bV)$. Then, for all $n \in \mathbb{N}$, 
\begin{equation}
\sup_{\mathscr{C} \in \mathfrak{C}_k} \left| \Pr\{ \bS_n \in \mathscr{C} \} - \Pr\{ \bZ \in \mathscr{C} \} \right| 
\le \frac{C_k \xi}{\lambda_{\min}(\bV)^{3/2} \sqrt{n}},
\end{equation}
where $\mathfrak{C}_k$ is the family of all convex, Borel measurable  subsets of $\mathbb{R}^k$, 
where $C_k$ is a constant
that depends on the dimension $k$, and where 
$\lambda_{\min}(\bV)$ is the smallest eigenvalue of $\bV$.
\end{corollary} 

\section{Achievability Proof of the Second-Order Coding Rate for    WAK in Theorem~\ref{thm:second}}\label{app:second}

\begin{proof}
It suffices to show the inclusion $\scR_{\mathrm{in}} (n,\veps; P_{UTXY}) \subset\scR_{\WAK}(n,\veps)$ 
for fixed $P_{UTXY}\in\tilde{\scP}(P_{XY})$. 

We first consider the case such that $\bV = \bV(P_{UTXY}) \succ 0$.
First, note that $\bR \in \scR_{\mathrm{in}} (n,\veps; P_{UTXY})$ implies
\begin{equation} \label{eq:redundancy-of-WAK}
\tilde{\bz} := \sqrt{n} \left( \bR - \bJ - \frac{2\log n}{n}\bone_2 \right) \in \scS(\bV,\veps).
\end{equation}
We fix a time-sharing sequence $t^n \in \calT^n$ with type $P_{t^n} \in\scP_n(\calT)$ such that
\begin{equation} \label{eq:time-sharing-sequence}
|P_{t^n}(t) - P_T(t)| \le \frac{1}{n}
\end{equation}
for every $t \in \calT$~\cite{Huang12}. Then, we consider the test channel given by $P_{U^n|Y^n}(u^n|y^n) = P_{U|TY}^n(u^n|t^n,y^n)$, and we use
Corollary \ref{cor:fein} for $P_{U^n X^n Y^n} = P_{XY}^n P_{U^n|Y^n}$ by setting $\gamma_{\rmb} = \log |\calM_n| - \log n$,
$\gamma_{\rmc} = \log |\calL_n| - \log n$, and $\delta = \frac{1}{n}$. Then, there exists a WAK code $\Phi_n$ such that
\begin{align}
& 1- \Pe(\Phi_n) \nn\\
&\ge  \Pr\left\{ \sum_{i=1}^n \bj(U_i,X_i,Y_i|t_i) \le n \bR - \log n \bone_2 \right\} - \frac{2}{n} - \sqrt{\frac{1}{n}} \\
&= \Pr\left\{ \frac{1}{\sqrt{n}} \sum_{i=1}^n \left( \bj(U_i,X_i,Y_i|t_i) - \bJ \right) \le \tilde{\bz} + \frac{\log n}{\sqrt{n}} \bone_2 \right\} \nn\\
&~~~ - \frac{2}{n} - \sqrt{\frac{1}{n}}.
\label{eq:lower-bound-correct-probability}
\end{align}
By using Corollary \ref{corollary:multidimensional-berry-esseen} to the first term of \eqref{eq:lower-bound-correct-probability}, we have
\begin{align}
1 - \Pe(\Phi_n) &\ge  \Pr\left\{ \bZ \le \tilde{\bz} + \frac{\log n}{\sqrt{n}} \bone_2 \right\} - O\left( \frac{1}{\sqrt{n}}\right)  \label{eqn:wakbe1}\\
&=\Pr\{ \bZ \le \tilde{\bz} \} + O\left( \frac{\log n}{\sqrt{n}} \right)  \label{eqn:wak1}\\
&\ge 1 - \veps  
\label{eq:second-order-WAK-final-bound}
\end{align}
for sufficiently large $n$, where \eqref{eqn:wak1} follows from the Taylor's approximation, and \eqref{eq:second-order-WAK-final-bound} follows from \eqref{eq:redundancy-of-WAK}.

Next, we consider the case with $\bV$ is singular but not $0$. In this case, we cannot apply Corollary \ref{corollary:multidimensional-berry-esseen}
because $\lambda_{\min}(\bV) = 0$. Since $\mathrm{rank}(\bV) =1$, we can write $\bV = \bv \bv^T$ by using the vector $\bv$.
Let $\bA_i = \bj(U_i,X_i,Y_i|t_i) - \bJ$. Then we can write $\bA_i = \bv B_i$ by using the scalar independent random variables $\{B_i \}_{i=1}^n$.
Thus, by using the ordinary Berry-Ess\'een theorem \cite[Ch.\ XVI]{feller} for $\{B_i \}_{i=1}^n$, we can derive \eqref{eq:second-order-WAK-final-bound}.

Finally, we consider the case where  $\bV = \mathbf{0}$. In this case, by setting $\tilde{\bz} = \mathbf{0}$ in \eqref{eq:lower-bound-correct-probability}, we can find that the
right hand side converges to $1$.  

For the bounds on the cardinalities of auxiliary random variables, see
 Appendix \ref{app:cardinality_bound}.
\end{proof}

\section{Achievability Proof of the Second-Order Coding Rate for    WAK in Theorem~\ref{thm:second-modified}} \label{app:second-modified}
\begin{proof}
We only provide a sketch of the proof because most of the steps are the same as Appendix \ref{app:second}. The only modification is that we use Theorem \ref{thm:dt-helper-binning} instead of Corollary \ref{cor:fein} by setting $\gamma_{\rmb} = \log |\calM_n| - \rho \sqrt{n} - \log n$,
$\gamma_{\rmc} = \log |\calL_n| + \rho \sqrt{n} - \log n$, $J_n = |\calL_n| 2^{\rho \sqrt{n}}$, and $\delta = \frac{1}{n}$. \end{proof}

\section{Achievability Proof of the Second-Order Coding Rate for    WZ in Theorem~\ref{thm:wz_2nd}}\label{app:wz_2nd}

\begin{proof}
It suffices to show the inclusion $\scR_{\mathrm{in}}(n,\veps; P_{UTXY},P_{\hatX|UYT})\subset \scR_{\WZ}(n,\veps)$ for fixed pair $(P_{UTXY},P_{\hatX|UYT})$ of $P_{UTXY}\in\tilde{\scP}(P_{XY} )$ and $P_{\hatX|UYT}$.
We assume that $\bV = \bV(P_{UTXY},P_{\hatX|UYT}) \succ 0$, since the case where $\bV$ is singular can be handled in a similar manner
as Appendix \ref{app:second} (see also \cite[Proof of Theorem 5]{TK12}).

First, note that $[R,D]^T \in \scR_{\mathrm{in}}(n,\veps; P_{UTXY},P_{\hatX|UYT})$ implies 
\begin{align} \label{eq:redundancy-of-WZ}
\tilde{\bz} := \sqrt{n}\left(\left[ \begin{array}{c} - \frac{1}{n} \log \frac{L_n}{|\calM_n|} \\ \frac{1}{n} \log L_n \\ D \end{array} \right]
 - \bJ - \frac{2 \log n}{n} \bone_3 \right) \in \scS(\bV,\veps)
\end{align}
for some positive integer $L_n$.  
We fix a sequence $t^n \in \calT^n$ satisfying \eqref{eq:time-sharing-sequence} for every $t \in \calT$. Then, we consider the test
channel given by $P_{U^n|X^n}(u^n|x^n) = P_{U|TX}^n(u^n|t^n,x^n)$
 and the reproduction channel given by $P_{\hatX^n|U^nY^n}(\hatx^n|u^n,y^n)=P_{\hatX|UYT}^n(\hatx^n|u^n,y^n,t^n)$.
Then, Corollary \ref{thm:fein_wz} for $P_{U^nX^nY^n\hatX^n} = P_{XY}^n P_{U^n|X^n} P_{\hatX^n|U^nY^n}$ with 
 $\gamma_{\rmp} = \log \frac{L_n}{|\calM_n|} + \log n$, $\gamma_{\rmc} = \log L_n - \log n$, and $\delta = \frac{1}{n}$
shows that there exists a WZ code such that
\begin{align}
& 1 - \Pe(\Phi_n;D) \ge \nn\\
&  \Pr\left\{ \sum_{i=1}^n \bj(U_i,X_i,Y_i,\hatX_i|t_i) \le 
	\left[ \begin{array}{c} -\log \frac{L_n}{|\calM_n|} \\ \log L_n \\ n D \end{array}\right] - \log n \bone_3 \right\} \nn\\
&~~~	- \frac{2}{n} - \sqrt{\frac{1}{n}} \\
&=  \Pr\left\{ \frac{1}{\sqrt{n}} \sum_{i=1}^n \left( \bj(U_i,X_i,Y_i,\hatX_i|t_i) - \bJ \right) \le \tilde{\bz} + \frac{\log n}{\sqrt{n}} \bone_3 \right\} \nn\\
&~~~- \frac{2}{n} - \sqrt{\frac{1}{n}}.
\label{eq:lower-bound-correct-probability-WZ}
\end{align} 
Now the rest of the proof proceeds by using the multidimensional
 Berry-Ess\'een theorem as in \eqref{eqn:wakbe1} to
 \eqref{eq:second-order-WAK-final-bound} for the WAK problem. 

For the bounds on the cardinalities of auxiliary random variables, see
 Appendix \ref{app:cardinality_bound}.
\end{proof}

\section{Achievability Proof of the Second-Order Coding Rate for Lossy Source Coding in Theorem~\ref{thm:rd_disp}} \label{app:2nd_wz_rd}

We slightly modify a special case of Corollary \ref{thm:fein_wz} as follows, which will be used in both Appendices \ref{app:proof-method-type} and \ref{app:proof-d-tilted}.

\begin{corollary} \label{thm:fein_rd}
For arbitrary distribution $Q_{\hat{X}} \in \scP(\hat{\calX})$, and for arbitrary constants $\gamma_{\rmc}, \nu \ge 0$ and $\delta, \tilde{\delta} > 0$, there exists a lossy source code $\Phi$ with probability of excess distortion satisfying 
\begin{align}
\Pe(\Phi;D) &\le P_{\hat{X} X}\left[ \log \frac{P_{\hat{X}|X}(\hatx|x)}{Q_{\hat{X}}(\hatx)} > \gamma_{\rmc} - \nu \mbox{ or } \rvd(x, \hatx) >D \right] \nn\\
&~~~ + \tilde{\delta} + \sqrt{\frac{2^{\gamma_{\rmc}}}{\tilde{\delta}|\calM|}} + \delta + 2^{-\nu}.
\end{align}
\end{corollary}
\begin{proof}
As a special case of Corollary \ref{thm:fein_wz}, we have
\begin{align} 
\Pe(\Phi;D) &\le P_{\hat{X} X}\left[ \log \frac{P_{\hat{X}|X}(\hatx|x)}{P_{\hat{X}}(\hatx)} > \gamma_{\rmc}  \mbox{ or } \rvd(x, \hatx) > D\right] \nn\\
&~~~+ \tilde{\delta} + \sqrt{\frac{2^{\gamma_{\rmc}}}{\tilde{\delta}|\calM|}} + \delta,
\label{eq:proof:thm:fein_rd-1}
\end{align}
where we set $\gamma_{\rmp} = 0$ and $L = \tilde{\delta} |\calM|$. We can further upper bound the first term of \eqref{eq:proof:thm:fein_rd-1} as
\begin{align}
& P_{\hat{X} X}\left[ \log \frac{P_{\hat{X}|X}(\hatx|x)}{P_{\hat{X}}(\hatx)} > \gamma_{\rmc}  \mbox{ or } \rvd(x, \hatx) > D \right]   \\
&= P_{\hat{X} X}\bigg[ \log \frac{P_{\hat{X}|X}(\hatx|x)}{Q_{\hat{X}}(\hatx)} + \log \frac{Q_{\hat{X}}(\hatx)}{P_{\hat{X}}(\hatx)} > \gamma_{\rmc} \nn\\
&~~~~~~~~~~~~ \mbox{ or } \rvd(x, \hatx) > D \bigg] \\
&\le P_{\hat{X} X}\bigg[ \log \frac{P_{\hat{X}|X}(\hatx|x)}{Q_{\hat{X}}(\hatx)} > \gamma_{\rmc} -\nu \mbox{ or }  \log \frac{Q_{\hat{X}}(\hatx)}{P_{\hat{X}}(\hatx)} > \nu \nn\\
&~~~~~~~~~~~~  \mbox{ or } \rvd(x, \hatx) > D \bigg] \\
&\le P_{\hat{X} X}\left[ \log \frac{P_{\hat{X}|X}(\hatx|x)}{Q_{\hat{X}}(\hatx)} > \gamma_{\rmc} -\nu  \mbox{ or } \rvd(x, \hatx) > D \right] \nn\\
&~~~	+ P_{\hat{X} X}\left[ \log \frac{Q_{\hat{X}}(\hatx)}{P_{\hat{X}}(\hatx)} > \nu \right] \\
&=  P_{\hat{X} X}\left[ \log \frac{P_{\hat{X}|X}(\hatx|x)}{Q_{\hat{X}}(\hatx)} > \gamma_{\rmc} -\nu  \mbox{ or } \rvd(x, \hatx) > D \right] \nn\\
&~~~	+ P_{\hat{X} }\left[ \log \frac{Q_{\hat{X}}(\hatx)}{P_{\hat{X}}(\hatx)} > \nu \right] \\
&\le  P_{\hat{X} X}\left[ \log \frac{P_{\hat{X}|X}(\hatx|x)}{Q_{\hat{X}}(\hatx)} > \gamma_{\rmc} -\nu  \mbox{ or } \rvd(x, \hatx) > D \right]
	+ 2^{-\nu}.
\end{align}
This completes the proof.
\end{proof}

\begin{remark}
By showing Corollary \ref{thm:fein_rd} directly instead of via Corollary \ref{thm:fein_wz}, we can eliminate the residual term $\tilde{\delta}$. 
\end{remark}

\subsection{Proof Based on the Method of Types } \label{app:proof-method-type}

To prove Theorem \ref{thm:rd_disp} by the method of types, we use the following lemma.
\begin{lemma}[Rate-Redundancy \cite{ingber11}] \label{lemma:rate-redundancy}
Suppose that $R(P_X,D)$ is differentiable w.r.t.~$D$ and twice differentiable w.r.t.~$P_X$ at some neighbourhood of $(P_X,D)$.
Let $\varepsilon$ be given probability and let $\Delta R$ be any quantity chosen such that
\begin{eqnarray}
P_X^n \left[ R(P_{x^n}, D) - R(P_X,D) > \Delta R \right] = \varepsilon + g_n ,
\end{eqnarray}
where $g_n = O\left(\frac{\log n}{\sqrt{n}}\right)$.   Then, as $n$ grows, 
\begin{eqnarray}
\Delta R = \sqrt{ \frac{\var(j(X,D))}{n}} Q^{-1}(\varepsilon) + O\left( \frac{\log n}{n} \right).
\end{eqnarray}
\end{lemma}
Note that the quantity $j(x,D)$ has an alternative representation as the derivative of $Q\mapsto R(Q,D)$ with respect to $Q(x)$ evaluated at $P_X(x)$; cf.~\eqref{eqn:der_rd}.

We also use the following lemma, which is a consequence of the argument right after \cite[Theorem 1]{Yu93}.
\begin{lemma} \label{lemma:approximation-of-rd}
For a type $q \in \scP_n(\calX)$, suppose that $\left| \frac{\partial R(q,D)}{\partial D} \right| < C$ for a constant $C>0$
in some neighbourhood of $q$. Then, there exists a test channel $V \in \scV_n(\calY;q)$ such that
\begin{eqnarray} \label{eq:approximation-of-rd-distortion-condition}
\sum_{x,\hatx} q(x) V(\hatx|x) \rvd(x, \hatx) \le D
\end{eqnarray}
and 
\begin{eqnarray}
I(q,V) \le R(q,D) + \frac{\tau}{n},
\end{eqnarray}
where $\tau$ is a constant depending on $C$, $|\calX|$, $|\hat{\calX}|$, and $D_{\max}$.
\end{lemma}

Using Lemmas~\ref{lemma:rate-redundancy} and \ref{lemma:approximation-of-rd}, we prove Theorem~\ref{thm:rd_disp}. 
\begin{proof}
We construct a test channel $P_{\hat{X}^n|X^n}$ as follows. For a fixed constant $\tilde{\tau} > 0$, we set 
\begin{eqnarray}
\Omega_n = \left\{ q \in \scP_n(\calX) : \|P_x -q \|^2 \le \frac{\tilde{\tau} \log n}{n}  \right\}.
\end{eqnarray} 
Since we assumed that $R(P_X,D)$ is differentiable w.r.t.~$D$ at $P_X$, the derivative is bounded over any small enough neighbourhood of $P_X$. In particular, it is bounded by some constant $C$ over $\Omega_n$ for sufficiently large $n$. For each $q \in \Omega_n$, we choose test channel $V_q \in\scV_n(\calY;q)$ satisfying the statement of Lemma \ref{lemma:approximation-of-rd}. Then, we define the test channel
\begin{eqnarray}
P_{\hat{X}^n|X^n}(\hatx^n|x^n) = \left\{ 
\begin{array}{ll}
\frac{1}{|\calT_{V_{P_{x^n}}}(x^n)|} & \mbox{if } \hatx^n \in \calT_{V_{P_{x^n}}}(x^n) \\
0 & \mbox{else}
\end{array}
\right.
\end{eqnarray}
for $x^n$ satisfying $P_{x^n} \in \Omega_n$, and otherwise we define $P_{\hat{X}^n|X^n}(\hatx^n|x^n)$ arbitrarily as long as the channel only outputs $\hatx^n$ satisfying $\rvd_n(x^n,\hatx^n) \le D$. Let $P_q \in \scP_n(\hat{\calX})$ be such that 
\begin{eqnarray}
P_q(\hat{x}) = \sum_x q(x) V_q(\hat{x}|x).
\end{eqnarray} 
Then, let $\tilde{P}^n_q \in \scP(\hat{\calX}^n)$ be the uniform distribution on $\calT_{P_q}$. Furthermore, let $Q_{\hat{X}^n} \in \scP(\hat{\calX}^n)$ be the distribution given by
\begin{eqnarray}
Q_{\hat{X}^n}(\hatx^n) = \sum_{q \in \Omega_n} \frac{1}{|\Omega_n|} \tilde{P}_q^n(\hatx^n).
\end{eqnarray}

We now use Corollary \ref{thm:fein_rd} for $P_X = P_X^n$, $P_{\hat{X}|X} = P_{\hat{X}^n|X^n}$, and $Q_{\hat{X}} = Q_{\hat{X}^n}$. Then, by noting that 
\begin{eqnarray}
\rvd_n(x^n,\hatx^n) = \sum_{x, \hat{x}} P_{x^n}(x) V_{P_{x^n}}(\hatx|x) \rvd(x,\hatx) > D
\end{eqnarray}
never occurs for the test channel $P_{\hat{X}^n|X^n}$, we have
\begin{align}
\Pe(\Phi_n;D) &\le P_{\hat{X}^nX^n}\left[ \log \frac{P_{\hat{X}^n|X^n}(\hatx^n|x^n)}{Q_{\hat{X}^n}(\hatx^n)} > \gamma_{\rmc} - \nu \right] \nn\\
&~~~	+ \tilde{\delta} + \sqrt{\frac{2^{\gamma_{\rmc}}}{\tilde{\delta} |\calM_n|}} + \delta + 2^{-\nu} \\
&= P_{\hat{X}^nX^n}\left[ \frac{1}{n} \log \frac{P_{\hat{X}^n|X^n}(\hatx^n|x^n)}{Q_{\hat{X}^n}(\hatx^n)} > \tilde{\gamma} - \frac{\log n}{n} \right] \nn\\
&~~~	+ \sqrt{\frac{n 2^{\tilde{\gamma} n}}{|\calM_n|}} + \frac{3}{n},
\end{align}
where we set $\gamma_{\rmc} = \tilde{\gamma} n$, $\tilde{\delta} = \delta = \frac{1}{n}$, and $\nu = \log n$.
Furthermore, by noting that 
\begin{eqnarray} \label{eq:proof-rd-type-1}
Q_{\hat{X}^n}(\hatx^n) \ge \frac{1}{|\Omega_n|} \tilde{P}_q^n(\hatx^n)
\end{eqnarray}
for any $q \in \Omega_n$, we have
\begin{align}
& P_{\hat{X}^nX^n}\left[ \frac{1}{n} \log \frac{P_{\hat{X}^n|X^n}(\hatx^n|x^n)}{Q_{\hat{X}^n}(\hatx^n)} > \tilde{\gamma} - \frac{\log n}{n} \right]  \\
&\le P_{\hat{X}^nX^n}\left[ \frac{1}{n} \log \frac{P_{\hat{X}^n|X^n}(\hatx^n|x^n)}{Q_{\hat{X}^n}(\hatx^n)} > \tilde{\gamma} - \frac{\log n}{n}, P_{x^n} \in \Omega_n \right] \nn\\
&~~~+ P_{X^n}[ P_{x^n} \notin \Omega_n ] \\
&\le  P_{\hat{X}^nX^n}\left[ \frac{1}{n} \log \frac{P_{\hat{X}^n|X^n}(\hatx^n|x^n)}{Q_{\hat{X}^n}(\hatx^n)} > \tilde{\gamma} - \frac{\log n}{n}, P_{x^n} \in \Omega_n \right]  \nn\\
&~~~+ \frac{2 \tilde{\tau}}{n^2}  \label{eqn:types1}\\
&\le   P_{\hat{X}^nX^n}\bigg[ \frac{1}{n} \log \frac{P_{\hat{X}^n|X^n}(\hatx^n|x^n)}{\tilde{P}_{P_{x^n}}^n(\hatx^n)} \nn\\
&~~~~~~~~~~~~~ > \tilde{\gamma} - \frac{\log n}{n} - \frac{|\calX| \log(n+1)}{n}, P_{x^n} \in \Omega_n \bigg]  + \frac{2 \tilde{\tau}}{n^2}, \label{eqn:types2}
\end{align}
where \eqref{eqn:types1} follows from \cite[Lemma 2]{ingber11} and \eqref{eqn:types2} follows from \eqref{eq:proof-rd-type-1} and the fact that $|\Omega_n| \le |\scP_n(\calX)| \le (n+1)^{|\calX|}$.

Furthermore, we also have
\begin{align}
\log \frac{P_{\hat{X}^n|X^n}(\hatx^n|x^n)}{\tilde{P}_{P_{x^n}}^n(\hatx^n)} 
&= \log \frac{|\calT_{P_{P_{x^n}}}|}{|\calT_{V_{P_{x^n}}(x^n)}|} \\
&= n I(P_{x^n}, V_{P_{x^n}}) + O(\log n).
\end{align}
Thus, for $\mu_n = O\left(\frac{\log n}{n}\right)$, we have
\begin{align}
\Pe(\Phi_n;D) &\le P_{\hat{X}^nX^n}\left[ I(P_{x^n},V_{P_{x^n}}) > \tilde{\gamma} - \mu_n, P_{x^n} \in \Omega_n \right] \nn\\
&~~~+ O\left( \frac{1}{n}\right) + \sqrt{\frac{n 2^{\tilde{\gamma} n}}{|\calM_n|}} \\
&\le P_{\hat{X}^nX^n}\left[ R(P_{x^n},D) > \tilde{\gamma} - \mu_n - \frac{\tau}{n}, P_{x^n} \in \Omega_n \right] \nn\\
&~~~ + O\left( \frac{1}{n}\right) + \sqrt{\frac{n 2^{\tilde{\gamma} n}}{|\calM_n|}} \\
&\le P_{\hat{X}^nX^n}\left[ R(P_{x^n},D) > \tilde{\gamma} - \mu_n - \frac{\tau}{n} \right] \nn\\
&~~~+ O\left( \frac{1}{n}\right) + \sqrt{\frac{n 2^{\tilde{\gamma} n}}{|\calM_n|}}\\
&\le P_{X^n}\left[ R(P_{x^n},D) > \tilde{\gamma} - \mu_n - \frac{\tau}{n} \right] \nn\\
&~~~+ O\left( \frac{1}{n}\right) + \sqrt{\frac{n 2^{\tilde{\gamma} n}}{|\calM_n|}}. \label{eqn:types_proof_end}
\end{align}
Thus, by setting $\tilde{\gamma} = R(P_X,D) + \Delta R$, $\frac{1}{n} \log |\calM_n| = \tilde{\gamma} + \frac{2 \log n}{n}$ and by using Lemma \ref{lemma:rate-redundancy} (with $g_n = O\left(\frac{\log n}{\sqrt{n}}\right)$ being the residual terms in~\eqref{eqn:types_proof_end}), we have
\begin{align}
R(n,\varepsilon;D) &\le R(P_X,D) + \sqrt{\frac{\var(j(X,D))}{n}} Q^{-1}(\varepsilon) \nn\\
&~~~+ O\left( \frac{\log n}{n}\right)
 \label{eq:rd-gaussian-approximation}
\end{align}
for sufficiently large $n$, which implies the statement of the theorem. \end{proof}

\subsection{Proof Based on the $D$-tilted Information} \label{app:proof-d-tilted}

Let 
\begin{eqnarray}
\calB_D(x^n) := \left\{ \hatx^n : \rvd_n(x^n,\hatx^n) \le D \right\}
\end{eqnarray}
be the $D$-sphere, and let $P_{\hat{X}^\star}$ be the output distribution of the optimal test channel of
\begin{eqnarray}
\min_{P_{\hat{X}|X} \atop \bbE[\rvd(X,\hat{X})] \le D} I(X;\hat{X}).
\end{eqnarray}
To prove Theorem \ref{thm:rd_disp} by the $D$-tilted information, we use the following lemma.
\begin{lemma}[Lemma 2 of \cite{kost12}] \label{lemma:kostina-verdu}
Under some regularity conditions, which are explicitly given in \cite[Lemma 2]{kost12} and satisfied by discrete memoryless sources, there exists constants $n_0,c,K > 0$ such that
\begin{align}
& P_X^n\left[ \log \frac{1}{P_{\hat{X}^\star}^n(\calB_D(x^n))} \le \sum_{i=1}^n j(x_i,D) + C \log n + c  \right] \nn\\
&\ge 1 - \frac{K}{\sqrt{n}}
\end{align}
for all $n \ge n_0$, where $C>0$ is a constant given by \cite[Equation~(86)]{kost12}. 
\end{lemma}

\begin{proof}
We construct test channel $P_{\hat{X}^n|X^n}$ as 
\begin{eqnarray}
P_{\hat{X}^n|X^n}(\hatx^n|x^n) = \left\{ 
\begin{array}{ll}
\frac{P_{\hat{X}^\star}^n(\hatx^n)}{P_{\hat{X}^\star}^n(\calB_D(x^n))} & \mbox{if } \hatx^n \in \calB_D(x^n) \\
0 & \mbox{else}
\end{array}
\right..
\end{eqnarray}
We now use Corollary \ref{thm:fein_rd} for $P_X = P_X^n$, $P_{\hat{X}|X} = P_{\hat{X}^n|X^n}$, $Q_{\hat{X}} = P_{\hat{X}^\star}^n$, $\gamma_{\rmc} = \tilde{\gamma} n$, $\tilde{\delta} = \delta = \frac{1}{n}$ and $\nu = \log n$. Then, by noting that $\rvd_n(x^n,\hatx^n) > D$ never occur for the test channel $P_{\hat{X}^n|X^n}$, we have 
\begin{align}
& \Pe(\Phi_n;D) \nn\\
&\le P_{\hat{X}^n X^n}\left[ \log \frac{P_{\hat{X}^n|X^n}(\hatx^n|x^n)}{P_{\hat{X}^\star}^n(\hatx^n)} > \tilde{\gamma} n - \log n \right] \nn\\
&~~~+ \sqrt{\frac{n 2^{\tilde{\gamma}n}}{|\calM_n|}} + \frac{3}{n} \\
&\le P_{\hat{X}^n X^n}\left[ \log \frac{1}{P_{\hat{X}^\star}^n(\calB_D(x^n))} > \tilde{\gamma} n - \log n \right] \nn\\
&~~~+ \sqrt{\frac{n 2^{\tilde{\gamma}n}}{|\calM_n|}} + \frac{3}{n} \\
&= P_X^n\left[ \log \frac{1}{P_{\hat{X}^\star}^n(\calB_D(x^n))} > \tilde{\gamma} n - \log n \right] + \sqrt{\frac{n 2^{\tilde{\gamma}n}}{|\calM_n|}} + \frac{3}{n} \\
&\le P_X^n\left[ \sum_{i=1}^n j(x_i,D) > \tilde{\gamma}n - (C+1) \log n - c \right] \nn\\
&~~~ + P_X^n\left[ \log \frac{1}{P_{\hat{X}^\star}^n(\calB_D(x^n))} > \sum_{i=1}^n j(x_i,D) + C \log n + c \right] \nn\\
&~~~+ \sqrt{\frac{n 2^{\tilde{\gamma}n}}{|\calM_n|}} + \frac{3}{n} \\
&\le P_X^n\left[ \sum_{i=1}^n j(x_i,D) > \tilde{\gamma}n - (C+1) \log n - c \right] \nn\\
&~~~+ \frac{K}{\sqrt{n}}  + \sqrt{\frac{n 2^{\tilde{\gamma}n}}{|\calM_n|}} + \frac{3}{n}, \label{eqn:dtilt}
\end{align}
where \eqref{eqn:dtilt} follows from Lemma \ref{lemma:kostina-verdu}. Thus, by setting $\tilde{\gamma} = \frac{1}{n} \log |\calM_n| - \frac{2 \log n}{n}$ and by applying the Berry-Ess\'een theorem \cite{feller}, we have \eqref{eq:rd-gaussian-approximation} for sufficiently large $n$, which implies the statement of the theorem. \end{proof}

\newcommand{\Var}{\mathsf{Var}}
\newcommand{\Cov}{\mathsf{Cov}}

\section{Cardinality Bound for Second-Order Coding Theorems}\label{app:cardinality_bound}
The following three theorems allow us to restrict the cardinalities of
auxiliary random variables in second-order coding theorems.

\begin{theorem}
[Cardinality Bound for WAK]\label{thm:cardinality_bound_WAK}
For any $P_{UTXY}\in\tilde{\scP}(P_{XY})$, where $\tilde{\scP}(P_{XY})$
 is defined in \ref{sec:2nd_wak}, there exists
$P_{U'T'XY}$ with
$\lvert\calU'\rvert\leq \lvert\calY\rvert+4$ and
$\lvert\calT'\rvert\leq 5$ such that
(i) $\calX\times\calY$-marginal of $P_{U'T'XY}$ is $P_{XY}$,
(ii) $U'-(Y,T')-X$ forms a Markov chain, 
(iii) $T'$ is independent of $(X,Y)$, 
and (iv) $P_{U'T'XY}$ preserves the mean $\bJ$ of the entropy-information density vector and the
 entropy-information dispersion matrix $\bV$, i.e.,
\begin{align}
 \bJ(P_{UTXY})&= \bJ(P_{U'T'XY})\\
 \bV(P_{UTXY})&= \bV(P_{U'T'XY}).
\end{align} 
\end{theorem}

\begin{theorem}
[Cardinality Bound for WZ]\label{thm:cardinality_bound_WZ}
For any pair of $P_{UTXY}\in\tilde{\scP}(P_{XY})$
and $P_{\hatX|UYT}$, where $\tilde{\scP}(P_{XY})$
 is defined in \ref{sec:2nd_wz}, there exist
$P_{U'T'XY}$ 
and $P_{\hatX'|U'YT'}\colon\calU'\times\calY\times\calT'\to\hatcalX$
with
$\lvert\calU'\rvert\leq \lvert\calY\rvert+8$ and
$\lvert\calT'\rvert\leq 9$ such that
(i) $\calX\times\calY$-marginal of $P_{U'T'XY}$ is $P_{XY}$,
(ii) $U'-(X,T')-Y$ forms a Markov chain, 
(iii) $T'$ is independent of $(X,Y)$, 
and (iv) $P_{U'T'XY}$ and $P_{\hatX'|U'YT'}$ preserve 
$\bJ$ and $\bV$, i.e.,
\begin{align}
 \bJ(P_{UTXY},P_{\hatX|UYT})&= \bJ(P_{U'T'XY},P_{\hatX'|U'YT'})\\
 \bV(P_{UTXY},P_{\hatX|UYT})&= \bV(P_{U'T'XY},P_{\hatX'|U'YT'}).
\end{align} 
\end{theorem}

\begin{theorem}
[Cardinality Bound for GP]\label{thm:cardinality_bound_GP}
For any $P_{UTSXY}\in\tilde{\scP}(W,P_S)$, where $\tilde{\scP}(W,P_S)$
 is defined in \ref{sec:2nd_gp}, there exists
$P_{U'T'SXY}$ 
with
$\lvert\calU'\rvert\leq \lvert\calY\rvert+6$ and
$\lvert\calT'\rvert\leq 9$ such that
(i) $\calS\times\calX\times\calY$-marginal of $P_{U'T'SXY}$ is $P_{SXY}$,
(ii) $U'-(X,S,T')-Y$ forms a Markov chain, 
(iii) $T'$ is independent of $S$, 
and (iv) $P_{U'T'SXY}$ preserves $\bJ$ and $\bV$, i.e.,
\begin{align}
 \bJ(P_{UTSXY})&= \bJ(P_{U'T'SXY})\\
 \bV(P_{UTSXY})&= \bV(P_{U'T'SXY}).
\end{align} 
\end{theorem}

We can prove all of the three theorems in the same manner.  Because the
proof for Wyner-Ziv problem is most complicated, we
prove Theorem \ref{thm:cardinality_bound_WZ} in
\ref{app:proof_cardinality_WZ}, and then, give proof sketches for
Theorems \ref{thm:cardinality_bound_WAK} and
\ref{thm:cardinality_bound_GP} in \ref{app:proof_cardinality_WAK-GP}.

\subsection{Proof of Cardinality Bound for WZ problem}\label{app:proof_cardinality_WZ}
To prove Theorem \ref{thm:cardinality_bound_WZ}, we use variations of
the support lemma. Note that we can identify
$\scP(\calX)\times\scP(\hatcalX|\calY)$ with a connected compact subset
of $\lvert\calX\rvert\lvert\hatcalX\rvert\lvert\calY\rvert$-dimensional
Euclidean space.  Hence, as a consequence of the
Fenchel-Eggleston-Carath\'eodory theorem (see, e.g.~\cite[Appendix
A]{elgamal}), we have the following lemma.

\begin{lemma}
\label{support_lemma_U}
Let $f_j$ ($j=1,2,\dots,k$) be real-valued continuous functions on
 $\scP(\calX)\times\scP(\hatcalX|\calY)$.
Then, for any $P_U\in\scP(\calU)$
and any collection 
$\{(P_{X|U}(\cdot|u),P_{\hatX|YU}(\cdot|\cdot,u)):u\in\calU\}\subset\scP(\calX)\times\scP(\hatcalX|\calY)$,
there exist a distribution $P_{U'}\in\scP(\calU')$ with $\lvert\calU'\rvert\leq k$
and a collection 
$\{(P_{X'|U'}(\cdot|u'),P_{\hatX'|Y'U'}(\cdot|\cdot,u')):u'\in\calU'\}\subset\scP(\calX)\times\scP(\hatcalX|\calY)$
such that for $j=1,2,\dots,k$,
\begin{align}
& \int_{\calU}f_j\left(P_{X|U}(\cdot|u),P_{\hatX|YU}(\cdot|\cdot,u)\right)dP_U(u) \nn\\
 &=\sum_{u'\in\calU'}f_j\left(P_{X'|U'}(\cdot|u'),P_{\hatX'|Y'U'}(\cdot|\cdot,u')\right)P_{U'}(u').
\label{eq:support_lemma_U}
\end{align}
\end{lemma}

\begin{remark}
\label{remark:support_lemma}
Let us consider applying Lemma \ref{support_lemma_U} to 
a case where $P_{\hatX|YU}$ is a deterministic function.
In this case, 
$P_U$ appearing in the left hand side of 
\eqref{eq:support_lemma_U} satisfies
 $P_U(u)>0$ only if 
$P_{\hatX|YU}(\cdot|\cdot,u)$ is deterministic, i.e., 
for each $y$ there exists $\hatx$ satisfying $P_{\hatX|YU}(\hatx|y,u)=1$.
On the other hand, 
Lemma \ref{support_lemma_U} does not guarantee that 
we can choose $\calU'$ and a collection of distributions 
so that  $P_{\hatX'|Y'U'}(\cdot|\cdot,u')\in\calP(\calY|\calX)$ is deterministic for all $u'\in\calU'$.
That is why 
we use a stochastic reproduction function to establish bounds on the
cardinalities of the auxiliary random variables.
\end{remark}

Similarly, by identifying
$\scP(\calU|\calX)\times\scP(\hatcalX|\calU\times\calY)$ with a
connected compact subset of Euclidean space, we have another variation
of the support lemma.

\begin{lemma}
\label{support_lemma_T}
Let $f_j$ ($j=1,2,\dots,k$) be real-valued continuous functions on $\scP(\calU|\calX)\times\scP(\hatcalX|\calU\times\calY)$.
Then, for any $P_T\in\scP(\calT)$
and any collection 
$\{(P_{U|XT}(\cdot|\cdot,t),P_{\hatX|UYT}(\cdot|\cdot,\cdot,t)):t\in\calT\}\subset\scP(\calU|\calX)\times\scP(\hatcalX|\calU\times\calY)$,
there exist a distribution $P_{T'}\in\scP(\calT')$ with $\lvert\calT'\rvert\leq k$
and a collection 
$\{(P_{U'|X'T'}(\cdot|\cdot,t'),P_{\hatX'|U'Y'T'}(\cdot|\cdot,\cdot,t')):t'\in\calT'\}\subset\scP(\calU|\calX)\times\scP(\hatcalX|\calU\times\calY)$
such that for $j=1,2,\dots,k$,
\begin{align}
& \int_{\calT}f_j\left(P_{U|XT}(\cdot|\cdot,t),P_{\hatX|UYT}(\cdot|\cdot,\cdot,t)\right)dP_T(t) \nn\\
 &=\sum_{t'\in\calT'}f_j\left(P_{U'|X'T'}(\cdot|\cdot,t'),P_{\hatX'|U'Y'T'}(\cdot|\cdot,\cdot,t')\right)P_{T'}(t').
\label{eq:support_lemma_T}
\end{align}
\end{lemma}

\begin{proof}
[Proof of Theorem \ref{thm:cardinality_bound_WZ}]

1) Bound on $\lvert\calU'\rvert$:
Fix $P_{UTXY}\in\tilde{\scP}(P_{XY})$. 
Without loss of generality, we assume that
$\calX=\{1,2,\dots,\lvert\calX\rvert\}$.
Let us consider the
following $\lvert\calX\rvert+8$ functions:
For $(Q,q)\in\scP(\calX)\times\scP(\hatcalX|\calY)$, 
\begin{align}
 f_j(Q,q)&:=Q(j),\quad j=1,2,\dots,\lvert\calX\rvert-1\label{eq:preserve_marginal_distrib}\\
 f_{\lvert\calX\rvert}(Q,q) &:=
 -\sum_{y\in\calY}\left[\sum_{x\in\calX}P_{Y|X}(y|x)Q(x)\right] \nn\\
 &~~~\times \log
\left[\sum_{x\in\calX}P_{Y|X}(y|x)Q(x)\right]\\
 f_{\lvert\calX\rvert+1}(Q,q)&:=
 -\sum_{x\in\calX}Q(x)\log Q(x)\\
 f_{\lvert\calY\rvert+2}(Q,q)&:=
\sum_{x\in\calX}\sum_{y\in\calY}\sum_{\hatx\in\hatcalX}Q(x)P_{Y|X}(y|x)
q(\hatx|y)
\rvd\left(
x, \hatx
\right)\\
 f_{\lvert\calX\rvert+3}(Q,q)&:=
 \sum_{y\in\calY}\left[\sum_{x\in\calX}P_{Y|X}(y|x)Q(x)\right] \nn\\
&~~~\times \left\{
\log\frac{\left[\sum_{x\in\calX}P_{Y|X}(y|x)Q(x)\right]}{P_Y(y)}
\right\}^2\\
 f_{\lvert\calX\rvert+4}(Q,q)&:=
 \sum_{x\in\calX}Q(x)\left\{\log\frac{Q(x)}{P_X(x)}\right\}^2\\
 f_{\lvert\calX\rvert+5}(Q,q)&:=
\sum_{x\in\calX}\sum_{y\in\calY}\sum_{\hatx\in\hatcalX}Q(x)P_{Y|X}(y|x)
q(\hatx|y) \nn\\
&~~~\times \left\{
\rvd\left(
x, \hatx
\right)\right\}^2\\
 f_{\lvert\calX\rvert+6}(Q,q)&:=
\sum_{x\in\calX}\sum_{y\in\calY}Q(x)P_{Y|X}(y|x) \nn\\
&~~~\times \left(
\log\frac{P_Y(y)}{\sum_{\barx\in\calX}Q(\barx)P_{Y|X}(y|\barx)}
\right) \nn\\
&~~~\times \left(
\log\frac{Q(x)}{P_X(x)}
\right)\\
 f_{\lvert\calX\rvert+7}(Q,q)&:=
\sum_{x\in\calX}\sum_{y\in\calY}\sum_{\hatx\in\hatcalX}
Q(x)P_{Y|X}(y|x)q(\hatx|y) \nn\\
&~~~\times \left(
\log\frac{P_Y(y)}{\sum_{\barx\in\calX}Q(\barx)P_{Y|X}(y|\barx)}
\right)
\rvd\left(
x, \hatx
\right)\\
 f_{\lvert\calX\rvert+8}(Q,q)&:=
\sum_{x\in\calX}\sum_{y\in\calY}\sum_{\hatx\in\hatcalX}
Q(x)P_{Y|X}(y|x)q(\hatx|y) \nn\\
&~~~\times \left(
\log\frac{Q(x)}{P_X(x)}
\right)
\rvd\left(
x, \hatx
\right).
\end{align}

Fix $t\in\calT$. Then, Lemma \ref{support_lemma_U} guarantees that there
exist $P_{U'|T}(\cdot|t)\in\scP(\calU')$ with
$\lvert\calU'\rvert\leq\lvert\calX\rvert+8$ and a collection
$\{(P_{X'|U'T}(\cdot|u',t),P_{\hatX'|Y'U'T}(\cdot|\cdot,u',t)):u'\in\calU'\}\subset\scP(\calX)\times\scP(\hatcalX|\calY)$
 such that for all $j=1,2,\dots,\lvert\calX\rvert+8$,
\begin{align}
 & \sum_{u\in\calU}f_j\left(P_{X|UT}(\cdot|u,t),P_{\hatX|YUT}(\cdot|\cdot,u,t)\right)P_{U|T}(u|t) \nn\\
 &=\sum_{u'\in\calU'}f_j\left(P_{X'|U'T}(\cdot|u',t),P_{\hatX'|Y'U'T}(\cdot|\cdot,u',t)\right)P_{U'|T}(u'|t).
\label{eq1:cardinal_proof_wz}
\end{align}

Now, we have $P_{U'|T}, P_{X'|U'T}, P_{\hatX'|Y'U'T}$ satisfying
 \eqref{eq1:cardinal_proof_wz} for each $t\in\calT$.
Let $U',T,X',Y',\hatX'$ be random variables induced by $P_{U'|T},
 P_{X'|U'T}, P_{\hatX'|Y'U'T}$, and $P_{Y|X}, P_T$, i.e., 
 for each $(u',t,x,y,\hatx)\in\calU'\times\calT\times\calX\times\calY\times\hatcalX$,
\begin{align}
& P_{U'TX'Y'\hatX'}(u',t,x,y,\hatx) \nn\\
&:=P_T(t)P_{U'|T}(u'|t)P_{X'|U'T}(x|u',t)P_{Y|X}(y|x) \nn\\
&~~~\times P_{\hatX'|Y'U'T}(\hatx|y,u',t).
\end{align}
Observe that $U'-(T',Y')-X'$ forms a Markov chain and that $T$ is independent
of $(X',Y')$.
Further, \eqref{eq1:cardinal_proof_wz} with $j=1,\dots,\lvert\calX\rvert-1$
guarantees that $P_{X'Y'}=P_{XY}$. Hence, we have $P_{TX'Y'}=P_TP_{XY}$, and thus, we can write $P_{U'TX'Y'\hatX'}=P_{U'TXY\hatX'}$.

On the other hand, some calculations show that, for each $t\in\calT$,
\begin{align}
 & H(Y|U,T=t) \nn\\
 &=\sum_{u\in\calU}f_{\lvert\calX\rvert}(P_{X|UT}(\cdot|u,t),P_{\hatX|YUT}(\cdot|\cdot,u,t))P_{U|T}(u|t)\\
\label{eq2:cardinal_proof_wz}
 & H(X|U,T=t) \nn\\
 &=\sum_{u\in\calU}f_{\lvert\calX\rvert+1}(P_{X|UT}(\cdot|u,t),P_{\hatX|YUT}(\cdot|\cdot,u,t))P_{U|T}(u|t)\\
 & \bbE[\rvd(X,\hatX|t)] \nn\\
 &=\sum_{u\in\calU}f_{\lvert\calX\rvert+2}(P_{X|UT}(\cdot|u,t),P_{\hatX|YUT}(\cdot|\cdot,u,t))P_{U|T}(u|t)\\
 &\Var\left(
-\log \frac{P_{Y|UT}(Y|U,t)}{P_Y(Y)}
\right)  \nn\\
&=\sum_{u\in\calU}f_{\lvert\calX\rvert+3}(P_{X|UT}(\cdot|u,t),P_{\hatX|YUT}(\cdot|\cdot,u,t))P_{U|T}(u|t) \nn\\
&~~~ -\left\{H(Y)-H(Y|U,T=t)\right\}^2\\
 & \Var\left(
\log \frac{P_{X|UT}(X|U,t)}{P_{Y}(Y)}
\right) \nn\\
&=\sum_{u\in\calU}f_{\lvert\calX\rvert+4}(P_{X|UT}(\cdot|u,t),P_{\hatX|YUT}(\cdot|\cdot,u,t))P_{U|T}(u|t) \nn\\
&~~~ - \left\{H(X)-H(X|U,T=t)\right\}^2\\
 & \Var\left(
\rvd(X,\hatX|t)
\right) \nn\\
&=\sum_{u\in\calU}f_{\lvert\calX\rvert+5}(P_{X|UT}(\cdot|u,t),P_{\hatX|YUT}(\cdot|\cdot,u,t))P_{U|T}(u|t) \nn\\
&~~~ -\bbE[\rvd(X,\hatX|t)]^2
\end{align}
and
\begin{align}
& \Cov\left(
-\log \frac{P_{Y|UT}(Y|U,t)}{P_{Y}(Y)}, \log \frac{P_{X|UT}(X|U,t)}{P_{X}(X)}
\right) \nn\\
&=\sum_{u\in\calU}f_{\lvert\calX\rvert+6}(P_{X|UT}(\cdot|u,t),P_{\hatX|YUT}(\cdot|\cdot,u,t))P_{U|T}(u|t)\nn\\
&~~~
+\left\{H(Y)-H(Y|U,T=t)\right\}\left\{H(X)-H(X|U,T=t)\right\},\\
& \Cov\left(
-\log \frac{P_{Y|UT}(Y|U,t)}{P_{Y}(Y)}, \rvd(X,\hatX|t)
\right) \nn\\
&=\sum_{u\in\calU}f_{\lvert\calX\rvert+7}(P_{X|UT}(\cdot|u,t),P_{\hatX|YUT}(\cdot|\cdot,u,t))P_{U|T}(u|t)\nn\\
&~~~
+\left\{H(Y)-H(Y|U,T=t)\right\}\bbE[\rvd(X,\hatX|t)],\\
& \Cov\left(
\log \frac{P_{X|UT}(Y|U,t)}{P_{X}(X)}, \rvd(X,\hatX|t)
\right) \nn\\
&=\sum_{u\in\calU}f_{\lvert\calX\rvert+8}(P_{X|UT}(\cdot|u,t),P_{\hatX|YUT}(\cdot|\cdot,u,t))P_{U|T}(u|t)\nn\\
&~~~
-\left\{H(X)-H(X|U,T=t)\right\}\bbE[\rvd(X,\hatX|t)].
\label{eq3:cardinal_proof_wz}
\end{align}
Thus, equations \eqref{eq1:cardinal_proof_wz} and
\eqref{eq2:cardinal_proof_wz}--\eqref{eq3:cardinal_proof_wz} guarantee
that a pair $P_{U'XY\hatX'|T=t}$ preserves all components of $\bJ$ and $\bV$
for each $t\in\calT$.  By taking the average with respect to $T$, we can
show that the pair $(P_{U'TXY},P_{\hatX'|U'YT})$ satisfies the all conditions of the theorem
except the cardinality of $T$.

2) Bound on $\lvert\calT'\rvert$:
Fix $P_{UTXY\hatX}\in\tilde\scP(P_{XY})$ and $P_{\hatX|UYT}$.
By the first part of the proof, we can assume that $\calU=\calU'$ and
$\lvert\calU\rvert=\lvert\calU'\rvert\leq\lvert\calX\rvert+8$.
Let us consider the following 9 functions on $\scP(\calU\times\calX\times\calY\times\hatcalX)$:
\begin{align}
 F_1(P_{UXY\hatX})&:=I(Y;U)\\
 F_2(P_{UXY\hatX})&:=I(X;U)\\
 F_3(P_{UXY\hatX})&:=\bbE[\rvd(X,\hatX)]\\
 F_4(P_{UXY\hatX})&:=\Var\left(-\log\frac{P_{Y|U}(Y|U)}{P_Y(Y)}\right) \\
 F_5(P_{UXY\hatX})&:=\Var\left(\log\frac{P_{X|U}(X|U)}{P_Y(X)}\right) \\
 F_6(P_{UXY\hatX})&:=\Var\left(\rvd(X,\hatX)\right) \\
 F_7(P_{UXY\hatX})&:=\Cov\left(-\log\frac{P_{Y|U}(Y|U)}{P_Y(Y)}, \log\frac{P_{X|U}(X|U)}{P_X(X)}\right) \\
 F_8(P_{UXY\hatX})&:=\Cov\left(-\log\frac{P_{Y|U}(Y|U)}{P_Y(Y)}, \rvd(X,\hatX)\right) \\
 F_9(P_{UXY\hatX})&:=\Cov\left(\log\frac{P_{X|U}(X|U)}{P_X(X)}, \rvd(X,\hatX)\right)
\end{align}
and a function
$F\colon\scP(\calU|\calX)\times\scP(\hatcalX|\calU\times\calY)\to\scP(\calU\times\calX\times\calY\times\hatcalX)$
such as $P_{UXY\hatX}=F(P_{U|X},P_{\hatX|UY})$ satisfies
\begin{align}
 P_{UXY\hatX}(u,x,y,\hatx)=P_{XY}(x,y)P_{U|X}(u|x)P_{\hatX|YU}(\hatx|y,u).
\end{align}
Then, by applying Lemma \ref{support_lemma_T} to
$f_j(\cdot):=F_j(F(\cdot))$ ($j=1,2,\dots,9$), we have
 $P_{T'}\in\scP(\calT')$ with $\lvert\calT'\rvert\leq 9$ and 
$\{(P_{U'|X'T'}(\cdot|\cdot,t'),P_{\hatX'|U'Y'T'}(\cdot|\cdot,\cdot,t')):t'\in\calT'\}\subset\scP(\calU|\calX)\times\scP(\hatcalX|\calU\times\calY)$
 satisfying \eqref{eq:support_lemma_T}.
By 
$P_{T'}$, $(P_{U'|X'T'}, P_{\hatX'|U'Y'T'})$ and $P_{XY}$, 
let us define $P_{U'T'X'Y'\hatX'}=P_{U'T'XY\hatX'}$ as 
\begin{align}
& P_{U'T'XY\hatX'}(u',t',x,y,\hatx') \nn\\
&= P_{XY}(x,y)P_{T'}(t)P_{U'|X'T'}(u'|x,t')P_{\hatX'|U'Y'T'}(\hatx'|u',y,t').
\end{align}
We can verify that the pair $(P_{U'T'XY},P_{\hatX'|U'YT'})$ derived from
 $P_{U'T'XY\hatX'}$ satisfies the conditions of the theorem.
\end{proof}

\subsection{Proof Sketches of Cardinality Bounds for WAK and GP problems}\label{app:proof_cardinality_WAK-GP}

\begin{proof}
[Proof of Theorem \ref{thm:cardinality_bound_WAK}]
We fix $t\in\calT$ and then consider the 
following $\lvert\calY\lvert+4$ quantities:
$\lvert\calY\rvert-1$ elements $P_Y(y)$
($y=1,2,\dots,\lvert\calY\rvert-1$) of $P_Y$, the conditional entropy
$H(X|U,T=t)$, the mutual information $I(U;Y|T=t)$, two variances on the
diagonals of $\cov(\bj(U,X,Y|t))$, and the covariance in the upper part
of $\cov(\bj(U,X,Y|t))$.  
Then, in the same manner as the first part of the proof for Wyner-Ziv problem, 
we
can choose a random variable $U'\sim P_{U'|T=t}\in\scP(\calU)$ with
$\lvert\calU'\rvert\leq\lvert\calY\lvert+4$ which preserves the marginal
distribution $P_{XY|T=t}$, $\bbE[\bj(U,X,Y|t)]$, and
$\cov(\bj(U,X,Y|t))$.  By taking the average with respect to $T$, we can
show that $U'$ satisfies the conditions of the theorem.
Further, in the same way as the second part of the proof for Wyner-Ziv problem,
 we can show that $T'$ with $\lvert\calT'\rvert\leq 5$ 
preserves the following five quantities: 
two elements of $\bJ$, two variances along the diagonals of $\bV$, and the
 covariance in the upper part of $\bV$.
\end{proof}

\begin{proof}
[Proof of Theorem \ref{thm:cardinality_bound_GP}]
We fix $t\in\calT$ and then consider the 
following 
$\lvert\calS\rvert\lvert\calX\rvert+6$ quantities:
$\lvert\calS\rvert\lvert\calX\rvert-1$ elements 
$P_{SX}(s,x)$ of $P_{SX}$, two mutual informations
$I(U;Y|t)$, $I(U;S|t)$, two variances
$\Var(\log P_{Y|UT}(Y|U,t)/P_{Y|T}(Y|t))$, $\Var(-\log P_{S|UT}(S|U,t)/P_{S}(S))$, and three covariances in the strict
 upper triangular part of $\cov(\bj(U,S,X,Y|t))$.
Note that, if the marginal distribution $P_{SXY|T=t}$ is preserved then 
the average
$\bbE[\rvg(X_T)|T=t]$
and the variance
$\Var(\rvg(X_T)|T=t)$
of $\rvg(X_T)$ with respect to the distribution $P_{X|T=t}$ is automatically preserved.
Hence, in the same manner as the first part of the proof for Wyner-Ziv problem, 
we
can choose a random variable $U'\sim P_{U'|T=t}\in\scP(\calU)$ with
$\lvert\calU'\rvert\leq\lvert\calS\lvert\lvert\calX\rvert+6$ which preserves the marginal
distribution $P_{SX|T=t}$, $\bbE[\bj(U,S,X,Y|t)]$, and
$\cov(\bj(U,S,X,Y|t))$.  By taking the average with respect to $T$, we can
show that $U'$ satisfies the conditions of the theorem.
Further, in the same way as the second part of the proof for Wyner-Ziv problem,
 we can show that $T'$ with $\lvert\calT'\rvert\leq 5$ 
preserves the following nine quantities: 
three elements of $\bJ$, three variances along the diagonals of $\bV$, and
 three covariances in the strict upper triangular part of $\bV$.
\end{proof}

\subsection*{Acknowledgements} 

The authors would like to thank J.~Scarlett for pointing out an error 
of the numerical calculation of the GP problem in an earlier version of the paper.
The authors also appreciate anonymous reviewers for valuable comments, 
in particular for pointing out Remark \ref{remark:channel-simulation}.
The work of first author is supported in part by JSPS Postdoctoral Fellowships for Research Abroad. 
The work of third author is supported in part by NUS startup grant R-263-000-A98-750/133 and in part by A*STAR, Singapore.

\bibliographystyle{IEEETran}
\bibliography{./isitbib.bib}


\begin{IEEEbiographynophoto}{Shun Watanabe}
(M'09) received the B.E., 
M.E., and Ph.D.\ degrees from the Tokyo Institute of Technology
in 2005, 2007, and 2009, respectively. Since April 2009, he has been
an Assistant Professor in the Department of Information 
Science and Intelligent Systems at  the University of Tokushima.
Since April 2013, he has also been a visiting Assistant Professor
in the Institute for Systems Research at the University of Maryland, College Park.
His current research interests are in the areas of
information theory, quantum information theory,
and quantum cryptography.
\end{IEEEbiographynophoto}

\begin{IEEEbiographynophoto}{Shigeaki Kuzuoka}
(S'05-M'07) received the B.E., M.E., and Ph.D.\ degrees from Tokyo
Institute of Technology in 2002, 2004, and 2007 respectively.  He was an
assistant professor from 2007 to 2009, and has been a lecturer since
2009 in the Department of Computer and Communication Sciences, Wakayama
University.  His current research interests are in the areas of
information theory, especially Shannon theory, source coding, and
multi-terminal information theory.
\end{IEEEbiographynophoto}

\begin{IEEEbiographynophoto}{Vincent Y. F. Tan} (S'07-M'11)  is an Assistant Professor   in the Department of Electrical and Computer Engineering (ECE) and the Department of  Mathematics at the National University of Singapore (NUS).  He received the B.A.\ and M.Eng.\ degrees in Electrical and Information Sciences from  Cambridge University in 2005. He received the Ph.D.\ degree in Electrical Engineering and Computer Science (EECS) from the Massachusetts Institute of Technology in 2011. He was   a postdoctoral researcher in the Department of  ECE  at the University of Wisconsin-Madison and following that, a research scientist at the Institute for Infocomm (I$^2$R) Research,  A*STAR, Singapore. His research interests include information theory, machine learning and signal processing.

Dr.\ Tan   received the  MIT EECS Jin-Au Kong outstanding doctoral thesis prize in 2011 and the NUS Young Investigator Award in 2014.  He has authored a research monograph on {\em Asymptotic Estimates in Information Theory with Non-Vanishing Error Probabilities} in the Foundations and Trends\textsuperscript{\textregistered} in Communications and Information Theory Series (NOW Publishers).
\end{IEEEbiographynophoto}

\end{document}